\documentclass[letterpaper]{article} 
\usepackage{aaai24}  
\usepackage{times}  
\usepackage{helvet}  
\usepackage{courier}  
\usepackage[hyphens]{url}  
\usepackage{graphicx} 
\usepackage{natbib}  
\usepackage{caption} 
\usepackage{algorithm}
\usepackage{algorithmic}

\usepackage{newfloat}
\usepackage{listings}
\DeclareCaptionStyle{ruled}{labelfont=normalfont,labelsep=colon,strut=off} 
\floatstyle{ruled}
\newfloat{listing}{tb}{lst}{}
\floatname{listing}{Listing}
\usepackage{amsfonts}       
\usepackage{amsmath}       
\usepackage{amsthm}
\usepackage{mathtools}
\usepackage{centernot}
\usepackage{subfigure}

\DeclareMathOperator*{\argmax}{arg\,max}
\DeclareMathOperator*{\argmin}{arg\,min}

\newcommand{\E}{\operatorname{\mathbb E}}

\newtheorem{lemma}{Lemma}

\newtheorem{remark}{Remark}
\newtheorem{corollary}{Corollary}
\newtheorem{theorem}{Theorem}

\newtheorem{definition}{Definition}
\newtheorem{assumption}{Assumption}

\newtheorem{property}{Property}

\allowdisplaybreaks
\usepackage[shortlabels]{enumitem}

\title{Learning Discrete-Time Major-Minor Mean Field Games}
\author {
    Kai Cui\equalcontrib\textsuperscript{\rm 1},
    Gökçe Dayanıklı\equalcontrib\textsuperscript{\rm 2},
    Mathieu Laurière\textsuperscript{\rm 3},\\
    Matthieu Geist\textsuperscript{\rm 4},
    Olivier Pietquin\textsuperscript{\rm 5},
    Heinz Koeppl\textsuperscript{\rm 1}
}
\affiliations {
    \textsuperscript{\rm 1}Technische Universität Darmstadt,\\
    \textsuperscript{\rm 2}University of Illinois at Urbana-Champaign,\\
    \textsuperscript{\rm 3}NYU Shanghai,\\
    \textsuperscript{\rm 4}Google DeepMind,\\
    \textsuperscript{\rm 5}Cohere\\
    kai.cui@tu-darmstadt.de, gokced@illinois.edu, mathieu.lauriere@nyu.edu, \\mfgeist@google.com, olivier@cohere.com, heinz.koeppl@tu-darmstadt.de
}

\begin{document}

\maketitle

\begin{abstract}
Recent techniques based on Mean Field Games (MFGs) allow the scalable analysis of multi-player games with many similar, rational agents. However, standard MFGs remain limited to homogeneous players that weakly influence each other, and cannot model major players that strongly influence other players, severely limiting the class of problems that can be handled. We propose a novel discrete time version of major-minor MFGs (M3FGs), along with a learning algorithm based on fictitious play and partitioning the probability simplex. Importantly, M3FGs generalize MFGs with common noise and can handle not only random exogeneous environment states but also major players. A key challenge is that the mean field is stochastic and not deterministic as in standard MFGs. Our theoretical investigation verifies both the M3FG model and its algorithmic solution, showing firstly the well-posedness of the M3FG model starting from a finite game of interest, and secondly convergence and approximation guarantees of the fictitious play algorithm. Then, we empirically verify the obtained theoretical results, ablating some of the theoretical assumptions made, and show successful equilibrium learning in three example problems. Overall, we establish a learning framework for a novel and broad class of tractable games.
\looseness=-1
\end{abstract}

\section{Introduction}
While reinforcement learning (RL) has achieved tremendous recent success \citep{mnih2015human, sutton2018reinforcement}, multi-agent RL (MARL) as its game-theoretic counterpart remains difficult due to its many challenges \citep{zhang2021multi}. In particular, the scalability challenge is hard to overcome due to the notorious complexity of non-cooperative stochastic games \citep{daskalakis2009complexity, yang2020overview}. Here, the recent introduction of mean field games (MFGs, \citep{lasry2007mean, huang2006large, saldi2018markov}) has contributed a mathematically rigorous and tractable approach to handling large-scale games, finding application in a variety of domains such as finance \citep{carmona2020applications} and engineering \citep{djehiche2017mean}. The general idea is to summarize many similar agents (players) and their interaction through their state distribution -- the mean field (MF). Owing to the amenable complexity of MFGs, many recent efforts have formulated equilibrium learning algorithms for MFGs \citep{lauriere2022learning}, including approaches based on regularization \citep{cui2021approximately, lauriere2022scalable, guo2022entropy}, optimization \citep{guo2022mf, guo2023mesob}, fictitious play \citep{perrin2020fictitious, geist2022concave} and online mirror descent \citep{perolat2022scaling, lauriere2022scalable, yardim2023policy}. For less-familiar readers, we refer to the survey of \citet{lauriere2022learning}.
\looseness=-1

So far however, most MFG learning frameworks remain unable to handle common noise \citep{carmona2016mean}, or more generally major players. Contrary to minor players, a major player directly affects all minor players and is affected by the MF of minor players, whereas common noise also affects all minor players, but is exogeneous and can be understood as a static major player without actions \citep{huang2016dynamic}. Notably, \citet{perrin2020fictitious} formulate an algorithm handling common noise using a continuous learning Lyapunov argument \citep{harris1998rate, hofbauer2002global}, assuming however that the common noise is known, while \citet{cui2023multi} consider a cooperative setting. Common noise and major players remain important in practice, as a system seldom consists only of many similar minor players. For example, strategic players on the market do not exist in a vacuum but must contend for instance with idiosyncratic shocks \citep{carmona2020applications} or government regulators \citep{aurell2022optimal}, while many cars on a road network \citep{cabannes2021solving} may be subject to traffic accidents or traffic lights. 
In continuous-time, such systems are known as MFGs with major and minor players \citep{carmona2018probabilistic}, and have been considered, e.g., by \citet{huang2010large, nguyen2012linear, bensoussan2016mean} for LQG systems, by \cite{nourian2013mm, sen2016mean} in non-linear and partially observed settings, and more recently by \citet{carmona2016probabilistic, carmona2017alternative, lasry2018mean, cardaliaguet2020remarks}. Major agents also generalize common noise, an important problem in MFG literature \citep{carmona2016mean, perrin2020fictitious, motte2022mean}. For an additional overview, we also point to \citet{carmona2018probabilistic}. In contrast to prior work, we focus on a computational learning framework that is in discrete time. Additionally, even existing discrete-time MFG frameworks with only common noise such as by \citet{perrin2020fictitious} have to the best of our knowledge not yet rigorously connected MFGs with the finite games of practical interest.
We note that another setting with major players has already been explored: Stackelberg MFGs. \cite{elie2019tale,carmona2021finite,carmona2022mean} consider a Stackelberg equilibrium instead of a Nash equilibrium, wherein a `major' principal player chooses their policy first and has priority (like a government or regulator); see \citep{guo2022optimization, vasal2022master} for discrete time versions of the problem. Though the Stackelberg setting is of importance, it is distinct from computing Nash equilibria where major and minor players are ``on the same level'': in the Stackelberg setting, minor players only respond with a Nash equilibrium between themselves \textit{after} the principal's policy choice. Furthermore, we are not aware of any propagation of chaos results even in discrete-time Stackelberg MFGs, for which our result also applies. The field of Stackelberg MFGs remains part of continued active research, to which our M3FG setting may also contribute, and vice versa.

By the preceding motivation, we propose the first general discrete-time Major-Minor MFG (M3FG) learning framework. We begin with providing a theoretical foundation of the proposed M3FG model, showing that equilibria in finite games with many players can be approximately learned in the M3FG instead. The proof is based upon showing propagation of chaos i.e., convergence of the empirical MF, which -- in contrast to its counterpart in MFGs without common noise -- converges only in distribution. We then move on to provide a learning algorithm based on fictitious play to solve M3FGs, with convergence results and approximation guarantees for its tractable and practical, tabular variant. Empirically, our learned policies do not assume that common noise is known a priori. Due to the resulting stochastic MF, for tractable dynamic programming we allow conditioning of player actions and policies also on the MF instead of just the player's own state. Finally, we verify the M3FG framework on three problems, empirically supporting theoretical claims, even when the assumptions are not entirely fulfilled.

\section{Major-Minor Mean Field Games} \label{sec:m3fg}
In this section, we begin by giving a description of considered problems and their corresponding mean field system.

\textit{Notation: Equip finite sets $\mathcal S$ with the discrete metric, products with the product $\sup$ metric, and probability measures $\mathcal P(\mathcal S)$ on $\mathcal S$  with the $L_1$ norm. Let $[N] \coloneqq \{ 1, \ldots, N \}$.}

\subsection{Finite Player Game}
We consider a game with $N$ minor players and one major player. 
Let $\mathcal X$ and $\mathcal U$ be finite state and action spaces for minor players, respectively. Let $\mathcal X^0$ and $\mathcal U^0$ be finite state and action spaces for the major player, respectively. Let $T \in \mathbb N$ be a finite time horizon and let $\mathcal T \coloneqq \{0, 1, \ldots, T-1\}$. We denote the state and the action of minor player $i \in [N]$ at time $t \in \mathcal T$ by $x^{i,N}_t$ and $u^{i,N}_t$, respectively. Similarly, we denote by $x^{0,N}_t$ and $u^{0,N}_t$ the state and the action of the major player at time $t$. Let $\mu_0$ and $\mu_0^0$ be initial probability distributions on $\mathcal X$ and $\mathcal X^0$, respectively. Define the empirical MF $\mu_t^N \coloneqq \frac{1}{N} \sum_{i=1}^N \mathbf 1_{x_t^{i,N}}$, where $\mathbf 1_x$ is the indicator function equal to $1$ for the argument $x$ and $0$ otherwise. The MF can be viewed as a histogram with $|\mathcal X|$ many bins.

We can consider several classes of policies. In this presentation, we focus on Markovian feedback policies in the following sense: the policy $\pi^{i,N}$ for minor player $i$ is a function of her own state, the major player's state and the MF; the policy $\pi^{0,N}$ for the major player is a function of her own state and the MF. We denote respectively by $\Pi$ and $\Pi_0$ the sets of such minor and major player policies. 

For a given tuple of policies $(\underline{\pi}^N, \pi^{0,N}) = ((\pi^{1,N}, \ldots, \pi^{N,N}), \pi^{0,N}) \in \Pi^N \times \Pi_0$, the game begins with states $x^{0,N}_0 \sim \mu^0_0$, $x^{i,N}_0 \sim \mu_0$ and subsequently, for $t=0,1,\dots,T-2$, let
\small
\begin{subequations} \label{eq:finite}
\begin{align}
    u^{i,N}_t &\sim \pi^{i,N}_t(x^{i,N}_t, x^{0,N}_t, \mu_t^N), \quad i \in [N]\\
    u^{0,N}_t &\sim \pi^{0,N}_t(x^{0,N}_t, \mu_t^N), \\
    x^{i,N}_{t+1} &\sim P(x^{i,N}_t, u^{i,N}_t, x^{0,N}_t, u^{0,N}_t, \mu_t^N), \quad i \in [N] \\
    x^{0,N}_{t+1} &\sim P^0(x^{0,N}_t, u^{0,N}_t, \mu_t^N).
\end{align}
\end{subequations}
\normalsize
where $P \colon \mathcal{X} \times \mathcal{U} \times \mathcal{X}^0 \times \mathcal{U}^0 \times \mathcal{P}(\mathcal{X}) \to \mathcal{P}(\mathcal{X})$ and ${P^0 \colon \mathcal{X}^0 \times \mathcal{U}^0 \times \mathcal{P}(\mathcal{X}) \to \mathcal{P}(\mathcal{X})}$ are transition kernels.

In contrast to classic MFGs such as studied e.g, in \citep{saldi2018markov}, the minor players' dynamics depend also on the major player's state. An important consequence is that the minor players' dynamics are influenced by a form of common noise. This explains why we decide to consider policies that depend on the MF $\mu_t^N$. Furthermore, this form of common noise is not simply an exogenous source of randomness because it is influenced by the major player's choice of policy. This makes the problem more challenging than MFGs with common noise.

Next, we define the minor and major total rewards
\small
\begin{align*}
    J_N^i(\underline{\pi}^N, \pi^{0,N})
    &= \E \left[ \sum_{t \in \mathcal T} r(x^{i,N}_t, u^{i,N}_t, x^{0,N}_t, u^{0,N}_t, \mu^N_t) \right], \\
    J_N^0(\underline{\pi}^N, \pi^{0,N})
    &= \E \left[ \sum_{t \in \mathcal T} r^0(x^{0,N}_t, u^{0,N}_t, \mu^N_t) \right],
\end{align*}
\normalsize
for some reward functions $r \colon \mathcal{X} \times \mathcal{U} \times \mathcal{X}^0 \times \mathcal{U}^0 \times \mathcal{P}(\mathcal{X}) \to \mathbb{R}$ and $r^0 \colon \mathcal{X}^0 \times \mathcal{U}^0 \times \mathcal{P}(\mathcal{X}) \to \mathbb{R}$.

In this work, we focus on the non-cooperative scenario where players try to maximize their own objectives while anticipating the behavior of other players. This is formalized by the solution concept of (approximate) Nash equilibria.

\begin{definition}
    Let $\varepsilon \ge 0$. An approximate $\varepsilon$-Nash equilibrium is a tuple $(\underline{\pi}^N, \pi^{0,N}) = ((\pi^{1,N}, \ldots, \pi^{N,N}), \pi^{0,N}) \in \Pi^N \times \Pi_0$ of policies, such that $J_N^0(\underline\pi^N, \pi^{0,N}) \geq \sup_{\tilde\pi^0} J_N^0((\pi^{1,N}, \ldots, \pi^{N,N}), \tilde\pi^0) - \varepsilon$ and $J_N^i(\underline\pi^N, \pi^{0,N}) \geq \sup_{\tilde\pi^i \in \Pi} J_N^i((\pi^1, \ldots, \pi^{i-1}, \tilde\pi^i, \pi^{i+1}, \ldots, \pi^N), \pi^{0,N}) - \varepsilon$ for all $i \in [N]$. A Nash equilibrium is an approximate $0$-Nash equilibrium.
\end{definition}

\begin{remark}
    We can also consider time-dependent dynamics or rewards, multiple major players, and infinite-horizon discounted objectives. Some results we prove below can be extended to such settings (e.g., propagation of chaos, equilibrium approximation, and fictitious play; see also generalized infinite-horizon experiments in Appendix~\ref{app:exp}). Similarly, we can extend the model to multiple minor agent populations with small changes, see e.g. \citet{perolat2022scaling}. Another possibility is to simply include types of players into their state \citep{mondal2022approximation}.
\end{remark}

\subsection{Mean Field Game}
When the number of minor players $N$ is large, we can approximate the game by an MFG, which corresponds formally to the limit $N\to\infty$. In an MFG, the empirical MF is replaced by a random limiting MF. Unlike standard MFGs, the limiting MF does not evolve in a deterministic way due to the influence of the major player.
Fixing major and minor player policies $\pi^0$, $\pi$ for all players, except for a single minor player deviating to $\hat \pi$, when $N \to \infty$, we obtain (intuitively by a law of large numbers argument) the major and deviating minor player M3FG dynamics $x^0_0 \sim \mu^0_0$, $x_0 \sim \mu_0$,
\small
\begin{subequations}
\begin{align}
    u_t &\sim \hat \pi_t(x_t, x^0_t, \mu_t), \\
    u^0_t &\sim \pi^0_t(x^0_t, \mu_t), \\
    x_{t+1} &\sim P(x_t, u_t, x^0_t, u^0_t, \mu_t), \\
    x^0_{t+1} &\sim P^0(x^0_t, u^0_t, \mu_t), \\
    \mu_{t+1} &= T^\pi_t(x^0_t, u^0_t, \mu_t)
\end{align} \label{eq:m3fg-min}
\end{subequations}
\normalsize
with the deterministic transition operator $T^\pi_t(x^0, u^0, \mu) \coloneqq \iint P(x, u, x^0, u^0, \mu) \, \pi_t(\mathrm du \mid x, x^0, \mu) \, \mu(\mathrm dx)$ as the conditional ``expectation'' of the next MF given the current major state $x^0$, action $u^0$, and random MF $\mu$. The policy $\pi$ is shared by all minor players except one who is deviating and using $\hat \pi$. This means that we look for symmetric Nash equilibria where all exchangeable minor players use the same policy, as usual in MFG literature. Still, a mean field equilibrium suffices as an approximate Nash equilibrium in the finite game, which is not to say that there cannot be other heterogeneous policy tuples in the finite game that are Nash.

M3FGs now consist of \textit{two} Markov decision process (MDP) optimality conditions, one for all minor players and one for the major player. An equilibrium is then optimal in each MDP simultaneously.
More precisely, from the point of view of a minor player, the goal is to optimize over $\hat \pi$ while $(\pi, \pi^0)$ are fixed. This yields the minor player MDP with state $(x_t, x^0_t, \mu_t) \in \mathcal X \times \mathcal X^0 \times \mathcal P(\mathcal X)$, and action $u_t \in \mathcal U$, and with the objective
\small
\begin{align}
    J(\hat \pi, \pi, \pi^0) &= \E \left[ \sum_{t \in \mathcal T} r(x_t, u_t, x^0_t, u^0_t, \mu_t) \right].
\end{align}
\normalsize
Note that, although $\mu_{t+1}$ is given by a deterministic function of  $(x^0_t, u^0_t, \mu_t)$, from the point of view of a minor player, the evolution of $(\mu_{t})_{t}$ is stochastic since it depends on the sequence $(x^0_t, u^0_t)_{t}$, which is random. By definition of a Nash equilibrium, only a \textit{single} minor player can deviate arbitrarily to $\hat \pi$, and by symmetry it does not matter which ``representative'' minor player deviates. Therefore there is only one MDP optimality condition for all minor players. We also stress that since $N\rightarrow\infty$, the representative player is insignificant and her deviation does not affect the mean field.

On a similar note, from the major player's point of view, we obtain the major player MDP with $(\mathcal X^0 \times \mathcal P(\mathcal X))$-valued states $(x^0_t, \mu_t)$ and $\mathcal U^0$-valued actions $u^0_t$ of the major player, using the same dynamics, forgetting about the (insignificant for the major player) deviating minor player, and optimizing instead for $\pi^0$, the corresponding major objective
\small
\begin{align}
    J^0(\pi, \pi^0) &= \E \left[ \sum_{t \in \mathcal T} r^0(x^0_t, u^0_t, \mu_t) \right].
\end{align}
\normalsize

\subsubsection{Mean field equilibrium.}
The Nash equilibrium in the finite game hence corresponds to a major-minor mean field equilibrium, as a fixed point of both MDPs \textit{at once}. In other words, major and minor policies $\pi^0$, $\pi$ that are optimal against themselves in the major and minor player MDPs.

\begin{definition}
    A Major-Minor Mean Field Nash Equilibrium (M3FNE) is a tuple $(\pi, \pi^0) \in \Pi \times \Pi_0$ of policies, such that $\pi \in \argmax J(\cdot, \pi, \pi^0)$ and $\pi^0 \in \argmax J^0(\pi, \cdot)$.
\end{definition}

We slightly weaken the concept of optimality to \textit{approximate} optimality, since the solution of a limiting MFG provides approximate Nash equilibria for the finite game, which are still achieved by solving for approximate M3FNE.

\begin{definition}
    An approximate $\varepsilon$-M3FNE is a tuple $(\pi, \pi^0) \in \Pi \times \Pi_0$ of policies, such that $J(\pi, \pi, \pi^0) \geq \sup J(\cdot, \pi, \pi^0) - \varepsilon$ and $J^0(\pi, \pi^0) \geq \sup J^0(\pi, \cdot) - \varepsilon$.
\end{definition}

The minimal such $\varepsilon$ for minor and major agents are also referred to as the minor and major exploitabilities $\mathcal E(\pi, \pi^0)$ and $\mathcal E^0(\pi, \pi^0)$ of $(\pi, \pi^0)$. Accordingly, an exploitability of $0$ means that $(\pi, \pi^0)$ is an exact M3FNE.

\section{Theoretical Analysis}
The M3FG is a theoretically rigorous formulation for large corresponding finite games. Note in particular that the MF will be stochastic due to the randomness of major players and their states, and therefore standard results based on determinism of MFs will no longer hold. We provide such a theoretical foundation of M3FG by propagation of chaos.

\paragraph{Continuity assumptions.}
We provide theoretical guarantees to prove that the M3FNE is an approximate Nash equilibrium in the finite game, despite having a non-deterministic MF in the limiting case, contrary to most of the existing literature \citep{huang2006distributed, guo2019learning}. For this, we need some common Lipschitz continuity assumptions \citep{gu2021mean, pasztor2023efficient}.

\begin{assumption} \label{ass:m3pcont}
The kernels $P$, $P^0$ are $L_P$, $L_{P^0}$-Lipschitz.
\end{assumption}
\begin{assumption} \label{ass:m3rcont}
The rewards $r$, $r^0$ are $L_r$, $L_{r^0}$-Lipschitz.
\end{assumption}
\begin{assumption} \label{ass:m3picont}
The classes of major and minor policies $\Pi^0$, $\Pi$ are equi-Lipschitz, i.e. there are $L_{\Pi^0}, L_\Pi > 0$ s.t. for all $t$, $\pi^0 \in \Pi^0$, $\pi \in \Pi$, we have that $\pi^0_t \colon \mathcal X^0 \times \mathcal P(\mathcal X) \to \mathcal P(\mathcal U)$ and $\pi_t \colon \mathcal X \times \mathcal X^0 \times \mathcal P(\mathcal X) \to \mathcal P(\mathcal U)$ are $L_{\Pi^0}$, $L_\Pi$-Lipschitz.
\end{assumption}

Here, we always consider Lipschitz continuity for all arguments using the $\sup$ metric for products, and the $L_1$ distance for probability measures, see e.g., Appendix~\ref{app:lem:Tcont}. We note that the Lipschitz assumption for policies -- while standard -- is technical. Empirically, the only piecewise Lipschitz policies obtained in Section~\ref{sec:algo} for tractability nonetheless remain close to the following approximations in the finite system. A theoretical investigation of guarantees for piecewise Lipschitz policies is left for future work.

\paragraph{Propagation of chaos.}
We achieve propagation of chaos ``in distribution'' for major and minor players to the M3FG at rate $\mathcal O(1 / \sqrt N)$, which is shown inductively in Appendix~\ref{app:m3muconv}. Here, propagation of chaos refers to the conditional independence of minor agents, and thus convergence in the limit to the deterministic mean field \citep{chaintron2022propagation}. In contrast to MFGs with deterministic MFs, a stronger mode of convergence such as the one considered by \citet{saldi2018markov} fails by stochasticity of the MF.  

\begin{theorem} \label{thm:m3muconv}
Consider Assumptions~\ref{ass:m3pcont} and \ref{ass:m3picont}, and any equi-Lipschitz family of functions $\mathcal F \subseteq \mathbb R^{\mathcal X \times \mathcal U \times \mathcal X^0 \times \mathcal U^0 \times \mathcal P(\mathcal X)}$ with shared Lipschitz constant $L_{\mathcal F}$. Then, the random variable $(x^{1,N}_t, u^{1,N}_{t}, x^{0,N}_t, u^{0,N}_{t}, \mu_t^N)$ in system \eqref{eq:finite} under $((\hat \pi, \pi, \pi, \ldots), \pi^0)$ converges weakly, uniformly over $f \in \mathcal F$ and $(\hat \pi, \pi, \pi^0) \in \Pi \times \Pi \times \Pi^0$, to $(x_t, u_{t}, x^0_t, u^0_{t}, \mu_t)$ in system \eqref{eq:m3fg-min} under $(\hat \pi, \pi, \pi^0)$,
\small
\begin{align} \label{eq:m3muconv-min}
    &\forall t \in \mathcal T,\ \sup_{\hat \pi, \pi, \pi^0} \sup_{f \in \mathcal F} \left| \E \left[ f(x^{1,N}_t, u^{1,N}_t, x^{0,N}_t, u^{0,N}_t, \mu^N_t) \right]
    \right.\nonumber\\&\hspace{1.9cm}\left.
    - \E \left[ f(x_t, u_t, x^{0}_t, u^0_t, \mu_t) \right] \right| = \mathcal O(1/\sqrt{N}).
\end{align}
\normalsize
\end{theorem}

\begin{corollary} \label{coro:m3muconv}
Similarly, consider Assumptions~\ref{ass:m3pcont} and \ref{ass:m3picont}, and any family of equi-Lipschitz functions $\mathcal F^0 \subseteq \mathbb R^{\mathcal X^0 \times \mathcal U^0 \times \mathcal P(\mathcal X)}$ with shared Lipschitz constant $L_{\mathcal F^0}$. Then the random variable $(x^{0,N}_t, u^{0,N}_{t}, \mu_t^N)$ in system \eqref{eq:finite} under $((\hat \pi, \pi, \pi, \ldots), \pi^0)$ converges weakly, uniformly over $f \in \mathcal F^0$, to $(x^0_t, u^0_{t}, \mu_t)$ in system \eqref{eq:m3fg-min} under $(\hat \pi, \pi, \pi^0)$,
\small
\begin{align} \label{eq:m3muconv-maj}
    \forall t \in \mathcal T,\ &\sup_{\hat \pi, \pi, \pi^0} \sup_{f \in \mathcal F^0} \left| \E \left[ f(x^{0,N}_t, u^{0,N}_t, \mu^N_t) \right]
    \right.\nonumber\\&\hspace{1.7cm}\left.
    - \E \left[ f(x^{0}_t, u^0_t, \mu_t) \right] \right|= \mathcal O(1/\sqrt{N}).
\end{align}
\normalsize
\end{corollary}

\paragraph{Approximate Nash equilibrium.}
From propagation of chaos, the approximate Nash property of M3FNE follows, suggesting that a solution of M3FGs provides a good game-theoretic solution of interest to practical $N$-player games, see Appendix~\ref{app:varepsNash} for the proof based on propagation of chaos.

\begin{corollary} \label{coro:varepsNash}
Consider Assumptions~\ref{ass:m3pcont}, \ref{ass:m3rcont}, \ref{ass:m3picont}, and a M3FNE $(\pi, \pi^0) \in \Pi \times \Pi^0$. 
Then, the policies $((\pi, \ldots, \pi), \pi^0)$ constitute an $\mathcal O(1/\sqrt N)$-Nash equilibrium in the finite game.
\end{corollary}

Finally, existence of a M3FNE is a difficult question under policies that depend on the stochastic MF. While assuming reactive policies unconditioned on the MF could help, choosing such policies makes the design of our algorithm based on dynamic programming difficult, as policies computed via dynamic programming need to depend on the entire M3FG system state. In contrast, in usual deterministic MFGs it is sufficient to remove policy dependence on the MF, which is deterministic. For practical purposes, learning equilibria and then checking the exploitability by Theorem~\ref{thm:disc-opt} may suffice.

\section{Fictitious Play} \label{sec:algo}
To find M3FNE and solve the fixed-point problem, we formulate a fictitious play (FP) algorithm and provide a theoretical analysis. Following the exact algorithm, as empirical contribution we provide and analyze an approximate, numerically tractable algorithm that does not assume knowledge of common noise, contrary to \citet{perrin2020fictitious}, and extend it to the setup with major and minor players. Since the space of MFs is continuous and does not allow general exact computation of value functions, we project MFs onto a finite partition with guarantees for policy evaluation.

\subsection{Fictitious Play for M3FNE}
In order to learn an M3FNE, we first propose an exact analytic algorithm based on FP \citep{perrin2020fictitious} and provide a theoretical analysis of convergence. 
For this part, we will assume that the major player's action does not affect the minor players' transition kernel. To simplify the presentation and the analysis, we will use conditioning with respect to the sources of randomness that affect the MF, i.e., the minors' distribution. 
For every $t \ge 0$, let the major and minor players' actions be determined not by the mean field $\mu_t$, but instead by the history of major states and actions, $u^0_t \sim \pi^0_t(x^0_t, x^0_{0:t-1}, u^0_{0:t-1})$, $x^0_{0:t-1} \coloneqq (x^0_0, x^0_1, \ldots, x^0_{t-1})$, $u^0_{0:t-1} \coloneqq (u^0_0, u^0_1, \ldots, u^0_{t-1})$. By induction, we can in fact view $\mu_{t}$ as a deterministic function of $(x^0_{0:t-1}, u^0_{0:t-1})$ given the minor players' policy $\pi$, since we simply have $\mu_{t+1} = T^\pi_t(x^0_t, u^0_t, \mu_t)$ recursively and deterministically.
This means that for fixed policies such as a given Nash equilibrium, any policies dependent on $\mu_{t}$ can instead be rewritten as functions of $(x^0_{0:t-1}, u^0_{0:t-1})$.
Therefore, instead of seeing policies as functions of $\mu_t$, we will see them as functions of the major player randomness $(x^0_{0:t-1}, u^0_{0:t-1})$ and we will write (slightly abusing notation) $\pi_t(x_t, x^0_{0:t}, u^0_{0:t-1})$ and $\pi^0_t(x_t^0, x^0_{0:t-1}, u^0_{0:t-1})$
respectively for the minor players' and the major player's policies.
The results we prove below go beyond existing results by (i) analyzing also the major exploitability similarly to the minor exploitability, and (ii) expanding analysis of minor exploitability under presence of major players. To this end, we formulate Assumption~\ref{as:fictitious_play_conv}.\ref{as:fic_play_conv_explo_2} and \ref{as:fictitious_play_conv}.\ref{as:fic_play_conv_explo}, which provide the conditions for convergence in the presence of major players.
See Appendix~\ref{app:fic_play} for more detail.

We first start by introducing the discrete time FP before analyzing it in continuous time. Here, time refers to the algorithm's current iteration and not to the time of the M3FG system, which remains discrete throughout the whole paper. At any given step $j$ of FP, we have: 
\begin{equation}
\label{eq:fict_dist_avg}
    \mu^{\bar{\pi}^j}_{t \mid x^0_{0:t-1}, u^0_{0:t-1}} = \frac{j-1}{j} \mu^{\bar{\pi}^{j-1}}_{t \mid x^0_{0:t-1}, u^0_{0:t-1}} + \frac{1}{j} \mu^{\pi^{BR,j}}_{t \mid x^0_{0:t-1}, u^0_{0:t-1}}
\end{equation}
where we use the notation $\mu^{\pi}_{t \mid x^0_{0:t-1}, u^0_{0:t-1}}$ for the minor state distribution at time $t$ induced by the minor agent policy $\pi$ and conditioned on the past sequence $(x^0_{0:t-1}, u^0_{0:t-1})$. 
Here, $\mu_{t|x^0_{0:t-1}, u^0_{0:t-1}}^{\pi^{BR,j}}$ is the conditional distribution induced by the best response (BR) policy $\pi^{BR,j}$ against $\bar{\pi}^{j-1}$ and $\bar{\pi}^{0,j-1}$, i.e., $\pi^{BR,j} := \argmax_{\pi} J(\pi, \bar{\pi}^{j-1}, \bar{\pi}^{0, j-1})$. The policy generating this average distribution is
\small
\begin{multline}
\label{eq:fict_minorpolicy_avg}
    \bar{\pi}_t^{j}(u|x, x^0_{0:t-1}, u^0_{0:t-1}) \\
    = \dfrac{\sum_{i=0}^j \mu^{\pi^{BR,i}}_{t|x^0_{0:t-1}, u^0_{0:t-1}}(x)\pi_t^{BR,i}(u|x, x^0_{0:t-1}, u^0_{0:t-1})}{\sum_{i=0}^j \mu^{\pi^{BR,i}}_{t|x^0_{0:t-1}, u^0_{0:t-1}}(x)}.
\end{multline}
\normalsize

Meanwhile, the major player state distribution is
\begin{equation*}
\label{eq:fict_dist_avg_maj}
    {\mu}_{t}^{\bar{\pi}^{0,j}} = \frac{j-1}{j} {\mu}_{t}^{\bar{\pi}^{0,j-1}} + \frac{1}{j}\mu_{t}^{\pi^{0,BR,j}}
\end{equation*}
where $\bar{\pi}_t^{0,j}$ analogous to~\eqref{eq:fict_minorpolicy_avg}, but in contrast to minor agents using joint distributions $\mu^{\pi^{0,BR,i}}_{t}(x^0_t, x^0_{0:t-1}, u^0_{0:t-1})$ and $\pi_t^{0,BR,i}(u^0_t \mid x^0_t, x^0_{0:t-1}, u^0_{0:t-1})$.

For the convergence analysis, we study the continuous time version of above discrete time FP, as~\citet{perrin2020fictitious}.
In the continuous time FP algorithm, we denote the time of the algorithm (its ``iterations'') with $\tau$ and we first initialize the algorithm for $\tau<1$ with arbitrary policies for the minor players, $\bar{\pi}^{\tau<1}=\{\bar{\pi}_t^{\tau<1}\}_{t\in \mathcal{T}}$, and major player, $\bar{\pi}^{0,\tau<1}=\{\bar{\pi}_t^{0,\tau<1}\}_{t\in \mathcal{T}}$. 
For all $\tau\geq 1$, $t\in\mathcal{T}$ and $x^0_{0:t-1}, u^0_{0:t-1}$, define the FP process
\small 
\begin{equation}
\begin{aligned}
\label{eq:fict_play_cont_alg_int}
    \bar \mu^{\tau}_{t \mid x^0_{0:t-1}, u^0_{0:t-1}} &= \frac 1 \tau \int_0^\tau \mu^{\pi^{BR, s}_{0:t-1}}_{t \mid x^0_{0:t-1}, u^0_{0:t-1}} ds \\
    \bar \mu^{0, \tau}_t &= \frac 1 \tau \int_0^\tau \mu^{\pi^{0, BR, s}_{0:t-1}}_t ds
\end{aligned}
\end{equation}
\normalsize
where $\mu_{t \mid x^0_{0:t-1}, u^0_{0:t-1}}^{\pi^{BR,\tau}_{0:t-1}}$ and $\mu_t^{\pi^{0, BR,\tau}_{0:t-1}}$ are conditional and joint distributions respectively, induced by the BR policies $\pi^{BR,\tau}$ and $\pi^{0, BR,\tau}$ up to time $t-1$ against ${\mu}_{t \mid x^0_{0:t-1}, u^0_{0:t-1}}^{\bar{\pi}^\tau}(x)$ and ${\mu}_t^{\bar{\pi}^{0,\tau}}(x^0_t, x^0_{0:t-1}, u^0_{0:t-1})$. In other words, $\pi^{BR,\tau} :=\argmin_{\pi} J(\pi, \bar{\pi}^\tau, \bar{\pi}^{0,\tau})$ and $\pi^{0,BR,\tau} :=\argmin_{\pi^0} J^0(\bar{\pi}^\tau, {\pi}^{0})$. 

Note that the distributions induced by the averaged policies $\{\bar{\pi}^{\tau}_t\}_{t\in \mathcal{T}}$ and $\{\bar{\pi}^{0,\tau}_t\}_{t\in \mathcal{T}}$ for $\tau\geq1$ are given as
\small
\begin{equation}
    \begin{aligned}
    \label{eq:fict_play_policy_int}
        &\bar{\pi}^{\tau}_t(u|x, x^0_{0:t-1}, u^0_{0:t-1})\int_{s=0}^\tau \mu_{t|x^0_{0:t-1}, u^0_{0:t-1}}^{\pi^{BR, s}}(x) ds \\
        &=\int_{s=0}^\tau \mu_{t,|x^0_{0:t-1}, u^0_{0:t-1}}^{\pi^{BR, s}}(x) \pi_t^{BR, s}(u|x, x^0_{0:t-1}, u^0_{0:t-1}) ds,\\
        &\bar{\pi}^{0,\tau}_t(u^0|x^0, x^0_{0:t-1}, u^0_{0:t-1})\int_{s=0}^\tau \mu_{t}^{\pi^{0,BR, s}}(x^0, x^0_{0:t-1}, u^0_{0:t-1}) ds \\
        &=\int_{s=0}^\tau \mu_{t}^{\pi^{0,BR, s}}(x^0, x^0_{0:t-1}, u^0_{0:t-1})\\ 
        &\qquad \qquad \qquad \cdot \pi_t^{0,BR, s}(u^0|x^0, x^0_{0:t-1}, u^0_{0:t-1}) ds,\\
    \end{aligned}
\end{equation}
\normalsize
for all $t\in \mathcal{T}$ and $x^0_{0:t-1}, u^0_{0:t-1}$. For $s<1$, $\pi^{BR,s}$ and $\pi^{0,BR,s}$ are chosen arbitrarily. The proof and the differential form of equations~\eqref{eq:fict_play_cont_alg_int} and \eqref{eq:fict_play_policy_int} can be found in Appendix~\ref{app:fic_play}.

As a result, below we give a convergence analysis together with assumptions for continuous time FP, converging in both minor and major exploitability $\mathcal E(\bar{\pi}^\tau, \bar{\pi}^{0,\tau}) = \max_{\pi^{\prime}}  J(\pi^{\prime},\bar{\pi}^\tau, \bar{\pi}^{0,\tau}) - J(\bar{\pi}^\tau,\bar{\pi}^\tau, \bar{\pi}^{0,\tau})$, $\mathcal E^0(\bar{\pi}^\tau, \bar{\pi}^{0,\tau}) = \max_{\pi^{0\prime}} J^0(\bar{\pi}^\tau, \pi^{0\prime}) - J^0(\bar{\pi}^\tau, \bar{\pi}^{0,\tau})$, summarized as the total exploitability $\mathcal E_{\mathrm{tot}}(\bar{\pi}^\tau, \bar{\pi}^{0,\tau}) = \mathcal E(\bar{\pi}^\tau, \bar{\pi}^{0,\tau}) + \mathcal E^0(\bar{\pi}^\tau, \bar{\pi}^{0,\tau})$. 

\begin{assumption}
\label{as:fictitious_play_conv}
\begin{enumerate}[(a)]
    \item \label{as:fic_play_conv_transition}  The transition kernels are in the form of $P(x_{t+1} \mid x_t,u_t, x^0_t, u^0_t)$ and $P^0(x^0_{t+1} \mid x^0_t, u^0_t)$ for minor players and major player, respectively.
    \item \label{as:fic_play_conv_reward} The reward of minor and major players are separable, i.e. for some reward functions $\tilde r, \overline{r}, \tilde r^0, \hat r^0, \check r^0$, we have
    \begin{equation*}
        \begin{aligned}
            r(x, u, x^0, u^0, \mu) &=  \tilde{r}(x, x^0, u) + \overline{r}(x, x^0, \mu),\\
            r^0(x^0, u^0, \mu) &=  \tilde{r}^0(x^0, u^0) + \overline{r}^0(x^0, \mu).
        \end{aligned}
    \end{equation*}
    \item \label{as:fic_play_conv_explo_2} The game is monotone; i.e., satisfies Lasry-Lions monotonicity condition: For minor players, we have $\forall x^0 \in \mathcal{X}^0$, $\forall \mu, \mu^{\prime}$: $\sum_{x \in \mathcal{X}} (\mu(x) - \mu^{\prime}(x))(\overline{r} (x, x^0, \mu) - \overline{r} (x, x^0, \mu^{\prime}))\leq 0$. Meanwhile, for major players, we have $\frac{d}{d\tau} \mu^{\bar \pi^{0, \tau}_{0:t}}_{t+1}(x^0_{t+1}, x^0_{0:t}, u^0_{0:t}) \cdot \Big\langle \nabla_\mu \bar r^0(x^0_{t+1}, \mu^{\bar \pi^{\tau}_{0:t}}_{t+1 \mid x^0_{0:t}, u^0_{0:t}}), \frac{d}{d\tau} \mu^{\bar \pi^{\tau}_{0:t}}_{t+1 \mid x^0_{0:t}, u^0_{0:t}} \Big\rangle \leq 0$.
    \item \label{as:fic_play_conv_explo}  We have $\tilde {\mathcal E}(\bar \pi^\tau, \pi^{0, BR, \tau}, \bar \pi^{0,\tau}) \leq \mathcal E(\bar \pi^\tau, \bar \pi^{0,\tau})$, where we define $\tilde {\mathcal E}(\bar \pi^\tau, \pi^{0, BR, \tau}, \bar \pi^{0,\tau}) = J(\pi^{BR, \tau}, \bar \pi^\tau, \pi^{0, BR, \tau}) - J(\bar \pi^\tau, \bar \pi^\tau, \pi^{0, BR, \tau})$ with any BR policy given as $\pi^{BR, \tau} = \argmax_{\pi} J(\pi, \bar \pi^\tau, \bar \pi^{0,\tau})$.
    \end{enumerate}
\end{assumption}
\begin{remark}
    Assumption~\ref{as:fictitious_play_conv}.\ref{as:fic_play_conv_explo_2} is fulfilled for major players if $\overline{r}^0(x^0, \mu) = \overline{r}^0(x^0)$. Assumption~\ref{as:fictitious_play_conv}.\ref{as:fic_play_conv_explo} is satisfied for instance if $r(x, u, x^0, \mu) = r(x, u, \mu) $ and $P(x_{t+1} \mid x_t,u_t, x^0_t, u^0_t) = P(x_{t+1} \mid x_t,u_t)$. Then, we trivially have $\tilde {\mathcal E}(\bar \pi^\tau, \pi^{0, BR, \tau}, \bar \pi^{0,\tau}) = \mathcal E(\bar \pi^\tau, \bar \pi^{0,\tau})$ by obtaining a minor player MFG independent of the major player.
    \looseness=-1
\end{remark}
\begin{theorem}
\label{the:fictitiousplayconvergence}
    Under Assumption~\ref{as:fictitious_play_conv}, the total exploitability is a strong Lyapunov function such that $\frac{d}{d\tau}\mathcal E_{\mathrm{tot}}(\bar \pi^\tau, \bar \pi^{0,\tau}) \leq -\frac{1}{\tau}\mathcal E_{\mathrm{tot}}(\bar \pi^\tau, \bar \pi^{0,\tau})$; i.e., we have $\mathcal E_{\mathrm{tot}}(\bar \pi^\tau, \bar \pi^{0,\tau}) = \mathcal{O}(1/\tau)$ in the continuous time FP algorithm.
\end{theorem}
The proof of Theorem~\ref{the:fictitiousplayconvergence} can be found in Appendix~\ref{app:fic_play} 
and is based on a monotonic decrease of exploitability, at the same rate as standard FP in MFGs \citep{perrin2020fictitious}.

In numerical experiments, for applicability and computational tractability (due to the exponential complexity of the histories in the horizon), we condition policies on the random MF and major state instead of the histories, averaging policies uniformly instead of for each possible major state-action sequence. Further, numerically we partition and represent the (naturally continuous) MFs as described in the following, to obtain tabular Algorithm~\ref{alg:2}. Experimentally, in Section~\ref{sec:exp} we nonetheless find that the algorithm optimizes exploitability, even if Assumption~\ref{as:fictitious_play_conv} is not fully satisfied. The dependence of policy actions on the MF and major state has the additional advantage of allowing standard dynamic programming for major and minor MDPs, as their full MDP states include both the MF and major state.

\begin{algorithm}[b!]
    \caption{Discrete-time, projected fictitious play}
    \label{alg:2}
    \begin{algorithmic}[1]
        \STATE Input: $\delta$-partition $\{ \mathcal P_i \}_{i=1\ldots,M}$.
        \STATE Initialize initial policies $\bar \pi_{(0)}$, $\bar \pi^0_{(0)}$.
        \FOR {iteration $n = 0, 1, 2, \ldots$}
            \STATE Compute discretized BR (as in Definition~\ref{def:delta-m3fne})
            \small
            \begin{align*}
                \pi_{(n+1)} &\in \argmax \hat Q_{\bar \pi_{(n)}, \bar \pi^0_{(n)}}, \\ 
                \pi^0_{(n+1)} &\in \argmax \hat Q^0_{\bar \pi_{(n)}, \bar \pi^0_{(n)}}.
            \end{align*}
            \normalsize
            \STATE Compute next average policies 
            \small
            \begin{align*}
                \bar \pi_{(n+1)} &\coloneqq \frac n {n+1} \bar \pi_{(n)} + \frac 1 {n+1} \pi_{(n+1)}, \\
                \bar \pi^0_{(n+1)} &\coloneqq \frac n {n+1} \bar \pi^0_{(n)} + \frac 1 {n+1} \pi^0_{(n+1)}.
            \end{align*}
            \normalsize
        \ENDFOR
    \end{algorithmic}
\end{algorithm}

\subsection{Projected mean field}
Observe that for given current MF and major state-actions, we obtain deterministic transitions from one MF to the next. Therefore, by partitioning we can obtain deterministic transitions in-between parts of a partition of $\mathcal P(\mathcal X)$, and a Bellman equation over \textit{finite} spaces.

\begin{definition}
    A $\delta$-partition $\mathcal M = \{ \mathcal P_i \}_{i \in [|\mathcal M|]}$ is a partition of $\mathcal P(\mathcal X)$, with $\lVert \mu - \nu \rVert < \delta$ for any $i \in [|\mathcal M|]$, $\mu, \nu \in \mathcal P_i$.
\end{definition}

Since $\mathcal P(\mathcal X)$ is compact, a finite $\delta$-partition of $\mathcal P(\mathcal X)$ exists for any $\delta > 0$. We will henceforth assume for any $\delta > 0$ some $\delta$-partition $\mathcal M$ of $\mathcal P(\mathcal X)$ with $M = M(\delta)$ parts.

\subsubsection{Discretized finite MDPs.}
To each part $\mathcal P_i$, we associate an arbitrary element $\hat \mu^{(i)} \in \mathcal P_i$ and write $\mathrm{proj}_\delta \mu$ for the $\delta$-partition projection of MFs $\mu \in \mathcal P(\mathcal X)$, i.e. whenever $\mu \in \mathcal P_i$ we project to the representative $\mathrm{proj}_\delta \mu = \hat \mu^{(i)} \in \mathcal P_i$. 

As a result, we obtain discretized, \textit{finite} MDP versions of the major and minor player MDPs, where the continuous MF state is replaced by finitely many states in $\hat{\mathcal P}(\mathcal X) \coloneqq \{ \hat \mu^{(1)}, \ldots, \hat \mu^{(M)} \}$, evolving by discretized MF evolutions in \eqref{eq:m3fg-min}, i.e. $\hat \mu_{t+1} = \mathrm{proj}_\delta T^\pi_t(x^0, u^0, \hat \mu_t)$ for any $x^0, u^0, \hat \mu_t$.

We can solve the discretized MDPs in a tabular manner: To compute best responses under policies $(\pi, \pi^0)$, observe that the true action-value function $Q^0_{\pi, \pi^0}$ of the (not discretized) major player MDP follows the Bellman equation
\small
\begin{multline*}
    Q^0_{\pi, \pi^0}(t, x^0, u^0, \mu)= r^0(x^0, u^0, \mu)
    + \sum_{x^{0\prime}} P^0(x^{0\prime} \mid x^0, u^0, \mu) \\
    \cdot \max_{u^{0\prime}} Q^0_{\pi, \pi^0}(t+1, x^{0\prime}, u^{0\prime}, T^\pi_t(x^0, u^0, \mu)).
\end{multline*}
\normalsize
The tabular approximate action-value function $\hat Q^0_{\pi, \pi^0}$ for the major player follows instead the Bellman equation of the discretized major player MDP (letting the domain of $\hat Q_{\pi, \pi^0}$ be the entirety of $\mathcal P(\mathcal X)$ as constants over each part $\mathcal P_i$), 
\small
\begin{multline*}
    \hat Q^0_{\pi, \pi^0}(t, x^0, u^0, \mu) = \hat Q^0_{\pi, \pi^0}(t, x^0, u^0, \mathrm{proj}_\delta \mu) \\
    = r^0(x^0, u^0, \mathrm{proj}_\delta\mu)
    + \sum_{x^{0\prime}} P^0(x^{0\prime} \mid x^0, u^0, \mathrm{proj}_\delta\mu) \\
    \cdot \max_{u^{0\prime}} \hat Q^0_{\pi, \pi^0}(t+1, x^{0\prime}, u^{0\prime}, T^\pi_t(x^0, u^0, \mathrm{proj}_\delta\mu))
\end{multline*}
\normalsize
with terminal condition zero, and the minor action-values analogously. The above can nonetheless provide a good approximation that can be computed in \textit{tabular} form, see Appendix~\ref{app:theo} and empirical support in Section~\ref{sec:exp}.

\subsubsection{Discretized equilibria.}
Building upon the preceding approximations, we define an approximate equilibrium as a fixed point of the discretized system. 

\begin{definition} \label{def:delta-m3fne}
    A $\delta$-partition M3FNE is a tuple $(\pi, \pi^0) \in \hat \Pi \times \hat \Pi^0$ with $\pi \in \argmax \hat Q_{\pi, \pi^0}$ and $\pi^0 \in \argmax \hat Q^0_{\pi, \pi^0}$ where policies in $\hat \Pi, \hat \Pi^0$ are instead defined as blockwise constant over each part $\mathcal P_i$ of the $\delta$-partition.
\end{definition}

Here, we understand $\hat \pi \in \argmax \hat Q_{\pi, \pi^0}$ by the defining equation $\sum_{u \in \argmax_{u'} \hat Q_{\pi, \pi^0}(t, x, u', x^0, \hat \mu)} \hat \pi_t(x, x^0, \hat \mu, u) = 1$ for all $(t, x, x^0, \hat \mu) \in \mathcal T \times \mathcal X \times \mathcal X^0 \times \hat{\mathcal P}(\mathcal X)$, and similarly for major players, noting that $\hat \pi$ optimizes the preceding discretized finite MDP \citep{hernandez2012discrete}. 

We note that while the discretized solutions only piecewise fulfill Assumption~\ref{ass:m3picont} by not being Lipschitz, in Section~\ref{sec:exp} we empirically find that the approximation of finite games and exploitability can nonetheless be accurate.

\subsubsection{Approximation guarantees.}
We evaluate solutions by tabular evaluation in the discretized MDP, for which we are able to obtain theoretical guarantees for evaluating the true exploitability via the approximate tabular exploitability. Under a $\delta$-partition, define the major approximate objective
\small
\begin{align*}
    \hat J^0(\pi, \pi^0) \coloneqq \sum_{x^0} \mu^0_0(x^0) \hat V^{0, \pi^0}_{\pi, \pi^0}(0, x^0, \mu_0)
\end{align*} 
\normalsize
and approximate exploitability
\small
\begin{multline*}
    \hat{\mathcal E}^0(\pi, \pi^0) \coloneqq \sum_{x^0} \mu^0_0(x^0) \\
    \cdot \Big( \max_{\hat \pi^{0\prime} \in \hat \Pi} \hat V^{0, \hat \pi^{0\prime}}_{\pi, \pi^0}(0, x^0, \mu_0)) - \hat V^{0, \pi^0}_{\pi, \pi^0}(0, x^0, \mu_0) \Big),
\end{multline*}
\normalsize
with approximate values $\hat V_{\pi, \pi^0}^{0, \hat \pi^0}$ of major deviation under $({\pi, \pi^0})$ to $\hat \pi^0$, following the ``discretized'' Bellman equation
\small
\begin{multline*}
    \hat V_{\pi, \pi^0}^{0, \hat \pi^0}(t, x^0, \mu) = \sum_{u^{0\prime}} \hat \pi^0_t(u^{0\prime} \mid x^{0\prime}, \mathrm{proj}_\delta\mu) \\
    \Big[ r^0(x^0, u^0, \mathrm{proj}_\delta\mu) + \sum_{x^{0\prime}} P^0(x^{0\prime} \mid x^0, u^0, \mathrm{proj}_\delta\mu) \\
    \hat V_{\pi, \pi^0}^{0, \hat \pi^0}(t+1, x^{0\prime}, u^{0\prime}, T^\pi_t(x^0, u^0, \mathrm{proj}_\delta\mu)) \Big],
\end{multline*}
\normalsize
and similarly for the minor player. Note that only for the major player, $\pi^0$ is irrelevant (replaced by $\hat \pi^0$). In other words, we approximate values and exploitability via the discretized finite MDPs, which has the advantage of enabling dynamic programming (backwards induction, value iteration). 

By analyzing the value functions under continuity, we show in Appendix~\ref{app:disc-opt} that these approximations are generally close to the true objectives and exploitabilities respectively, as the discretization becomes sufficiently fine.

\begin{theorem} \label{thm:disc-opt}
Under Assumptions~\ref{ass:m3pcont}, \ref{ass:m3rcont}, \ref{ass:m3picont}, as $\delta \to 0$, approximate minor and major values tend to the exact values, and approximate exploitabilities tend to the exact exploitabilities, at rate $\mathcal O(\delta)$ uniformly over $(\pi, \pi^0) \in \Pi \times \Pi^0$.
\end{theorem}

\section{Experiments} \label{sec:exp}
We evaluate FP by comparing against fixed-point iteration (FPI), which iterates discretized best response policies. For reproducibility, note that the algorithms used are deterministic, and details can be found in the appendix.\footnote{For code, see \url{https://github.com/tudkcui/M3FG-learning}}
\begin{figure}[b]
    \centering
    \includegraphics[width=0.99\linewidth]{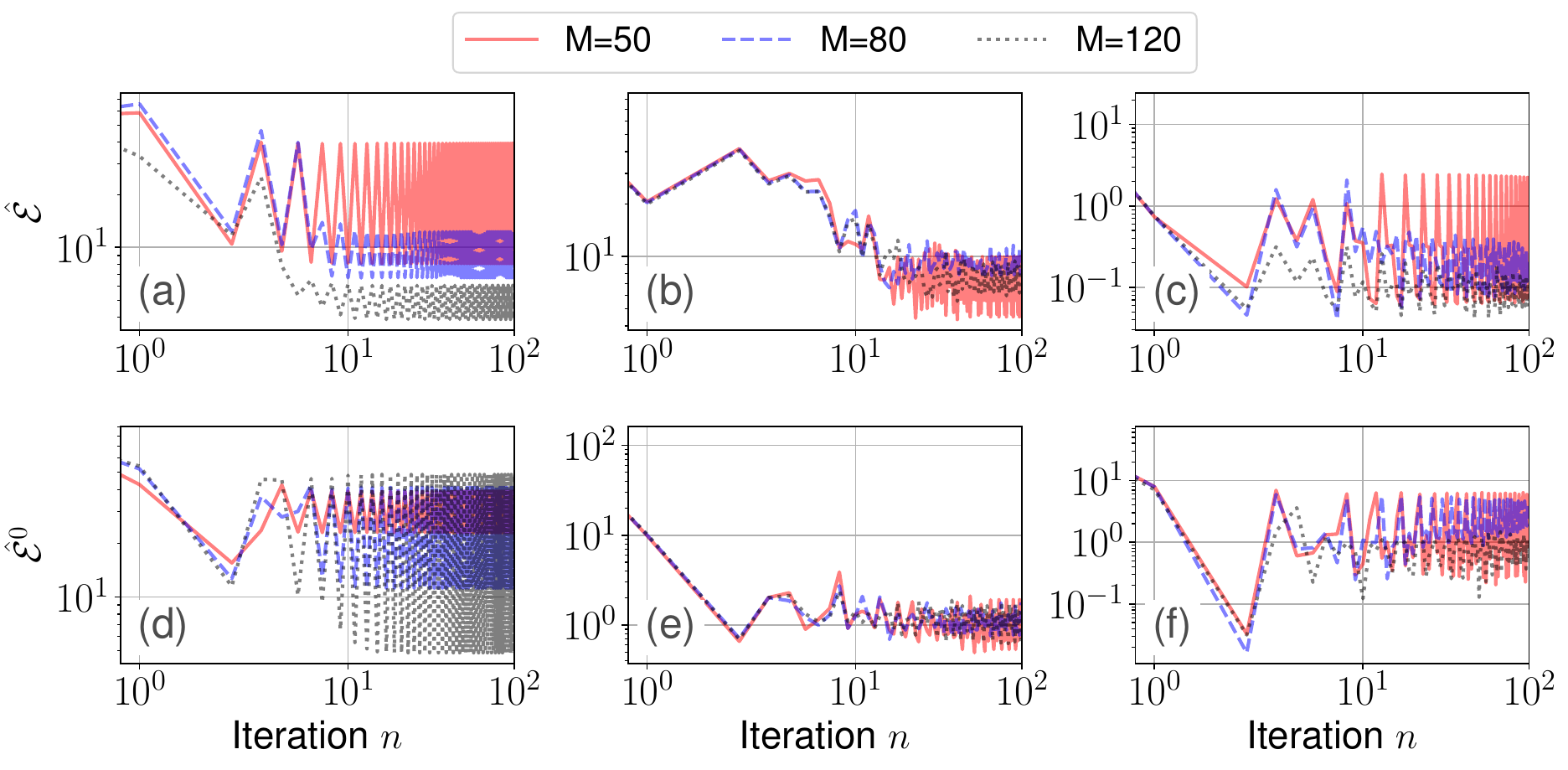}
    \caption{The approximate exploitability oscillates over iterations of FPI. (a, d): SIS, (b, e): Buffet, (c, f): Advertisement. (a-c): Minor exploitability, (d-f): major exploitability.}
    \label{fig:exploitability-fpi}
\end{figure}

\begin{figure}[t]
    \centering
    \includegraphics[width=0.99\linewidth]{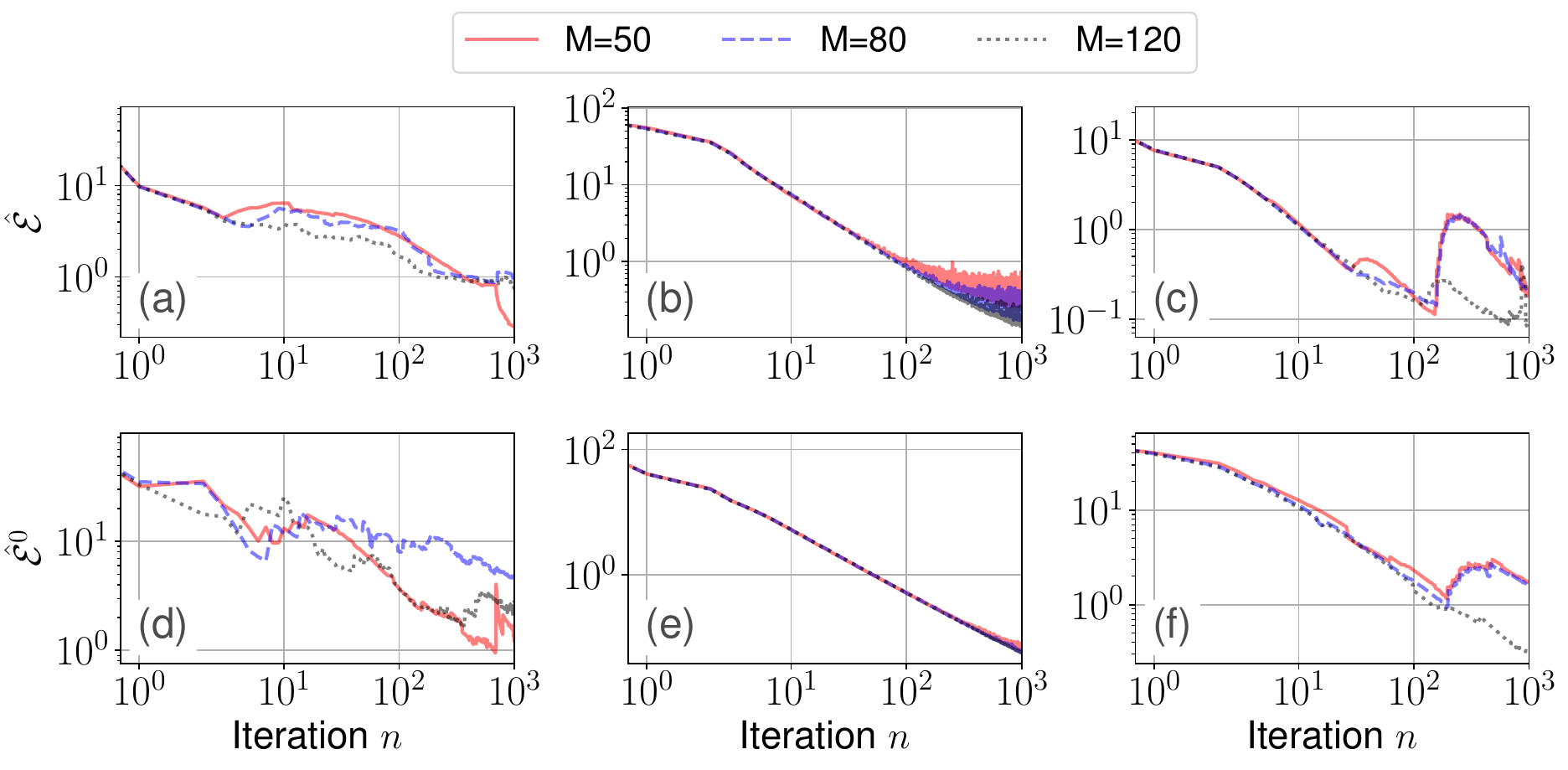}
    \caption{The approximate exploitability is optimized via FP. (a, d): SIS, (b, e): Buffet, (c, f): Advertisement. (a-c): Minor exploitability, (d-f): major exploitability.}
    \label{fig:exploitability}
\end{figure}

\begin{figure}[b]
    \centering
    \includegraphics[width=0.99\linewidth]{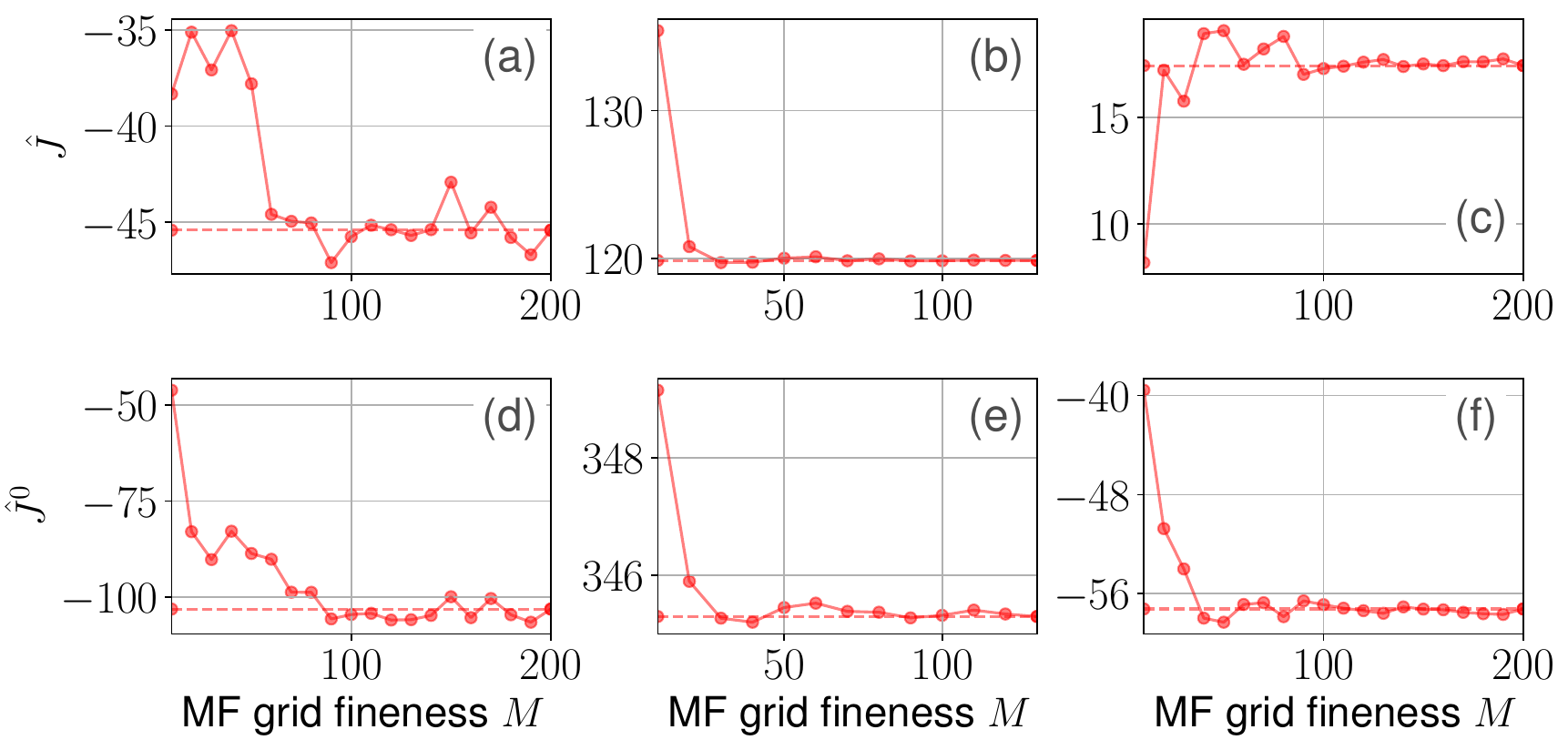}
    \caption{The final objectives of FP under discretization (dashed: right-most entry) are stable with high discretization. (a, d): SIS, (b, e): Buffet, (c, f): Advertisement. (a-c): Minor exploitability, (d-f): major exploitability.}
    \label{fig:J_discretization}
\end{figure}

\begin{figure}[t]
    \centering
    \includegraphics[width=0.99\linewidth]{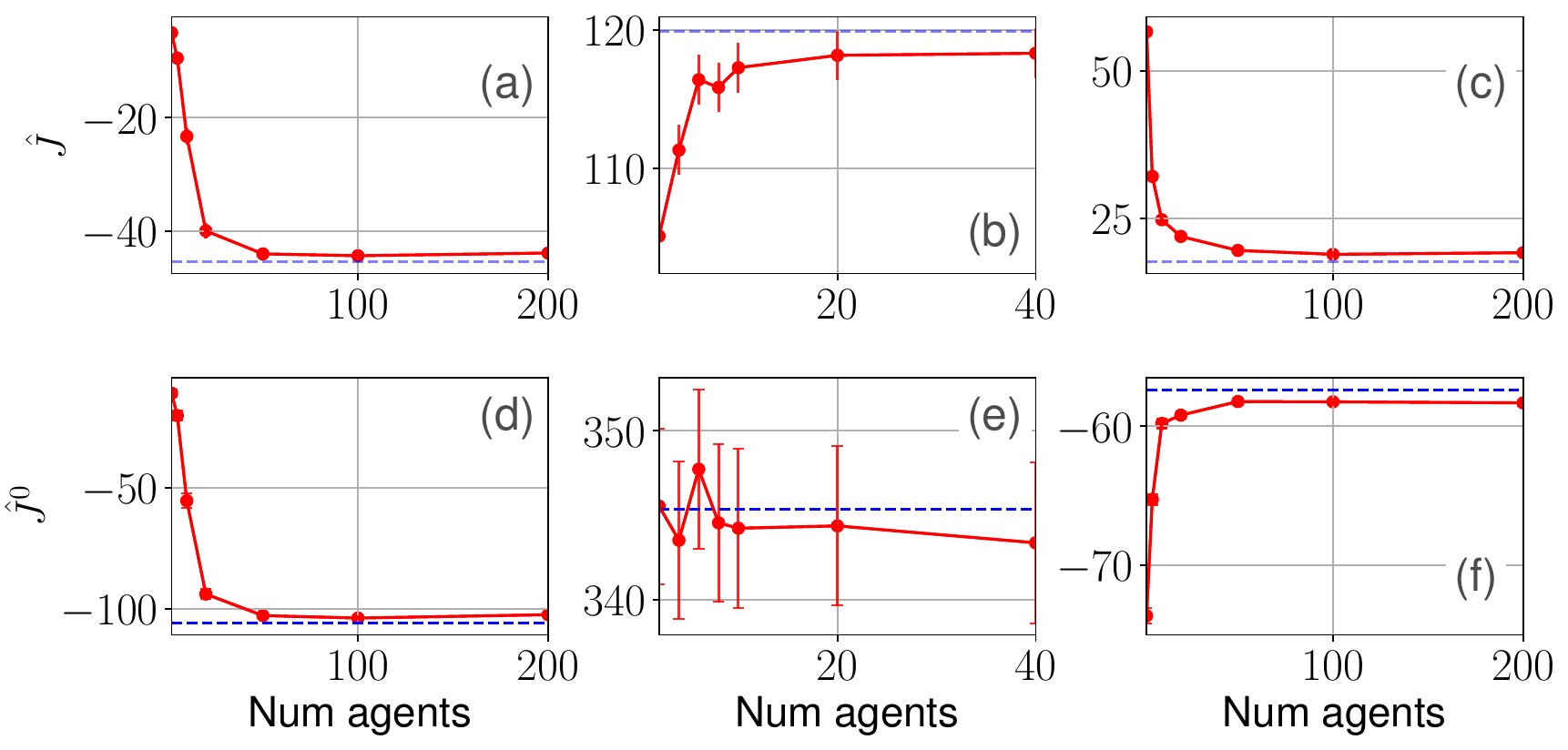}
    \caption{The mean $N$-player objective (red) over $1000$ (or $5000$ for Buffet) episodes, with $95\%$ confidence interval, against MF predictions $\hat J$, $\hat J^0$ for FP and $M=120$ (blue, dashed). (a, d): SIS, (b, e): Buffet, (c, f): Advertisement.}
    \label{fig:J_num_agents}
\end{figure}

\begin{figure}[b!]
    \centering
    \includegraphics[width=0.99\linewidth]{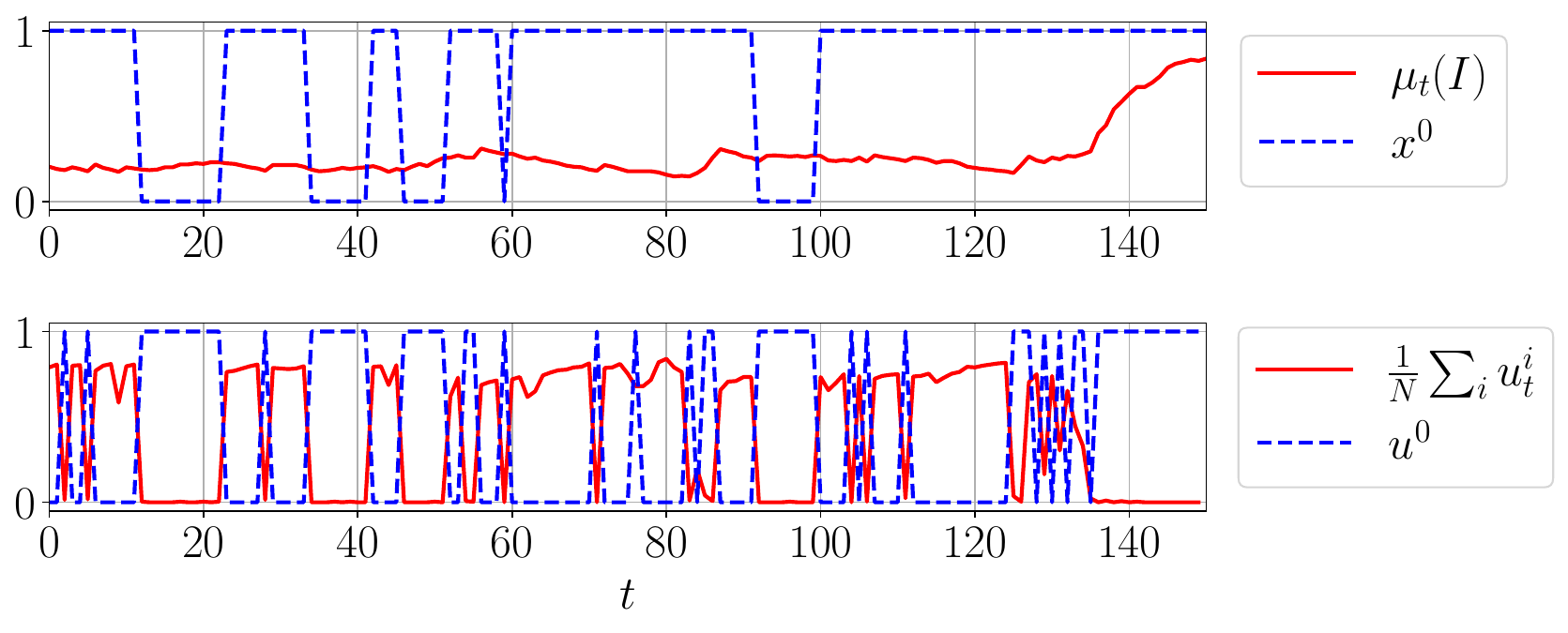}
    \includegraphics[width=0.99\linewidth]{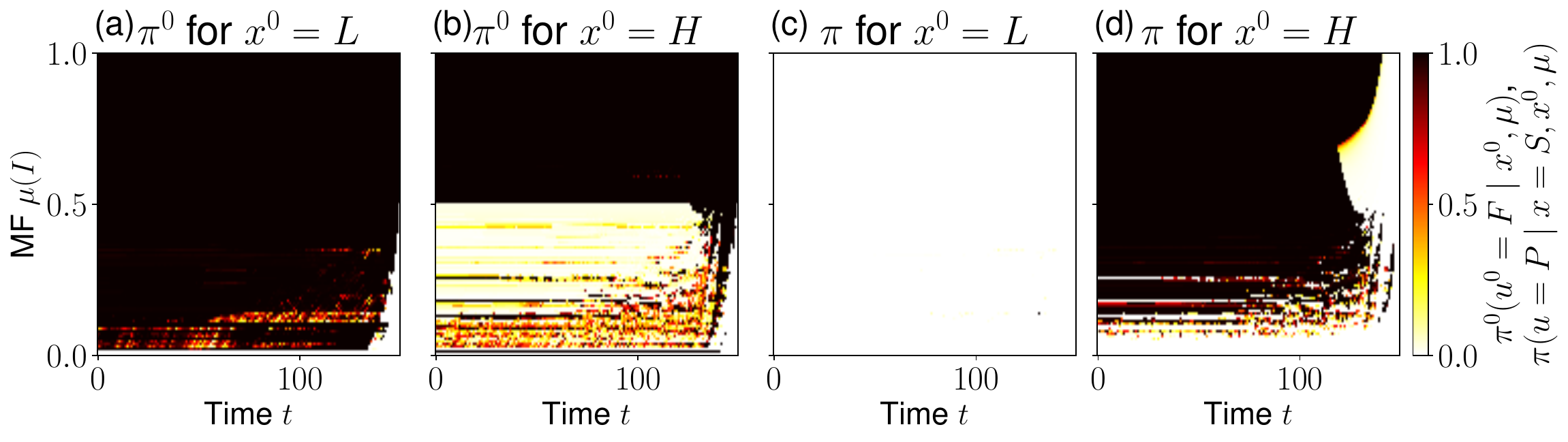}
    \caption{The FP-learned M3FNE in SIS, for $M=120$. Top: example trajectory (for visualization, $L=\bar P=\bar F=0, H=P=F=1$, see Appendix~\ref{app:exp}); bottom: policy.}
    \label{fig:qual}
\end{figure}

\subsection{Problems}
For the evaluation, we use the following problem instances for exemplary, practically applicable M3FG scenarios.

\subsubsection{SIS epidemics control.}
The SIS problem is an epidemics control scenario, where each individualistic minor player may decide whether to take costly preventative actions against becoming infected at a rate proportional to the proportion of infected. The major player (e.g. government) is responsible for the well-being of minor players, and can encourage preventative actions, while its state models random low- and high-infectivity seasons. The finite time horizon can be considered the time until a cure is found. The original problem without major players has been used as a benchmark for MFG learning \citep{cui2021approximately, lauriere2022scalable}.

\subsubsection{Buffet problem.}
In the Buffet problem, we consider the following scenario: At a conference with multiple buffet locations, players desire to be at locations that are filled with food and uncrowded. However, each location depletes faster with increasing number of players. The major player (caterer) must keep buffets full and equally filled. The Buffet problem fulfills most assumptions (except Assm.~\ref{as:fictitious_play_conv}.\ref{as:fic_play_conv_explo}) and shows accordingly stable FP learning.

\subsubsection{Advertisement duopoly model.}
Lastly, in the advertisement model, a regulator sets the price of advertisement. Depending on the regulator's state and price of advertisement, two companies exogeneously decide on advertisement efficiencies of their subscription service. Minor players are consumers and choose whether to change to subscriptions for the better-funded product, while the regulator avoids formation of a monopoly. Duopoly advertisement competition in a static MFG was modeled in~\cite{carmona_dayanikli2021adv}.

\subsection{Numerical results}
In the following, we provide a numerical evaluation via exploitability as the primary metric of interest, since it describes the quality of achieved equilibria. Additional experiments and parameter details are shown in Appendix~\ref{app:exp}, including more qualitative results, the effect of alternative initializations, and analogous results for infinite-horizon discounted objectives. Beyond supporting the theoretical results, we also ablate both convergence assumptions for the algorithm and the Lipschitz policy assumption for propagation of chaos in the finite player system.

\paragraph{Exploitability convergence.}
As observed in Figure~\ref{fig:exploitability-fpi}, naive FPI usually fails to converge and runs into limit cycles, motivating FP. In Figure~\ref{fig:exploitability}, we see that the proposed FP algorithm optimizes both approximate major and minor exploitabilities $\hat{\mathcal E}$, $\hat{\mathcal E}^0$ over its iterations. Especially for Buffet, which fulfills most of Assumption~\ref{as:fictitious_play_conv}, learning is smooth and exploitability descends monotonically as in Theorem~\ref{the:fictitiousplayconvergence}, while exploitability is nevertheless optimized in the other problems. Overall, the proposed FP algorithm improves achieved exploitabilities significantly over FPI.

\paragraph{Stability over discretization.}
Comparing approximation results empirically over discretization bins $M$ per dimension, i.e. using $\delta$-partitions with $\delta \approx \frac{2}{M}$, in Figure~\ref{fig:J_discretization} we observe that the FP-learned policies quickly stabilize as the discretization becomes sufficiently fine. The result supports not only the discretization approximation in Theorem~\ref{thm:disc-opt}, but also shows insensitivity of our FP algorithm to the fineness of the grid, as long as it is sufficiently fine to approximate the problem well. Hence, in the following we will use $M=120$.

\paragraph{Finite-player convergence.}
In Figure~\ref{fig:J_num_agents}, the convergence of episodic returns by propagation of chaos is depicted as the number of players $N \to \infty$. The limiting performance as the number of players grows, quickly approaches the performance of the projected MF prediction, up to a small, negligible error from discretization and finite players. The result supports propagation of chaos in Theorem~\ref{thm:m3muconv} by convergence of the empirical objective to the limiting objective, despite the non-Lipschitz projected MF policies. In Appendix~\ref{app:exp}, similar results hold for (Lipschitz) uniform policies.

\paragraph{Qualitative analysis.}
Lastly, we visualize the qualitative behavior obtained and find plausible equilibrium behavior, e.g., for the SIS problem. As seen in Figure~\ref{fig:qual}, the equilibrium behavior plausibly reaches an equilibrium of infected players, where the cost of actions equilibrates. The number of infected increases over time due to the finite horizon, discounting costs of infection beyond the horizon. Furthermore, minor players take precautions only down to some infection threshold, at which point the expected cost of not taking precautions is higher. The major player prevents infections in the low-infectivity regime ($x^0 = L$), while in the high-infectivity regime ($x^0 = H$) the high infection probability for minor players already encourages preventative actions.

\section{Conclusion and Discussion}
We have developed a new model and algorithm for a novel, broad class of tractable games. The framework allows scalable analysis of a large number of players with theoretical guarantees. The proposed methods have been empirically supported through a variety of experiments. Still, for problems with multiple Nash equilibria, the FP algorithm finds only some equilibrium. Future work could address finding all or specific, e.g., socially-optimal equilibria. One could also try to relax theoretical assumptions. Lastly, since scalability of the discretization method remains an issue for larger minor state spaces, one may consider deep RL methods.


\section*{Acknowledgments}
This work has been co-funded by the LOEWE initiative (Hesse, Germany) within the emergenCITY center. The authors acknowledge the Lichtenberg high performance computing cluster of the TU Darmstadt for providing computational facilities for the calculations of this research.
M. Lauriere is affiliated with the Shanghai Frontiers Science Center of Artificial Intelligence and Deep Learning, and with the NYU-ECNU Institute of Mathematical Sciences at NYU Shanghai, NYU Shanghai, 567 West Yangsi Road, Shanghai, 200126, People’s Republic of China. O. Pietquin was affiliated with Google DeepMind during the preparation of this work. 

\bibliography{aaai24}


\onecolumn
\appendix

\section{Continuous Time Fictitious Play with Major and Minor Players}
\label{app:fic_play}
In this section, we prove the fictitious play convergence result for the mean field game with a major player by extending the ideas of~\cite{perrin2020fictitious} to include the major player. \vskip6pt

The aim of this proof is to show the total exploitability is a strong Lyapunov function by showing $\frac{d}{d\tau}\left(\mathcal E(\bar{\pi}^\tau, \bar{\pi}^{0,\tau}) + \mathcal E^0(\bar{\pi}^\tau, \bar{\pi}^{0,\tau})\right)\leq -\frac{1}{\tau}\left(\mathcal E(\bar{\pi}^\tau, \bar{\pi}^{0,\tau}) + \mathcal E^0(\bar{\pi}^\tau, \bar{\pi}^{0,\tau})\right)$. In the proof we focus on showing $\frac{d}{d\tau}\left(\mathcal E(\bar{\pi}^\tau, \bar{\pi}^{0,\tau})\right)\leq -\frac{1}{\tau}\left(\mathcal E(\bar{\pi}^\tau, \bar{\pi}^{0,\tau})\right)$ first and $\frac{d}{d\tau}\left(\mathcal E^0(\bar{\pi}^\tau, \bar{\pi}^{0,\tau})\right)\leq -\frac{1}{\tau}\left(\mathcal E^0(\bar{\pi}^\tau, \bar{\pi}^{0,\tau})\right)$ second and we combine the results at the end. Before going into the details of the proof of Theorem~\ref{the:fictitiousplayconvergence}, we introduce some properties that extend from the ones introduced by~\citet{perrin2020fictitious} where we exchange their common noise formulation $\Sigma_t^0$ with state-action histories (i.e., we condition on the noise in order to avoid having to conditioning on the major agent's randomness as common noise).

We first recall a few definitions:

\paragraph{Mean field.}
We recall the conditional MFs for minor agents, now conditioned on histories, i.e. recursively
\begin{align*}
    &\mu^{\pi}_{t+1 \mid x^0_{0:t}, u^0_{0:t}}(x') \coloneqq \mathbb P_{\pi}(x_{t+1} = x' \mid x^0_{0:t}, u^0_{0:t}) \\
    &= \sum_{x,u\in \mathcal{X}\times\mathcal{U}} \mathbb P_{\pi}(x_{t+1} = x' \mid x_t = x, u_t = u, x^0_{0:t}, u^0_{0:t}) \mathbb P_{\pi}(u_t = u \mid x_t = x, x^0_{0:t}, u^0_{0:t}) \mathbb P_{\pi}(x_t = x \mid x^0_{0:t}, u^0_{0:t}) \\
    &= \sum_{x,u\in \mathcal{X}\times\mathcal{U}} P(x^\prime|x,u, x^0_t, u^0_t)\pi_t(u|x, x^0_{0:t}, u^0_{0:t-1}) \mu^{\pi}_{t \mid x^0_{0:t-1}, u^0_{0:t-1}}(x).
\end{align*}
while for major agents, we define joint history mean fields
\begin{align*}
    \mu^{\pi^0}_{t+1}(x^0_{t+1}, x^0_{0:t}, u^0_{0:t}) &\coloneqq \mathbb P_{\pi^0}(x^0_{t+1}, x^0_{0:t}, u^0_{0:t}) \\
    &= \mu^{\pi^0}_t(x^0_t, x^0_{0:t-1}, u^0_{0:t-1}) \pi^0_t(u^0_t \mid x^0_t, x^0_{0:t-1}, u^0_{0:t-1}) P^0(x^0_{t+1} \mid x^0_t, u^0_t)
\end{align*} 

\paragraph{Averaged policies.}
We recall the minor agent FP policy as
\begin{align}
    \bar \pi^\tau_t(u_t \mid x_t, x^0_{0:t}, u^0_{0:t-1}) = \frac{\int_0^\tau \pi^{BR, s}(u_t \mid x_t, x^0_{0:t}, u^0_{0:t-1}) \mu^{\pi^{BR, s}_{0:t-1}}_{t \mid x^0_{0:t-1}, u^0_{0:t-1}}(x_t) ds}{\int_0^\tau \mu^{\pi^{BR, s}_{0:t-1}}_{t \mid x^0_{0:t-1}, u^0_{0:t-1}}(x_t) ds}
\end{align}
and similarly for the major agent
\begin{align}
    \bar \pi^{0, \tau}_t(u^0_t \mid x^0_t, x^0_{0:t-1}, u^0_{0:t-1}) = \frac{\int_0^\tau \pi^{0, BR, s}(u^0_t \mid x^0_t, x^0_{0:t-1}, u^0_{0:t-1}) \mu^{\pi^{0, BR, s}_{0:t-1}}_t(x^0_t, x^0_{0:t-1}, u^0_{0:t-1}) ds}{\int_0^\tau \mu^{\pi^{0, BR, s}_{0:t-1}}_t(x^0_t, x^0_{0:t-1}, u^0_{0:t-1}) ds}.
\end{align}

\paragraph{Averaged mean field.}
We also recall the average FP MFs for minor players
\begin{align}
    \bar \mu^{\tau}_{t \mid x^0_{0:t-1}, u^0_{0:t-1}}(x_t) &\coloneqq \frac 1 \tau \int_0^\tau \mu^{\pi^{BR, s}_{0:t-1}}_{t \mid x^0_{0:t-1}, u^0_{0:t-1}}(x_t) ds.
\end{align}
and for major players
\begin{align}
    \bar \mu^{0, \tau}_t(x^0_t, x^0_{0:t-1}, u^0_{0:t-1}) &\coloneqq \frac 1 \tau \int_0^\tau \mu^{\pi^{0, BR, s}_{0:t-1}}_t(x^0_t, x^0_{0:t-1}, u^0_{0:t-1}) ds
\end{align}

\vskip6pt
\begin{property}
\label{property:monotonicity}
    When the game is monotone, we have:
    \begin{equation*}
        \begin{aligned}
            \sum_{x\in \mathcal{X}} \Big\langle \nabla_\mu \overline{r}(x, x^0, \mu), \frac{d}{d\tau} \mu\Big\rangle \frac{d}{d\tau}\mu(x) \leq 0
        \end{aligned}
    \end{equation*}
\end{property}
\begin{proof}
    For all $s\geq0$ monotonicity condition says that for a fixed $x^0\in\mathcal{X}^0$:
    \begin{equation*}
    \begin{aligned}
        &\sum_{x \in \mathcal{X}} (\mu^\tau(x) - \mu^{\tau+s}(x))(\overline{r} (x, x^0, \mu^\tau) - \overline{r} (x, x^0, \mu^{\tau+s}))\leq 0\\
        \Rightarrow& \sum_{x \in \mathcal{X}} \frac{\mu^\tau(x) - \mu^{\tau+s}(x)}{s}\frac{\overline{r} (x, x^0, \mu^\tau) - \overline{r} (x, x^0, \mu^{\tau+s})}{s}\leq 0.
    \end{aligned}
    \end{equation*}
    Therefore, the result follows when $s\rightarrow 0$.
\end{proof}

\begin{property} \label{prop:1}
The above FP policy $\bar \pi^0$ generates the average FP mean field $\bar \mu$, i.e.
\begin{align}
    \bar \mu^{0, s}_t = \mu^{\bar \pi^{0, s}_{0:t-1}}_t
\end{align}
at all times $t, \tau$, and analogously for the minor players
\begin{align}
    \bar \mu^{s}_t = \mu^{\bar \pi^{s}_{0:t-1}}_t.
\end{align}

\end{property}
\begin{proof}
The joint probabilities of any policy $\pi^{0, BR, s}$ are always given by
\begin{align*}
    \mu^{\pi^{0, BR, s}_{0:t}}_{t+1}(x^0_{t+1}, x^0_{0:t}, u^0_{0:t}) = P^0(x^0_{t+1} \mid x^0_t, u^0_t) \mu^{\pi^{0, BR, s}_{0:t-1}}_t(x^0_t, x^0_{0:t-1}, u^0_{0:t-1}) \pi^{0, BR, s}_t(u^0_t \mid x^0_t, x^0_{0:t-1}, u^0_{0:t-1})
\end{align*}
and therefore, integrating over all times $\tau$, we have
\begin{align*}
    \frac 1 \tau \int_0^\tau \mu^{\pi^{0, BR, s}_{0:t}}_{t+1}(x^0_{t+1}, x^0_{0:t}, u^0_{0:t}) ds = P^0(x^0_{t+1} \mid x^0_t, u^0_t) \frac 1 \tau \int_0^\tau \mu^{\pi^{0, BR, s}_{0:t-1}}_t(x^0_t, x^0_{0:t-1}, u^0_{0:t-1}) \pi^{0, BR, s}_t(u^0_t \mid x^0_t, x^0_{0:t-1}, u^0_{0:t-1}) ds.
\end{align*}
Then, using definitions of $\bar \mu^{0, s}$ and $\bar \pi^{0, s}$, we have by induction starting with $\bar \mu^{0, s} = \mu^0_0 = \mu^{\bar \pi^{0, s}_{0:t}}_0$,
\begin{align*}
    \bar \mu^{0, s}_{t+1}(x^0_{t+1}, x^0_{0:t}, u^0_{0:t}) &= P^0(x^0_{t+1} \mid x^0_t, u^0_t) \bar \mu^{0, s}_t(x^0_t, x^0_{0:t-1}, u^0_{0:t-1}) \bar \pi^{0, s}_t(u^0_t \mid x^0_t, x^0_{0:t-1}, u^0_{0:t-1}) \\ 
    &= P^0(x^0_{t+1} \mid x^0_t, u^0_t) \mu^{\bar \pi^{0, s}_{0:t-1}}_t(x^0_t, x^0_{0:t-1}, u^0_{0:t-1}) \bar \pi^{0, s}_t(u^0_t \mid x^0_t, x^0_{0:t-1}, u^0_{0:t-1}) \\
    &= \mu^{\bar \pi^{0, s}_{0:t}}_{t+1}(x^0_{t+1}, x^0_{0:t}, u^0_{0:t})
\end{align*}
which is the desired result, i.e. the MF generated by $\bar \pi^0$ is the same as the average MF of best responses.

For the minor players, the proof is analogous: The conditional probabilities of any policy $\pi^{BR, s}$ are always given by
\begin{align*}
    \mu^{\pi^{BR, s}_{0:t}}_{t+1 \mid x^0_{0:t}, u^0_{0:t}}(x_{t+1}) &= \mathbb P_{\pi^{BR, s}}(x^0_{t+1} \mid x^0_{0:t}, u^0_{0:t}) \\
    &= \sum_{x_t,u_t\in \mathcal{X}\times\mathcal{U}} P(x_{t+1} \mid x_t,u_t, x^0_t, u^0_t) \mu^{\pi^{BR, s}_{0:t}}_{t \mid x^0_{0:t-1}, u^0_{0:t-1}}(x_t) \pi^{BR, s}_t(u_t|x_t, x^0_{0:t}, u^0_{0:t-1})
\end{align*}
and therefore, by integrating over all times $\tau$,
\begin{align*}
    \mu^{\pi^{BR, s}_{0:t}}_{t+1 \mid x^0_{0:t}, u^0_{0:t}}(x_{t+1}) &= \sum_{x_t,u_t\in \mathcal{X}\times\mathcal{U}} P(x_{t+1} \mid x_t,u_t, x^0_t, u^0_t) \frac 1 \tau \int_0^\tau \mu^{\pi^{BR, s}_{0:t}}_{t \mid x^0_{0:t-1}, u^0_{0:t-1}}(x_t) \pi^{BR, s}_t(u_t|x_t, x^0_{0:t}, u^0_{0:t-1}) ds
\end{align*}
we obtain
\begin{align*}
    \bar \mu^{s}_{t+1 \mid x^0_{0:t}, u^0_{0:t}}(x_{t+1}) &= \sum_{x_t,u_t\in \mathcal{X}\times\mathcal{U}} P(x_{t+1} \mid x_t,u_t, x^0_t, u^0_t) \mu^{\bar \pi^{s}_{0:t}}_{t \mid x^0_{0:t-1}, u^0_{0:t-1}}(x_t) \bar \pi^{s}_t(u_t|x_t, x^0_{0:t}, u^0_{0:t-1}) ds \\
    &= \mu^{\bar \pi^{s}_{0:t}}_{t+1 \mid x^0_{0:t}, u^0_{0:t}}(x_{t+1})
\end{align*}
which again implies the desired result for minor agents by induction.
\end{proof}

\begin{property} \label{prop:2}
At all times $t$, $\tau$ and state-actions $x^0_{0:t-1}, u^0_{0:t-1}, x^0_t$, we have
\begin{align} \label{eq:2.1}
    \frac{d}{d\tau} \mu^{\bar \pi^{0, \tau}_{0:t-1}}_t(x^0_t, x^0_{0:t-1}, u^0_{0:t-1}) 
    = \frac 1 \tau \Big[ \mu^{\pi^{0, BR, \tau}_{0:t-1}}_t(x^0_t, x^0_{0:t-1}, u^0_{0:t-1}) - \mu^{\bar \pi^{0, \tau}_{0:t-1}}_t(x^0_t, x^0_{0:t-1}, u^0_{0:t-1}) \Big]
\end{align}
\begin{multline} \label{eq:2.2}
    \mu^{\bar \pi^{0, \tau}_{0:t-1}}_t(x^0_t, x^0_{0:t-1}, u^0_{0:t-1}) \frac{d}{d\tau} \bar \pi^{0, \tau}_t(u^0_t \mid x^0_t, x^0_{0:t-1}, u^0_{0:t-1}) \\
    = \mu^{\pi^{0, BR, \tau}_{0:t-1}}_t(x^0_t, x^0_{0:t-1}, u^0_{0:t-1}) \frac 1 \tau \Big[ \pi^{0, BR, \tau}_t(u^0_t \mid x^0_t, x^0_{0:t-1}, u^0_{0:t-1}) - \bar \pi^{0, \tau}_t(u^0_t \mid x^0_t, x^0_{0:t-1}, u^0_{0:t-1}) \Big].
\end{multline}
and also for the minor players 
\begin{align} \label{eq:2.3}
    \frac{d}{d\tau} \mu^{\bar \pi^{\tau}_{0:t-1}}_{t \mid x^0_{0:t-1}, u^0_{0:t-1}}(x_t) = \frac 1 \tau \Big[ \mu^{\pi^{BR, \tau}_{0:t-1}}_{t \mid x^0_{0:t-1}, u^0_{0:t-1}}(x_t) - \mu^{\bar \pi^{\tau}_{0:t-1}}_{t \mid x^0_{0:t-1}, u^0_{0:t-1}}(x_t) \Big] 
\end{align}
\begin{multline} \label{eq:2.4}
    \mu^{\bar \pi^{\tau}_{0:t-1}}_{t \mid x^0_{0:t-1}, u^0_{0:t-1}}(x_t) \frac{d}{d\tau} \bar \pi^\tau_t(u_t \mid x_t, x^0_{0:t}, u^0_{0:t-1}) \\
    = \mu^{\pi^{BR, \tau}_{0:t-1}}_{t \mid x^0_{0:t-1}, u^0_{0:t-1}}(x_t) \frac 1 \tau \Big[ \pi^{BR, \tau}_t(u_t \mid x_t, x^0_{0:t}, u^0_{0:t-1}) - \bar \pi^\tau_t(u_t \mid x_t, x^0_{0:t}, u^0_{0:t-1}) \Big].
\end{multline}
\end{property}
\begin{proof}
The properties are shown together inductively. First, note that \eqref{eq:2.1} implies \eqref{eq:2.2}: Start with Property~\ref{prop:1}, giving
\begin{align*}
    \bar \mu^{0, s}_t(x^0_t, x^0_{0:t-1}, u^0_{0:t-1}) \bar \pi^{0, \tau}_t(x^0, x^0_{0:t-1}, u^0_{0:t-1}) = \frac 1 \tau \int_0^\tau \pi^{0, BR, s}(x^0, x^0_{0:t-1}, u^0_{0:t-1}) \mu^{\pi^{0, BR, s}_{0:t-1}}_t(x^0_t, x^0_{0:t-1}, u^0_{0:t-1}) ds
\end{align*}
which implies by taking the derivative $\frac{d}{d\tau}$ that
\begin{align*}
    &\frac{d}{d\tau} \bar \mu^{0, s}_t(x^0_t, x^0_{0:t-1}, u^0_{0:t-1}) \bar \pi^{0, \tau}_t(x^0, x^0_{0:t-1}, u^0_{0:t-1}) + \bar \mu^{0, s}_t(x^0_t, x^0_{0:t-1}, u^0_{0:t-1}) \frac{d}{d\tau} \bar \pi^{0, \tau}_t(x^0, x^0_{0:t-1}, u^0_{0:t-1}) \\
    &= -\frac{1}{\tau^2} \int_0^\tau \pi^{0, BR, s}(x^0, x^0_{0:t-1}, u^0_{0:t-1}) \mu^{\pi^{0, BR, s}_{0:t-1}}_t(x^0_t, x^0_{0:t-1}, u^0_{0:t-1}) ds \\
    &\quad + \frac 1 \tau \Big[ \pi^{0, BR, s}(x^0, x^0_{0:t-1}, u^0_{0:t-1}) \mu^{\pi^{0, BR, s}_{0:t-1}}_t(x^0_t, x^0_{0:t-1}, u^0_{0:t-1}) ds \Big].
\end{align*}
Applying \eqref{eq:2.1} and Property~\ref{prop:1} gives
\begin{align*}
    &\bar \mu^{0, s}_t(x^0_t, x^0_{0:t-1}, u^0_{0:t-1}) \frac{d}{d\tau} \bar \pi^{0, \tau}_t(x^0, x^0_{0:t-1}, u^0_{0:t-1}) \\
    &= -\frac{1}{\tau^2} \int_0^\tau \pi^{0, BR, s}(x^0, x^0_{0:t-1}, u^0_{0:t-1}) \mu^{\pi^{0, BR, s}_{0:t-1}}_t(x^0_t, x^0_{0:t-1}, u^0_{0:t-1}) ds \\
    &\quad + \frac 1 \tau \Big[ \pi^{0, BR, s}(x^0, x^0_{0:t-1}, u^0_{0:t-1}) \mu^{\pi^{0, BR, s}_{0:t-1}}_t(x^0_t, x^0_{0:t-1}, u^0_{0:t-1}) ds \Big] \\
    &\quad - \frac 1 \tau \Big[ \mu^{\pi^{0, BR, \tau}_{0:t-1}}_t(x^0_t, x^0_{0:t-1}, u^0_{0:t-1}) - \mu^{\bar \pi^{0, \tau}_{0:t-1}}_t(x^0_t, x^0_{0:t-1}, u^0_{0:t-1}) \Big] \bar \pi^{0, \tau}_t(x^0, x^0_{0:t-1}, u^0_{0:t-1}) \\ 
    &= \mu^{\pi^{0, BR, \tau}_{0:t-1}}_t(x^0_t, x^0_{0:t-1}, u^0_{0:t-1}) \frac 1 \tau \Big[ \bar \pi^{0, \tau}_t(u^0_t \mid x^0_t, x^0_{0:t-1}, u^0_{0:t-1}) - \pi^{0, BR, \tau}_t(u^0_t \mid x^0_t, x^0_{0:t-1}, u^0_{0:t-1}) \Big].
\end{align*}

Now to show the property \eqref{eq:2.1} at all times $t$, we use induction. At time $0$, the property is trivially fulfilled by fixed $\mu^0_0$. Assume the property \eqref{eq:2.1} and therefore \eqref{eq:2.2} holds at time $t$, then for the induction step, at time $t+1$, by Property~\ref{prop:1}, we have
\begin{align*}
    &\frac{d}{d\tau} \mu^{\bar \pi^{0, \tau}_{0:t}}_{t+1}(x^0_{t+1}, x^0_{0:t}, u^0_{0:t}) \\
    &= \frac{d}{d\tau} \Big[ \mu^{\bar \pi^{0, \tau}_{0:t-1}}_t(x^0_t, x^0_{0:t-1}, u^0_{0:t-1}) \bar \pi^{0, \tau}_t(u^0_t \mid x^0_t, x^0_{0:t-1}, u^0_{0:t-1}) P^0(x^0_{t+1} \mid x^0_t, u^0_t) \Big] \\
    &= \frac{d}{d\tau} \mu^{\bar \pi^{0, \tau}_{0:t-1}}_t(x^0_t, x^0_{0:t-1}, u^0_{0:t-1}) \bar \pi^{0, \tau}_t(u^0_t \mid x^0_t, x^0_{0:t-1}, u^0_{0:t-1}) P^0(x^0_{t+1} \mid x^0_t, u^0_t) \\
    &\quad + \mu^{\bar \pi^{0, \tau}_{0:t-1}}_t(x^0_t, x^0_{0:t-1}, u^0_{0:t-1}) \frac{d}{d\tau} \bar \pi^{0, \tau}_t(u^0_t \mid x^0_t, x^0_{0:t-1}, u^0_{0:t-1}) P^0(x^0_{t+1} \mid x^0_t, u^0_t) \\
    &= \frac 1 \tau \Big[ \mu^{\pi^{0, BR, \tau}_{0:t-1}}_t(x^0_t, x^0_{0:t-1}, u^0_{0:t-1}) - \mu^{\bar \pi^{0, \tau}_{0:t-1}}_t(x^0_t, x^0_{0:t-1}, u^0_{0:t-1}) \Big] \bar \pi^{0, \tau}_t(u^0_t \mid x^0_t, x^0_{0:t-1}, u^0_{0:t-1}) P^0(x^0_{t+1} \mid x^0_t, u^0_t) \\
    &\quad + \mu^{\pi^{0, BR, \tau}_{0:t-1}}_t(x^0_t, x^0_{0:t-1}, u^0_{0:t-1}) \frac 1 \tau \Big[ \pi^{0, BR, \tau}_t(u^0_t \mid x^0_t, x^0_{0:t-1}, u^0_{0:t-1}) - \bar \pi^{0, \tau}_t(u^0_t \mid x^0_t, x^0_{0:t-1}, u^0_{0:t-1}) \Big] P^0(x^0_{t+1} \mid x^0_t, u^0_t) \\
    &= \frac 1 \tau \mu^{\pi^{0, BR, \tau}_{0:t-1}}_t(x^0_t, x^0_{0:t-1}, u^0_{0:t-1}) \pi^{0, BR, \tau}_t(u^0_t \mid x^0_t, x^0_{0:t-1}, u^0_{0:t-1}) P^0(x^0_{t+1} \mid x^0_t, u^0_t) \\
    &\quad - \frac 1 \tau \mu^{\bar \pi^{0, \tau}_{0:t-1}}_t(x^0_t, x^0_{0:t-1}, u^0_{0:t-1}) \bar \pi^{0, \tau}_t(u^0_t \mid x^0_t, x^0_{0:t-1}, u^0_{0:t-1}) P^0(x^0_{t+1} \mid x^0_t, u^0_t) \\
    &= \frac 1 \tau \Big[ \mu^{\pi^{0, BR, \tau}_{0:t}}_{t+1}(x^0_{t+1}, x^0_{0:t}, u^0_{0:t}) - \mu^{\bar \pi^{0, \tau}_{0:t}}_{t+1}(x^0_{t+1}, x^0_{0:t}, u^0_{0:t}) \Big]
\end{align*}
where we used the induction assumption on $\mu^{\bar \pi^{0, \tau}_{0:t-1}}_t(x^0_t, x^0_{0:t-1}, u^0_{0:t-1}) \frac{d}{d\tau} \bar \pi^{0, \tau}_t(u^0_t \mid x^0_t, x^0_{0:t-1}, u^0_{0:t-1})$ to obtain \eqref{eq:2.1}.

For the minor players, the proof is analogous in that \eqref{eq:2.3} implies \eqref{eq:2.4} by Property~\ref{prop:1}. Meanwhile, \eqref{eq:2.3} follows readily by noting $\bar \mu^{s}_t = \mu^{\bar \pi^{s}_{0:t-1}}_t$ by Property~\ref{prop:1} and taking instead the derivative of the definition of $\bar \mu^{s}_t$.
\end{proof}

\begin{proof}[Proof of Theorem~\ref{the:fictitiousplayconvergence}]
With the above properties established, we can show the convergence of exploitabilities to zero.

\paragraph{Step 1: Focusing on the exploitability of the minor player.} We first start with showing that the exploitability of minor players, $\mathcal E(\bar{\pi}^\tau, \bar{\pi}^{0,\tau})$ is a strong Lyapunov function. Using the definition of exploitability of the minor player, we can write:
\begin{align*}
    &\frac{d}{d\tau}\mathcal E(\bar \pi^\tau, \bar \pi^{0,\tau}) \\
    &= \frac{d}{d\tau} \Big[ \max_{\pi'} J(\pi', \bar \pi^\tau, \bar \pi^{0,\tau}) - J(\bar \pi^\tau, \bar \pi^\tau, \bar \pi^{0,\tau})\Big] \\
    &= \frac{d}{d\tau} \sum_{t \in \mathcal{T}} \sum_{\substack{x_t \in \mathcal X, u_t \in \mathcal X, \\x^0_{0:t} \in {\mathcal X^0}^{t+1}, \\u^0_{0:t-1} \in {\mathcal U^0}^{t}}} \Big[ \mathbb P_{\pi^{BR, \tau}, \bar \pi^\tau, \bar \pi^{0,\tau}}(x_t, u_t, x^0_{0:t}, u^0_{0:t-1}) - \mathbb P_{\bar \pi^\tau, \bar \pi^\tau, \bar \pi^{0,\tau}}(x_t, u_t, x^0_{0:t}, u^0_{0:t-1}) \Big] r(x_t, u_t, x^0_t, \mu^{\bar \pi^{\tau}_{0:t-1}}_{t \mid x^0_{0:t-1}, u^0_{0:t-1}})
\end{align*}
and, e.g., at time $t$ for the FP policy
\begin{align*}
    &\frac{d}{d\tau} \sum_{\substack{x_t \in \mathcal X, u_t \in \mathcal X, \\x^0_{0:t} \in {\mathcal X^0}^{t+1}, \\u^0_{0:t-1} \in {\mathcal U^0}^{t}}} \mathbb P_{\bar \pi^\tau, \bar \pi^\tau, \bar \pi^{0,\tau}}(x_t, u_t, x^0_{0:t}, u^0_{0:t-1}) r(x_t, x^0_t, \mu^{\bar \pi^{\tau}_{0:t-1}}_{t \mid x^0_{0:t-1}, u^0_{0:t-1}}) \\
    &= \frac{d}{d\tau} \sum_{t \in \mathcal{T}} \sum_{\substack{x_t \in \mathcal X, u_t \in \mathcal X, \\x^0_{0:t} \in {\mathcal X^0}^{t+1}, \\u^0_{0:t-1} \in {\mathcal U^0}^{t}}} \mu^{\bar \pi^{0, \tau}_{0:t-1}}_t(x^0_t, x^0_{0:t-1}, u^0_{0:t-1}) \mu^{\bar \pi^{\tau}_{0:t-1}}_{t \mid x^0_{0:t-1}, u^0_{0:t-1}}(x_t) \bar \pi^\tau_t(u_t \mid x_t, x^0_{0:t}, u^0_{0:t-1}) r(x_t, u_t, x^0_t, \mu^{\bar \pi^{\tau}_{0:t-1}}_{t \mid x^0_{0:t-1}, u^0_{0:t-1}})
\end{align*}
and similarly for the BR policy. By the envelope theorem on $\max_{\pi'} J^0(\bar \pi^\tau, \pi')$, the partial derivative with respect to $\pi^{BR}$ can be dropped. \footnote{To be precise, here and in \citet[Appendix~A]{perrin2020fictitious}, continuous differentiability of the objectives with respect to $\tau$ is implicitly assumed, which if needed may be guaranteed by introducing a minimal regularization. This is not a problem however, as for example entropy-regularized MFNE still achieve arbitrarily good unregularized finite game equilibria, see e.g. \citet[Theorem~4]{cui2021approximately}.}
Therefore, we obtain
\normalsize
\begin{align*}
    &\frac{d}{d\tau} \sum_{t \in \mathcal{T}} \sum_{\substack{x_t \in \mathcal X, u_t \in \mathcal X, \\x^0_{0:t} \in {\mathcal X^0}^{t+1}, \\u^0_{0:t-1} \in {\mathcal U^0}^{t}}} \Big[ \mathbb P_{\pi^{BR, \tau}, \bar \pi^\tau, \bar \pi^{0,\tau}}(x_t, u_t, x^0_{0:t}, u^0_{0:t-1}) - \mathbb P_{\bar \pi^\tau, \bar \pi^\tau, \bar \pi^{0,\tau}}(x_t, u_t, x^0_{0:t}, u^0_{0:t-1}) \Big] r(x_t, u_t, x^0_t, \mu^{\bar \pi^{\tau}_{0:t-1}}_{t \mid x^0_{0:t-1}, u^0_{0:t-1}}) \\
    &= \sum_{t \in \mathcal{T}} \sum_{\substack{x_t \in \mathcal X, u_t \in \mathcal X, \\x^0_{0:t} \in {\mathcal X^0}^{t+1}, \\u^0_{0:t-1} \in {\mathcal U^0}^{t}}} \Bigg[ \\
    &\quad \frac{d}{d\tau} \mu^{\bar \pi^{0, \tau}_{0:t-1}}_t(x^0_t, x^0_{0:t-1}, u^0_{0:t-1}) \mu^{\pi^{BR, \tau}_{0:t-1}}_{t \mid x^0_{0:t-1}, u^0_{0:t-1}}(x_t) \pi^{BR, \tau}_t(u_t \mid x_t, x^0_{0:t}, u^0_{0:t-1}) r(x_t, u_t, x^0_t, \mu^{\bar \pi^{\tau}_{0:t-1}}_{t \mid x^0_{0:t-1}, u^0_{0:t-1}}) \\
    &\quad + \mu^{\bar \pi^{0, \tau}_{0:t-1}}_t(x^0_t, x^0_{0:t-1}, u^0_{0:t-1}) \mu^{\pi^{BR, \tau}_{0:t-1}}_{t \mid x^0_{0:t-1}, u^0_{0:t-1}}(x_t) \pi^{BR, \tau}_t(u_t \mid x_t, x^0_{0:t}, u^0_{0:t-1}) \frac{d}{d\tau} r(x_t, u_t, x^0_t, \mu^{\bar \pi^{\tau}_{0:t-1}}_{t \mid x^0_{0:t-1}, u^0_{0:t-1}}) \\
    &\quad - \frac{d}{d\tau} \mu^{\bar \pi^{0, \tau}_{0:t-1}}_t(x^0_t, x^0_{0:t-1}, u^0_{0:t-1}) \mu^{\bar \pi^{\tau}_{0:t-1}}_{t \mid x^0_{0:t-1}, u^0_{0:t-1}}(x_t) \bar \pi^\tau_t(u_t \mid x_t, x^0_{0:t}, u^0_{0:t-1}) r(x_t, u_t, x^0_t, \mu^{\bar \pi^{\tau}_{0:t-1}}_{t \mid x^0_{0:t-1}, u^0_{0:t-1}}) \\
    &\quad - \mu^{\bar \pi^{0, \tau}_{0:t-1}}_t(x^0_t, x^0_{0:t-1}, u^0_{0:t-1}) \frac{d}{d\tau} \mu^{\bar \pi^{\tau}_{0:t-1}}_{t \mid x^0_{0:t-1}, u^0_{0:t-1}}(x_t) \bar \pi^\tau_t(u_t \mid x_t, x^0_{0:t}, u^0_{0:t-1}) r(x_t, u_t, x^0_t, \mu^{\bar \pi^{\tau}_{0:t-1}}_{t \mid x^0_{0:t-1}, u^0_{0:t-1}}) \\
    &\quad - \mu^{\bar \pi^{0, \tau}_{0:t-1}}_t(x^0_t, x^0_{0:t-1}, u^0_{0:t-1}) \mu^{\bar \pi^{\tau}_{0:t-1}}_{t \mid x^0_{0:t-1}, u^0_{0:t-1}}(x_t) \frac{d}{d\tau} \bar \pi^\tau_t(u_t \mid x_t, x^0_{0:t}, u^0_{0:t-1}) r(x_t, u_t, x^0_t, \mu^{\bar \pi^{\tau}_{0:t-1}}_{t \mid x^0_{0:t-1}, u^0_{0:t-1}}) \\
    &\quad - \mu^{\bar \pi^{0, \tau}_{0:t-1}}_t(x^0_t, x^0_{0:t-1}, u^0_{0:t-1}) \mu^{\bar \pi^{\tau}_{0:t-1}}_{t \mid x^0_{0:t-1}, u^0_{0:t-1}}(x_t) \bar \pi^\tau_t(u_t \mid x_t, x^0_{0:t}, u^0_{0:t-1}) \frac{d}{d\tau} r(x_t, u_t, x^0_t, \mu^{\bar \pi^{\tau}_{0:t-1}}_{t \mid x^0_{0:t-1}, u^0_{0:t-1}}) \Bigg]
\end{align*}
\normalsize
where for the first and third term, we use the property from major player, while for fourth term we use \eqref{eq:2.3} in Property~\ref{prop:2}, and for the fifth we use \eqref{eq:2.4} in Property~\ref{prop:2}.

Therefore, combining the first and third, second and last term, and fourth and fifth terms, we have
\small
\begin{align*}
    &\frac{d}{d\tau}\mathcal E(\bar \pi^\tau, \bar \pi^{0,\tau}) \\
    &= \sum_{t \in \mathcal{T}} \sum_{\substack{x_t \in \mathcal X, u_t \in \mathcal X, \\x^0_{0:t} \in {\mathcal X^0}^{t+1}, \\u^0_{0:t-1} \in {\mathcal U^0}^{t}}} \Bigg[ \\
    &\quad \frac 1 \tau \Big[ \mu^{\pi^{0, BR, \tau}_{0:t-1}}_t(x^0_t, x^0_{0:t-1}, u^0_{0:t-1}) - \mu^{\bar \pi^{0, \tau}_{0:t-1}}_t(x^0_t, x^0_{0:t-1}, u^0_{0:t-1}) \Big] r(x_t, u_t, x^0_t, \mu^{\bar \pi^{\tau}_{0:t-1}}_{t \mid x^0_{0:t-1}, u^0_{0:t-1}}) \\
    &\qquad \cdot \Big[ \mu^{\pi^{BR, \tau}_{0:t-1}}_{t \mid x^0_{0:t-1}, u^0_{0:t-1}}(x_t) \pi^{BR, \tau}_t(u_t \mid x_t, x^0_{0:t}, u^0_{0:t-1}) - \mu^{\bar \pi^{\tau}_{0:t-1}}_{t \mid x^0_{0:t-1}, u^0_{0:t-1}}(x_t) \bar \pi^\tau_t(u_t \mid x_t, x^0_{0:t}, u^0_{0:t-1}) \Big] \\
    &\quad + \mu^{\bar \pi^{0, \tau}_{0:t-1}}_t(x^0_t, x^0_{0:t-1}, u^0_{0:t-1}) \frac \tau \tau \Big[ \mu^{\pi^{BR, \tau}_{0:t-1}}_{t \mid x^0_{0:t-1}, u^0_{0:t-1}}(x_t) - \mu^{\bar \pi^{\tau}_{0:t-1}}_{t \mid x^0_{0:t-1}, u^0_{0:t-1}}(x_t) \Big] \Big\langle \nabla_\mu r(x_t, u_t, x^0_t, \mu^{\bar \pi^{\tau}_{0:t-1}}_{t \mid x^0_{0:t-1}, u^0_{0:t-1}}), \frac{d}{d\tau} \mu^{\bar \pi^{\tau}_{0:t-1}}_{t \mid x^0_{0:t-1}, u^0_{0:t-1}}) \Big\rangle \\
    &\quad - \frac 1 \tau \mu^{\bar \pi^{0, \tau}_{0:t-1}}_t(x^0_t, x^0_{0:t-1}, u^0_{0:t-1}) r(x_t, u_t, x^0_t, \mu^{\bar \pi^{\tau}_{0:t-1}}_{t \mid x^0_{0:t-1}, u^0_{0:t-1}}) \\
    &\qquad \cdot \Big[ \mu^{\pi^{BR, \tau}_{0:t-1}}_{t \mid x^0_{0:t-1}, u^0_{0:t-1}}(x_t) \pi^{BR, \tau}_t(u_t \mid x_t, x^0_{0:t}, u^0_{0:t-1}) - \mu^{\bar \pi^{\tau}_{0:t-1}}_{t \mid x^0_{0:t-1}, u^0_{0:t-1}}(x_t) \bar \pi^\tau_t(u_t \mid x_t, x^0_{0:t}, u^0_{0:t-1}) \Big] \Bigg] \\
    &= \frac 1 \tau \tilde {\mathcal E}(\bar \pi^\tau, \pi^{0, BR, \tau}, \bar \pi^{0,\tau}) - \frac 1 \tau \mathcal E(\bar \pi^\tau, \bar \pi^{0,\tau}) \\
    &\quad + \sum_{t \in \mathcal{T}} \sum_{\substack{x_t \in \mathcal X, u_t \in \mathcal X, \\x^0_{0:t} \in {\mathcal X^0}^{t+1}, \\u^0_{0:t-1} \in {\mathcal U^0}^{t}}} \tau \mu^{\bar \pi^{0, \tau}_{0:t-1}}_t(x^0_t, x^0_{0:t-1}, u^0_{0:t-1}) \frac{d}{d\tau} \mu^{\bar \pi^{\tau}_{0:t-1}}_{t \mid x^0_{0:t-1}, u^0_{0:t-1}} \Big\langle \nabla_\mu r(x_t, u_t, x^0_t, \mu^{\bar \pi^{\tau}_{0:t-1}}_{t \mid x^0_{0:t-1}, u^0_{0:t-1}}), \frac{d}{d\tau} \mu^{\bar \pi^{\tau}_{0:t-1}}_{t \mid x^0_{0:t-1}, u^0_{0:t-1}}) \Big\rangle \Bigg] \\
    &\quad - \frac 1 \tau \mathcal E(\bar \pi^\tau, \bar \pi^{0,\tau}) \leq - \frac 1 \tau \mathcal E(\bar \pi^\tau, \bar \pi^{0,\tau})
\end{align*}
\normalsize
where we use monotonicity and Assumption~\ref{as:fictitious_play_conv} to obtain the last inequality.

\paragraph{Step 2: Focusing on the exploitability of the major player.} Similarly to the case of minor players, we can write: 
\begin{align*}
    &\frac{d}{d\tau}\mathcal E^0(\bar \pi^\tau, \bar \pi^{0,\tau}) \\
    &= \frac{d}{d\tau} \Big[ \max_{\pi'} J^0(\bar \pi^\tau, \pi') - J^0(\bar \pi^\tau, \bar \pi^{0,\tau})\Big] \\
    &= \frac{d}{d\tau} \sum_{t \in \mathcal{T}} \sum_{x^0_{0:t} \in {\mathcal X^0}^{t+1}, u^0_{0:t} \in {\mathcal U^0}^{t+1}} \Big[ \mathbb P_{\bar \pi^{\tau}, \pi^{0, BR, \tau}}(x^0_{0:t}, u^0_{0:t}) - \mathbb P_{\bar \pi^{\tau}, \bar \pi^{0, \tau}}(x^0_{0:t}, u^0_{0:t}) \Big] r^0(x^0_t, u^0_t, \mu^{\bar \pi^{\tau}_{0:t-1}}_{t \mid x^0_{0:t-1}, u^0_{0:t-1}}) \\
    &= \frac{d}{d\tau} \sum_{x^0_{0} \in {\mathcal X^0}, u^0_{0} \in {\mathcal U^0}} \Big[ \mathbb P_{\bar \pi^{\tau}, \pi^{0, BR, \tau}}(x^0_{0}, u^0_{0}) - \mathbb P_{\bar \pi^{\tau}, \bar \pi^{0, \tau}}(x^0_{0}, u^0_{0}) \Big] r^0(x^0_0, u^0_{0}, \mu_0) \\
    &\quad + \frac{d}{d\tau} \sum_{x^0_{0:1} \in {\mathcal X^0}^2, u^0_{0:1} \in {\mathcal U^0}^2} \Big[ \mathbb P_{\bar \pi^{\tau}, \pi^{0, BR, \tau}}(x^0_{0:1}, u^0_{0:1}) - \mathbb P_{\bar \pi^{\tau}, \bar \pi^{0, \tau}}(x^0_{0:1}, u^0_{0:1}) \Big] r^0(x^0_1, u^0_1, \mu^{\bar \pi^{\tau}_{0}}_{1 \mid x^0_{0}, u^0_{0}}) \\
    &\quad + \frac{d}{d\tau} \sum_{x^0_{0:2} \in {\mathcal X^0}^{3}, u^0_{0:2} \in {\mathcal U^0}^{3}} \Big[ \mathbb P_{\bar \pi^{\tau}, \pi^{0, BR, \tau}}(x^0_{0:2}, u^0_{0:2}) - \mathbb P_{\bar \pi^{\tau}, \bar \pi^{0, \tau}}(x^0_{0:2}, u^0_{0:2}) \Big] r^0(x^0_2, u^0_2, \mu^{\bar \pi^{\tau}_{0:1}}_{2 \mid x^0_{0:1}, u^0_{0:1}}) \\
    &\quad + ...
\end{align*}
Generally, at all times $t$, we have
\begin{align*}
    &\frac{d}{d\tau} \sum_{x^0_{0:t} \in {\mathcal X^0}^{t+1}, u^0_{0:t} \in {\mathcal U^0}^{t+1}} \mathbb P_{\bar \pi^{\tau}, \bar \pi^{0, \tau}}(x^0_{0:t+1}, u^0_{0:t+1}) r^0(x^0_{t+1}, u^0_{t+1}, \mu^{\bar \pi^{\tau}_{0:t}}_{t+1 \mid x^0_{0:t}, u^0_{0:t}}) \\
    &= \frac{d}{d\tau} \sum_{x^0_{0:t} \in {\mathcal X^0}^{t+1}, u^0_{0:t} \in {\mathcal U^0}^{t+1}} \mu^{\bar \pi^{0, \tau}_{0:t}}_{t+1}(x^0_{t+1}, x^0_{0:t}, u^0_{0:t}) \bar \pi^{0, \tau}_{t+1}(u^0_{t+1} \mid x^0_{t+1}, x^0_{0:t}, u^0_{0:t}) r^0(x^0_{t+1}, u^0_{t+1}, \mu^{\bar \pi^{\tau}_{0:t}}_{t+1 \mid x^0_{0:t}, u^0_{0:t}}) \\
    &= \sum_{x^0_{0:t} \in {\mathcal X^0}^{t+1}, u^0_{0:t} \in {\mathcal U^0}^{t+1}} \frac{d}{d\tau} \mu^{\bar \pi^{0, \tau}_{0:t}}_{t+1}(x^0_{t+1}, x^0_{0:t}, u^0_{0:t}) \bar \pi^{0, \tau}_{t+1}(u^0_{t+1} \mid x^0_{t+1}, x^0_{0:t}, u^0_{0:t}) r^0(x^0_{t+1}, u^0_{t+1}, \mu^{\bar \pi^{\tau}_{0:t}}_{t+1 \mid x^0_{0:t}, u^0_{0:t}}) \\
    &\quad + \sum_{x^0_{0:t} \in {\mathcal X^0}^{t+1}, u^0_{0:t} \in {\mathcal U^0}^{t+1}} \mu^{\bar \pi^{0, \tau}_{0:t}}_{t+1}(x^0_{t+1}, x^0_{0:t}, u^0_{0:t}) \frac{d}{d\tau} \bar \pi^{0, \tau}_{t+1}(u^0_{t+1} \mid x^0_{t+1}, x^0_{0:t}, u^0_{0:t}) r^0(x^0_{t+1}, u^0_{t+1}, \mu^{\bar \pi^{\tau}_{0:t}}_{t+1 \mid x^0_{0:t}, u^0_{0:t}}) \\
    &\quad + \sum_{x^0_{0:t} \in {\mathcal X^0}^{t+1}, u^0_{0:t} \in {\mathcal U^0}^{t+1}} \mu^{\bar \pi^{0, \tau}_{0:t}}_{t+1}(x^0_{t+1}, x^0_{0:t}, u^0_{0:t}) \bar \pi^{0, \tau}_{t+1}(u^0_{t+1} \mid x^0_{t+1}, x^0_{0:t}, u^0_{0:t}) \frac{d}{d\tau} r^0(x^0_{t+1}, u^0_{t+1}, \mu^{\bar \pi^{\tau}_{0:t}}_{t+1 \mid x^0_{0:t}, u^0_{0:t}})
\end{align*}
This leads to the desired result
\begin{align*}
    &\frac{d}{d\tau} \sum_{x^0_{0:t} \in {\mathcal X^0}^{t+1}, u^0_{0:t} \in {\mathcal U^0}^{t+1}} \Big[ \mathbb P_{\bar \pi^{\tau}, \pi^{0, BR, \tau}}(x^0_{0:t+1}, u^0_{0:t+1}) - \mathbb P_{\bar \pi^{\tau}, \bar \pi^{0, \tau}}(x^0_{0:t+1}, u^0_{0:t+1}) \Big] r^0(x^0_{t+1}, u^0_{t+1}, \mu^{\bar \pi^{\tau}_{0:t}}_{t+1 \mid x^0_{0:t}, u^0_{0:t}}) \\
    &= \sum_{x^0_{0:t} \in {\mathcal X^0}^{t+1}, u^0_{0:t} \in {\mathcal U^0}^{t+1}} \mu^{\pi^{0, BR, \tau}_{0:t-1}}_{t+1}(x^0_{t+1}, x^0_{0:t}, u^0_{0:t}) \pi^{0, BR, \tau}_{t+1}(u^0_{t+1} \mid x^0_{t+1}, x^0_{0:t}, u^0_{0:t}) \frac{d}{d\tau} r^0(x^0_{t+1}, u^0_{t+1}, \mu^{\bar \pi^{\tau}_{0:t}}_{t+1 \mid x^0_{0:t}, u^0_{0:t}}) \\
    &\quad - \sum_{x^0_{0:t} \in {\mathcal X^0}^{t+1}, u^0_{0:t} \in {\mathcal U^0}^{t+1}} \frac{d}{d\tau} \mu^{\bar \pi^{0, \tau}_{0:t}}_{t+1}(x^0_{t+1}, x^0_{0:t}, u^0_{0:t}) \bar \pi^{0, \tau}_{t+1}(u^0_{t+1} \mid x^0_{t+1}, x^0_{0:t}, u^0_{0:t}) r^0(x^0_{t+1}, u^0_{t+1}, \mu^{\bar \pi^{\tau}_{0:t}}_{t+1 \mid x^0_{0:t}, u^0_{0:t}}) \\
    &\quad - \sum_{x^0_{0:t} \in {\mathcal X^0}^{t+1}, u^0_{0:t} \in {\mathcal U^0}^{t+1}} \mu^{\bar \pi^{0, \tau}_{0:t}}_{t+1}(x^0_{t+1}, x^0_{0:t}, u^0_{0:t}) \frac{d}{d\tau} \bar \pi^{0, \tau}_{t+1}(u^0_{t+1} \mid x^0_{t+1}, x^0_{0:t}, u^0_{0:t}) r^0(x^0_{t+1}, u^0_{t+1}, \mu^{\bar \pi^{\tau}_{0:t}}_{t+1 \mid x^0_{0:t}, u^0_{0:t}}) \\
    &\quad - \sum_{x^0_{0:t} \in {\mathcal X^0}^{t+1}, u^0_{0:t} \in {\mathcal U^0}^{t+1}} \mu^{\bar \pi^{0, \tau}_{0:t}}_{t+1}(x^0_{t+1}, x^0_{0:t}, u^0_{0:t}) \bar \pi^{0, \tau}_{t+1}(u^0_{t+1} \mid x^0_{t+1}, x^0_{0:t}, u^0_{0:t}) \frac{d}{d\tau} r^0(x^0_{t+1}, u^0_{t+1}, \mu^{\bar \pi^{\tau}_{0:t}}_{t+1 \mid x^0_{0:t}, u^0_{0:t}})
\end{align*}
by equating the middle two terms to the exploitability terms at time $t+1$, and analyzing the first and last term as 
\begin{align*}
    &\sum_{x^0_{0:t} \in {\mathcal X^0}^{t+1}, u^0_{0:t} \in {\mathcal U^0}^{t+1}} \mu^{\pi^{0, BR, \tau}_{0:t-1}}_{t+1}(x^0_{t+1}, x^0_{0:t}, u^0_{0:t}) \pi^{0, BR, \tau}_{t+1}(u^0_{t+1} \mid x^0_{t+1}, x^0_{0:t}, u^0_{0:t}) \frac{d}{d\tau} r^0(x^0_{t+1}, u^0_{t+1}, \mu^{\bar \pi^{\tau}_{0:t}}_{t+1 \mid x^0_{0:t}, u^0_{0:t}}) \\
    &\quad - \sum_{x^0_{0:t} \in {\mathcal X^0}^{t+1}, u^0_{0:t} \in {\mathcal U^0}^{t+1}} \mu^{\bar \pi^{0, \tau}_{0:t}}_{t+1}(x^0_{t+1}, x^0_{0:t}, u^0_{0:t}) \bar \pi^{0, \tau}_{t+1}(u^0_{t+1} \mid x^0_{t+1}, x^0_{0:t}, u^0_{0:t}) \frac{d}{d\tau} r^0(x^0_{t+1}, u^0_{t+1}, \mu^{\bar \pi^{\tau}_{0:t}}_{t+1 \mid x^0_{0:t}, u^0_{0:t}}) \\
    &= \sum_{x^0_{0:t} \in {\mathcal X^0}^{t+1}, u^0_{0:t} \in {\mathcal U^0}^{t+1}} \mu^{\pi^{0, BR, \tau}_{0:t-1}}_{t+1}(x^0_{t+1}, x^0_{0:t}, u^0_{0:t}) \Big\langle \nabla_\mu \bar r^0(x^0_{t+1}, \mu^{\bar \pi^{\tau}_{0:t}}_{t+1 \mid x^0_{0:t}, u^0_{0:t}}), \frac{d}{d\tau} \mu^{\bar \pi^{\tau}_{0:t}}_{t+1 \mid x^0_{0:t}, u^0_{0:t}} \Big\rangle \\
    &\quad - \sum_{x^0_{0:t} \in {\mathcal X^0}^{t+1}, u^0_{0:t} \in {\mathcal U^0}^{t+1}} \mu^{\bar \pi^{0, \tau}_{0:t}}_{t+1}(x^0_{t+1}, x^0_{0:t}, u^0_{0:t}) \Big\langle \nabla_\mu \bar r^0(x^0_{t+1}, \mu^{\bar \pi^{\tau}_{0:t}}_{t+1 \mid x^0_{0:t}, u^0_{0:t}}), \frac{d}{d\tau} \mu^{\bar \pi^{\tau}_{0:t}}_{t+1 \mid x^0_{0:t}, u^0_{0:t}} \Big\rangle \\
    &= \tau \sum_{x^0_{0:t} \in {\mathcal X^0}^{t+1}, u^0_{0:t} \in {\mathcal U^0}^{t+1}} \frac{d}{d\tau} \mu^{\bar \pi^{0, \tau}_{0:t}}_{t+1}(x^0_{t+1}, x^0_{0:t}, u^0_{0:t}) \Big\langle \nabla_\mu \bar r^0(x^0_{t+1}, \mu^{\bar \pi^{\tau}_{0:t}}_{t+1 \mid x^0_{0:t}, u^0_{0:t}}), \frac{d}{d\tau} \mu^{\bar \pi^{\tau}_{0:t}}_{t+1 \mid x^0_{0:t}, u^0_{0:t}} \Big\rangle \leq 0
\end{align*}
using Property~\ref{prop:2}, and that the term is non-positive by Assumption~\ref{as:fictitious_play_conv}.

Therefore, we have
\begin{align*}
    &\frac{d}{d\tau} \mathcal E^0(\bar \pi^\tau, \bar \pi^{0,\tau}) \leq - \frac 1 \tau \mathcal E^0(\bar \pi^\tau, \bar \pi^{0,\tau}).
\end{align*}

\vskip6pt
\noindent\textbf{Step 3: Combining the results.} In Step 1, we showed that $\frac{d}{d\tau}\mathcal E(\bar{\pi}^\tau, \bar{\pi}^{0,\tau}) \leq -\frac{1}{\tau}\mathcal E(\bar{\pi}^\tau, \bar{\pi}^{0,\tau})$ and in Step 2, we showed that $\frac{d}{d\tau}\mathcal E^0(\bar{\pi}^\tau, \bar{\pi}^{0,\tau})=-\frac{1}{\tau} \mathcal E^0(\bar{\pi}^\tau, \bar{\pi}^{0,\tau})$. Therefore we can conclude that
\begin{equation*}
    \frac{d}{d\tau}\left(\mathcal E(\bar{\pi}^\tau, \bar{\pi}^{0,\tau}) + \mathcal E^0(\bar{\pi}^\tau, \bar{\pi}^{0,\tau})\right)\leq -\frac{1}{\tau}\left(\mathcal E(\bar{\pi}^\tau, \bar{\pi}^{0,\tau}) + \mathcal E^0(\bar{\pi}^\tau, \bar{\pi}^{0,\tau})\right)
\end{equation*}
which shows the total exploitability of the system is a strong Lyapunov function and therefore it converges with rate $\frac{1}{\tau}$.
\end{proof}

\section{Continuity of MF Dynamics} \label{app:lem:Tcont}
In this section, we show continuity of the MF dynamics $T^\pi_t$, which will be used in the following proofs.
\begin{lemma} \label{lem:Tcont}
    Under Assumptions~\ref{ass:m3pcont} and \ref{ass:m3picont}, the transition operator $T^\pi_t$ is uniformly Lipschitz continuous with constant $L_T \coloneqq |\mathcal X| L_P + |\mathcal X| |\mathcal U| L_\Pi + |\mathcal X|^2 |\mathcal U| $.
\end{lemma}

\begin{proof}[Proof of Lemma~\ref{lem:Tcont}]
For any $(x^0, u^0, \mu), (x^{0\prime}, u^{0\prime}, \mu') \in \mathcal X^0 \times \mathcal U^0 \times \mathcal P(\mathcal X)$, we have
\begin{align*}
    &\left\lVert T^\pi_t(x^0, u^0, \mu) - T^\pi_t(x^{0\prime}, u^{0\prime}, \mu') \right\rVert \\
    &\quad = \sum_{x^* \in \mathcal X} \left| \iint P(x^* \mid x, u, x^0, u^0, \mu) \pi_t(\mathrm du \mid x, x^0, \mu) \mu(\mathrm dx) - \iint P(x^* \mid x, u, x^{0\prime}, u^{0\prime}, \mu') \pi_t(\mathrm du \mid x, x^{0\prime}, \mu') \mu'(\mathrm dx) \right| \\
    &\quad \leq \sum_{x^* \in \mathcal X} \iint \left| P(x^* \mid x, u, x^0, u^0, \mu) - P(x^* \mid x, u, x^{0\prime}, u^{0\prime}, \mu')  \right| \pi_t(\mathrm du \mid x, x^0, \mu) \mu(\mathrm dx) \\
    &\qquad + \sum_{x^* \in \mathcal X} \int \left| \int P(x^* \mid x, u, x^{0\prime}, u^{0\prime}, \mu') \left( \pi_t(\mathrm du \mid x, x^0, \mu) - \pi_t(\mathrm du \mid x, x^{0\prime}, \mu') \right) \right| \mu(\mathrm dx) \\
    &\qquad + \sum_{x^* \in \mathcal X} \left| \iint P(x^* \mid x, u, x^{0\prime}, u^{0\prime}, \mu') \pi_t(\mathrm du \mid x, x^{0\prime}, \mu') \left( \mu - \mu' \right) (\mathrm dx) \right| \\
    &\quad \leq \left( |\mathcal X| L_P + |\mathcal X| |\mathcal U| L_\Pi + |\mathcal X|^2 |\mathcal U| \right) d((x^0, u^0, \mu), (x^{0\prime}, u^{0\prime}, \mu'))
\end{align*}
by Assumption~\ref{ass:m3pcont} and Assumption~\ref{ass:m3picont}, where $d((x^0, u^0, \mu), (x^{0\prime}, u^{0\prime}, \mu') = \max(\mathbf 1_{x^0}(x^{0\prime}), \mathbf 1_{u^0}(u^{0\prime}), \mathbf 1_{\mu}(\mu'))$ as discussed in the main text.
\end{proof}

\section{Approximation of Action-Value functions} \label{app:theo}
In this section, we give approximation lemmas for the value functions, together with definitions that were omitted in the main text, and are used in some of the following proofs.

For fixed $(\pi, \pi^0)$, the true minor action-value function is defined by the Bellman equation
\begin{multline*}
    Q_{\pi, \pi^0}(t, x, u, x^0, \mu) = \sum_{u^0} \pi^0_t(u^0 \mid x^0, \mu)
    \bigg[ r(x, u, x^0, u^0, \mu) + \sum_{x^{0\prime}} P^0(x^{0\prime} \mid x^0, u^0, \mu) \\
    \cdot \sum_{x'} P(x' \mid x, u, x^0, u^0, \mu)
    \max_{u'} Q_{\pi, \pi^0}(t+1, x', u', x^{0\prime}, T^\pi_t(x^0, u^0, \mu)) \bigg],
\end{multline*}
while its approximated variant follows
\begin{multline*}
    \hat Q_{\pi, \pi^0}(t, x, u, x^0, \mu) = \sum_{u^0} \pi^0_t(u^0 \mid x^0, \mathrm{proj}_\delta\mu)
    \bigg[ r(x, u, x^0, u^0, \mathrm{proj}_\delta\mu) + \sum_{x^{0\prime}} P^0(x^{0\prime} \mid x^0, u^0, \mathrm{proj}_\delta\mu) \\
    \cdot \sum_{x'} P(x' \mid x, u, x^0, u^0, \mathrm{proj}_\delta\mu)
    \max_{u'} \hat Q_{\pi, \pi^0}(t+1, x', u', x^{0\prime}, T^\pi_t(x^0, u^0, \mathrm{proj}_\delta\mu)) \bigg],
\end{multline*}
since we have $\hat Q_{\pi, \pi^0}(t+1, x', u', x^{0\prime}, \mathrm{proj}_\delta T^\pi_t(x^0, u^0, \mathrm{proj}_\delta\mu)) = \hat Q_{\pi, \pi^0}(t+1, x', u', x^{0\prime}, T^\pi_t(x^0, u^0, \mathrm{proj}_\delta\mu))$ by definition. 

We can show that the approximate Q functions tend uniformly to the true Q functions as the $\delta$-partition becomes fine. Here, the supremum over policies is over $\Pi, \Pi^0$.

\begin{lemma} \label{lem:Qest}
    Under Assumptions~\ref{ass:m3pcont}, \ref{ass:m3rcont}, \ref{ass:m3picont}, we have for $\mu, \nu \in \mathcal P(\mathcal X)$ at all times $t \in \mathcal T$ that
    \begin{equation}
        \sup_{x^0, u^0, \pi, \pi^0} \left| \hat Q^0_{\pi, \pi^0}(t, x^0, u^0, \mu) - \hat Q^0_{\pi, \pi^0}(t, x^0, u^0, \nu) \right| = \mathcal O(\delta + \lVert \mu - \nu \rVert), \label{eq:Q0est} 
    \end{equation}
    \begin{equation}
        \sup_{x, u, x^0, \pi, \pi^0} \left| \hat Q_{\pi, \pi^0}(t, x, u, x^0, \mu) - \hat Q_{\pi, \pi^0}(t, x, u, x^0, \nu) \right| = \mathcal O(\delta + \lVert \mu - \nu \rVert). \label{eq:Qest}
    \end{equation}
\end{lemma}

\begin{lemma} \label{lem:Qconv}
    Under Assumptions~\ref{ass:m3pcont}, \ref{ass:m3rcont}, \ref{ass:m3picont}, at all times $t$, the approximate major and minor action-value functions uniformly converge to the true action-value functions,
    \begin{equation}
        \sup_{x^0, u^0, \mu, \pi, \pi^0} | \hat Q^0_{\pi, \pi^0}(t, x^0, u^0, \mu) - Q^0_{\pi, \pi^0}(t, x^0, u^0, \mu) | = \mathcal O(\delta), \label{eq:Q0conv}   
    \end{equation}     
    \begin{equation}
        \sup_{x, u, x^0, \mu, \pi, \pi^0} | \hat Q_{\pi, \pi^0}(t, x, u, x^0, \mu) - Q_{\pi, \pi^0}(t, x, u, x^0, \mu) | = \mathcal O(\delta). \label{eq:Qconv}
    \end{equation}
\end{lemma}

\section{Estimate for Action-Value Functions}
\begin{proof}[Proof of Lemma~\ref{lem:Qest}]
At time $T-1$, we have for any $\delta > 0$ and $\mu, \nu$ that by Assumption~\ref{ass:m3rcont},
\begin{align*}
    &\sup_{x^0, u^0, \pi, \pi^0} \left| \hat Q^0_{\pi, \pi^0}(T-1, x^0, u^0, \mu) - \hat Q^0_{\pi, \pi^0}(T-1, x^0, u^0, \nu) \right| \\
    &\quad = \sup_{x^0, u^0, \pi, \pi^0} \left| r^0(x^0, u^0, \mathrm{proj}_\delta\mu) - r^0(x^0, u^0, \mathrm{proj}_\delta\nu) \right|
    \leq L_r (2 \delta + \lVert \mu - \nu \rVert) = \mathcal O(\delta + \lVert \mu - \nu \rVert)
\end{align*}
by triangle inequality, as the projection of $\mu, \nu$ can shift $\mu, \nu$ at most by $\delta$ each.

Similarly, for the induction step, assuming \eqref{eq:Q0est} at time $t+1$, then at time $t$ we have:
\begin{align*}
    &\sup_{x^0, u^0, \pi, \pi^0} \left| \hat Q^0_{\pi, \pi^0}(t, x^0, u^0, \mu) - \hat Q^0_{\pi, \pi^0}(t, x^0, u^0, \nu) \right| \\
    &\quad \leq \sup_{x^0, u^0, \pi, \pi^0} \left| r^0(x^0, u^0, \mathrm{proj}_\delta\mu) - r^0(x^0, u^0, \mathrm{proj}_\delta\nu) \right| \\
    &\qquad + \sup_{x^0, u^0, \pi, \pi^0} \left| \sum_{x^{0\prime}} P^0(x^{0\prime} \mid x^0, u^0, \mathrm{proj}_\delta\mu) \max_{u^{0\prime}} \hat Q^0_{\pi, \pi^0}(t+1, x^{0\prime}, u^{0\prime}, T^\pi_t(x^0, u^0, \mathrm{proj}_\delta\mu)) 
    \right.\\&\hspace{2.8cm}\left.
    - \sum_{x^{0\prime}} P^0(x^{0\prime} \mid x^0, u^0, \mathrm{proj}_\delta\nu) \max_{u^{0\prime}} \hat Q^0_{\pi, \pi^0}(t+1, x^{0\prime}, u^{0\prime}, T^\pi_t(x^0, u^0, \mathrm{proj}_\delta\nu)) \right| \\
    &\quad \leq L_r (2 \delta + \lVert \mu - \nu \rVert) + Q^0_{\mathrm{max}} |\mathcal X^0| L_{P^0} (2 \delta + \lVert \mu - \nu \rVert) \\
    &\qquad + 2 \sup_{x^0, u^0, x^{0\prime}, u^{0\prime}, \pi, \pi^0} \left| \hat Q^0_{\pi, \pi^0}(t+1, x^{0\prime}, u^{0\prime}, T^\pi_t(x^0, u^0, \mathrm{proj}_\delta\mu)) - Q^0_{\pi, \pi^0}(t+1, x^{0\prime}, u^{0\prime}, T^\pi_t(x^0, u^0, \mathrm{proj}_\delta\nu)) \right| \\
    &\quad = \mathcal O(\delta + \lVert \mu - \nu \rVert)
\end{align*}
by Assumption~\ref{ass:m3rcont}, Assumption~\ref{ass:m3pcont}, and induction assumption using $\sup_{x^0, u^0} \left\lVert T^\pi_t(x^0, u^0, \mathrm{proj}_\delta\mu) - T^\pi_t(x^0, u^0, \mathrm{proj}_\delta\nu) \right\rVert \leq L_T (2\delta + \lVert \mu - \nu \rVert)$ by Lemma~\ref{lem:Tcont}. Here,  $Q^0_{\mathrm{max}} \coloneqq T \max r^0$.

The same argument for the minor agent completes the proof for \eqref{eq:Qest}: 

At time $T-1$, we have for any $\delta > 0$ and $\mu, \nu$ that by Assumption~\ref{ass:m3rcont}, again
\begin{align*}
    &\sup_{x, u, x^0, \pi, \pi^0} \left| \hat Q_{\pi, \pi^0}(T-1, x, u, x^0, \mu) - \hat Q_{\pi, \pi^0}(T-1, x, u, x^0, \nu) \right| \\
    &\quad = \sup_{x, u, x^0, \pi, \pi^0} \left| \sum_{u^0} \pi^0_{T-1}(u^0 \mid x^0, \mathrm{proj}_\delta\mu) r(x, u, x^0, u^0, \mathrm{proj}_\delta\mu) - \sum_{u^0} \pi^0_{T-1}(u^0 \mid x^0, \mathrm{proj}_\delta\nu) r(x, u, x^0, u^0, \mathrm{proj}_\delta\nu) \right| \\
    &\quad \leq |\mathcal U^0| L_{\Pi^0} (2 \delta + \lVert \mu - \nu \rVert) \max r^0 + L_r (2 \delta + \lVert \mu - \nu \rVert) = \mathcal O(\delta + \lVert \mu - \nu \rVert).
\end{align*}

For the induction step, assuming \eqref{eq:Qest} at time $t+1$, then at time $t$ we have:
\begin{align*}
    &\sup_{x, u, x^0, \pi, \pi^0} \left| \hat Q_{\pi, \pi^0}(t, x, u, x^0, \mu) - \hat Q_{\pi, \pi^0}(t, x, u, x^0, \nu) \right| \\
    &\quad \leq \sup_{x, u, x^0, \pi, \pi^0} \left| \sum_{u^0} \pi^0_t(u^0 \mid x^0, \mathrm{proj}_\delta\mu) r(x, u, x^0, u^0, \mathrm{proj}_\delta\mu) - \sum_{u^0} \pi^0_t(u^0 \mid x^0, \mathrm{proj}_\delta\nu) r(x, u, x^0, u^0, \mathrm{proj}_\delta\nu) \right| \\
    &\qquad + \sup_{x, u, x^0, \pi, \pi^0} \left| \sum_{u^0} \pi^0_t(u^0 \mid x^0, \mathrm{proj}_\delta\mu) \sum_{x^{0\prime}} P^0(x^{0\prime} \mid x^0, u^0, \mathrm{proj}_\delta\mu) 
    \right.\\&\hspace{4.8cm}\left.
    \cdot \sum_{x'} P(x' \mid x, u, x^0, u^0, \mathrm{proj}_\delta\mu) \max_{u'} \hat Q_{\pi, \pi^0}(t+1, x', u', x^{0\prime}, T^\pi_t(x^0, u^0, \mathrm{proj}_\delta\mu)) 
    \right.\\&\hspace{2.8cm}\left.
    - \sum_{u^0} \pi^0_t(u^0 \mid x^0, \mathrm{proj}_\delta\mu) \sum_{x^{0\prime}} P^0(x^{0\prime} \mid x^0, u^0, \mathrm{proj}_\delta\nu) 
    \right.\\&\hspace{4.8cm}\left.
    \cdot \sum_{x'} P(x' \mid x, u, x^0, u^0, \mathrm{proj}_\delta\mu) \max_{u'} \hat Q_{\pi, \pi^0}(t+1, x', u', x^{0\prime}, T^\pi_t(x^0, u^0, \mathrm{proj}_\delta\nu)) \right| \\
    &\quad \leq |\mathcal U^0| L_{\Pi^0} (2 \delta + \lVert \mu - \nu \rVert) \max r^0 + L_r (2 \delta + \lVert \mu - \nu \rVert) \\
    &\qquad + |\mathcal U^0| Q_{\mathrm{max}} L_{\Pi^0} (2 \delta + \lVert \mu - \nu \rVert) + |\mathcal X^0| Q_{\mathrm{max}} L_{P^0} (2 \delta + \lVert \mu - \nu \rVert) + |\mathcal X| Q_{\mathrm{max}} L_P (2 \delta + \lVert \mu - \nu \rVert) \\
    &\qquad + 2 \sup_{x^0, u^0, x', u', x^{0\prime}, \pi, \pi^0} \left| \hat Q_{\pi, \pi^0}(t+1, x', u', x^{0\prime}, T^\pi_t(x^0, u^0, \mathrm{proj}_\delta\mu)) - Q^0_{\pi, \pi^0}(t+1, x', u', x^{0\prime}, T^\pi_t(x^0, u^0, \mathrm{proj}_\delta\nu)) \right| \\
    &\quad = \mathcal O(\delta + \lVert \mu - \nu \rVert)
\end{align*}
by Assumption~\ref{ass:m3picont}, Assumption~\ref{ass:m3rcont}, Assumption~\ref{ass:m3pcont}, and applying the induction assumption on the last term, where we again use $\sup_{x^0, u^0} \left\lVert T^\pi_t(x^0, u^0, \mathrm{proj}_\delta\mu) - T^\pi_t(x^0, u^0, \mathrm{proj}_\delta\nu) \right\rVert \leq L_T (2\delta + \lVert \mu - \nu \rVert)$ by Lemma~\ref{lem:Tcont}. Here,  $Q_{\mathrm{max}} \coloneqq T \max r$. This completes the proof for \eqref{eq:Qest}.
\end{proof}

\section{Convergence of Approximate Action-Value Functions}
\begin{proof}[Proof of Lemma~\ref{lem:Qconv}]
    The proof is by (reverse) induction. At terminal time $t=T-1$, we have by Assumption~\ref{ass:m3rcont}
    \begin{align*}
        \sup_{x^0, u^0, \mu, \pi, \pi^0} \left| \hat Q^0_{\pi, \pi^0}(T-1, x^0, u^0, \mu) - Q^0_{\pi, \pi^0}(T-1, x^0, u^0, \mu) \right| 
        = \sup_{x^0, u^0, \mu, \pi, \pi^0} \left| r^0(x^0, u^0, \mathrm{proj}_\delta\mu) - r^0(x^0, u^0, \mu) \right|
        \leq L_r \delta
    \end{align*}

    Assume \eqref{eq:Q0conv} holds at time $t+1$, then at time $t$ we have
    \begin{align*}
        &\sup_{x^0, u^0, \mu, \pi, \pi^0} \left| \hat Q^0_{\pi, \pi^0}(t, x^0, u^0, \mu) - Q^0_{\pi, \pi^0}(t, x^0, u^0, \mu) \right| \\
        &\quad \leq \sup_{x^0, u^0, \mu, \pi, \pi^0} \left| r^0(x^0, u^0, \mathrm{proj}_\delta\mu) - r^0(x^0, u^0, \mu) \right| \\
        &\qquad + \sup_{x^0, u^0, \mu, \pi, \pi^0} \left| \sum_{x^{0\prime}} P^0(x^{0\prime} \mid x^0, u^0, \mathrm{proj}_\delta\mu) \max_{u^{0\prime}} \hat Q^0_{\pi, \pi^0}(t+1, x^{0\prime}, u^{0\prime}, T^\pi_t(x^0, u^0, \mathrm{proj}_\delta\mu)) 
        \right.\\&\hspace{2.8cm}\left.
        - \sum_{x^{0\prime}} P^0(x^{0\prime} \mid x^0, u^0, \mu) \max_{u^{0\prime}} Q^0_{\pi, \pi^0}(t+1, x^{0\prime}, u^{0\prime}, T^\pi_t(x^0, u^0, \mu)) \right| \\
        &\quad \leq L_r \delta + Q^0_{\mathrm{max}} |\mathcal X^0| L_{P^0} \delta + \mathcal O(\delta) \\
        &\qquad + 2 \sup_{x^0, u^0, \mu, x^{0\prime}, u, \pi, \pi^0} \left| \hat Q^0_{\pi, \pi^0}(t+1, x^{0\prime}, u, T^\pi_t(x^0, u^0, \mu)) - Q^0_{\pi, \pi^0}(t+1, x^{0\prime}, u, T^\pi_t(x^0, u^0, \mu)) \right| = \mathcal O(\delta)
    \end{align*}
    by Assumption~\ref{ass:m3rcont}, Assumption~\ref{ass:m3pcont}, the estimate from Lemma~\ref{lem:Qest} with $| T^\pi_t(x^0, u^0, \mathrm{proj}_\delta\mu) - T^\pi_t(x^0, u^0, \mu) | \leq L_T \delta = \mathcal O(\delta)$ by Lemma~\ref{lem:Tcont}, and the induction assumption for the final term. Here, $Q^0_{\mathrm{max}} \coloneqq T \max r^0$. This completes the proof for \eqref{eq:Q0conv}.

    For the minor player in \eqref{eq:Qconv}, we have by the same argument
    \begin{align*}
        &\sup_{x, u, x^0, \mu, \pi, \pi^0} \left| \hat Q_{\pi, \pi^0}(T-1, x, u, x^0, \mu) - Q_{\pi, \pi^0}(T-1, x, u, x^0, \mu) \right| \\
        &\quad \leq \sup_{x, u, x^0, u^0, \mu, \pi, \pi^0} \left| r(x, u, x^0, u^0, \mathrm{proj}_\delta\mu) - r(x, u, x^0, u^0, \mu) \right|
        \leq L_r \delta
    \end{align*}
    at terminal time $T-1$, and then inductively at any time $t$
    \begin{align*}
        &\sup_{x, u, x^0, \mu, \pi, \pi^0} \left| \hat Q_{\pi, \pi^0}(t, x, u, x^0, \mu) - Q_{\pi, \pi^0}(t, x, u, x^0, \mu) \right| \\
        &\quad \leq \sup_{x, u, x^0, u^0, \mu, \pi, \pi^0} \left| r(x, u, x^0, u^0, \mathrm{proj}_\delta\mu) - r(x, u, x^0, u^0, \mu) \right| \\
        &\qquad + \sup_{x^0, u^0, \mu, \pi, \pi^0} \left| \sum_{u^0} \pi^0_t(u^0 \mid x^0, \mathrm{proj}_\delta\mu) \sum_{x^{0\prime}} P^0(x^{0\prime} \mid x^0, u^0, \mathrm{proj}_\delta\mu)
        \right.\\&\hspace{4.5cm}\left.
        \cdot \sum_{x'} P(x' \mid x, u, x^0, u^0, \mathrm{proj}_\delta\mu) \max_{u'} \hat Q_{\pi, \pi^0}(t+1, x', u', x^{0\prime}, T^\pi_t(x^0, u^0, \mathrm{proj}_\delta\mu))
        \right.\\&\hspace{2.8cm}\left.
        - \sum_{u^0} \pi^0_t(u^0 \mid x^0, \mu) \sum_{x^{0\prime}} P^0(x^{0\prime} \mid x^0, u^0, \mu)
        \right.\\&\hspace{4.5cm}\left.
        \cdot \sum_{x'} P(x' \mid x, u, x^0, u^0, \mu) \max_{u'} Q_{\pi, \pi^0}(t+1, x', u', x^{0\prime}, T^\pi_t(x^0, u^0, \mu)) \right| \\
        &\quad \leq L_r \delta + Q_{\mathrm{max}} |\mathcal U^0| L_{\Pi^0} \delta + Q_{\mathrm{max}} |\mathcal X^0| L_{P^0} \delta + Q_{\mathrm{max}} |\mathcal X| L_{P} \delta + \mathcal O(\delta) \\
        &\qquad + \sup_{\mu, x^{0\prime}, u, \pi, \pi^0} \left| \hat Q_{\pi, \pi^0}(t+1, x', u', x^{0\prime}, T^\pi_t(x^0, u^0, \mu)) - Q_{\pi, \pi^0}(t+1, x', u', x^{0\prime}, T^\pi_t(x^0, u^0, \mu)) \right| = \mathcal O(\delta)
    \end{align*}
    using also Assumption~\ref{ass:m3picont} and $Q_{\mathrm{max}} \coloneqq T \max r$, which completes the proof by induction.
\end{proof}

\section{Propagation of Chaos} \label{app:m3muconv}
\begin{proof}[Proof of Theorem~\ref{thm:m3muconv}]
For readability, we abbreviate the states and actions at time $t$ as $\mathcal J_t \coloneqq (x^{0,N}_t, u^{0,N}_t, x^{1,N}_t, u^{1,N}_t, \ldots, x^{N,N}_t, u^{N,N}_t)$, and write $\E_{\mathcal J_t}$ for the conditional expectation given $\mathcal J_t$. Without loss of generality, we show the statements for families $\mathcal F$ that are additionally uniformly bounded by some constant $M_f$, since the support of $f \in \mathcal F$ is compact and we can add any constant to $f$ without changing the difference between expectations in \eqref{eq:m3muconv-min}. We also define the conditional expectation of the empirical mean field at time $t+1$ given variables at time $t$,
\begin{align*}
    \hat \mu^N_{t+1} \coloneqq T^\pi_t(x^{0,N}_t, u^{0,N}_t, \mu^N_t).
\end{align*}

We show the statement \eqref{eq:m3muconv-min} at all times by induction. At time $t=0$, the statement follows from a law of large numbers (LLN), see also below. Assuming that \eqref{eq:m3muconv-min} holds at time $t$, then at time $t+1$ we have
\begin{align}
    &\sup_{\hat \pi, \pi, \pi^0} \sup_{f \in \mathcal F} \left| \E \left[ f(x^{1,N}_{t+1}, u^{1,N}_{t+1}, x^{0,N}_{t+1}, u^{0,N}_{t+1}, \mu^N_{t+1}) - f(x_{t+1}, u_{t+1}, x^0_{t+1}, u^0_{t+1}, \mu_{t+1}) \right] \right| \nonumber\\
    &\quad \leq \sup_{\hat \pi, \pi, \pi^0} \sup_{f \in \mathcal F} \left| \E \left[ f(x^{1,N}_{t+1}, u^{1,N}_{t+1}, x^{0,N}_{t+1}, u^{0,N}_{t+1}, \mu^N_{t+1}) - f(x^{1,N}_{t+1}, u^{1,N}_{t+1}, x^{0,N}_{t+1}, u^{0,N}_{t+1}, \hat \mu^N_{t+1}) \right] \right| \label{eq:t1}\\
    \begin{split}
        &\quad + \sup_{\hat \pi, \pi, \pi^0} \sup_{f \in \mathcal F} \left| \E \left[ \int f(x^{1,N}_{t+1}, u^{1,N}_{t+1}, x^{0,N}_{t+1}, u^0, \hat \mu^N_{t+1}) \pi^0_{t+1}(\mathrm du^0 \mid x^{0,N}_{t+1}, \mu^N_{t+1}) \right]
        \right.\\&\hspace{2.5cm}\left.
        - \E \left[ \int f(x^{1,N}_{t+1}, u^{1,N}_{t+1}, x^{0,N}_{t+1}, u^0, \hat \mu^N_{t+1}) \pi^0_{t+1}(\mathrm du^0 \mid x^{0,N}_{t+1}, \hat \mu^N_{t+1}) \right] \right|
    \end{split}
    \label{eq:t2} \\
    \begin{split}
        &\quad + \sup_{\hat \pi, \pi, \pi^0} \sup_{f \in \mathcal F} \left| \E \left[ \iint f(x^{1,N}_{t+1}, u, x^{0,N}_{t+1}, u^0, \hat \mu^N_{t+1}) \pi^0_{t+1}(\mathrm du^0 \mid x^{0,N}_{t+1}, \hat \mu^N_{t+1}) \hat \pi_{t+1}(\mathrm du \mid x^{1,N}_{t+1}, x^{0,N}_{t+1}, \mu^N_{t+1}) \right]
        \right.\\&\hspace{2.5cm}\left.
        - \E \left[ \iint f(x^{1,N}_{t+1}, u, x^{0,N}_{t+1}, u^0, \hat \mu^N_{t+1}) \pi^0_{t+1}(\mathrm du^0 \mid x^{0,N}_{t+1}, \hat \mu^N_{t+1}) \hat \pi_{t+1}(\mathrm du \mid x^{1,N}_{t+1}, x^{0,N}_{t+1}, \hat \mu^N_{t+1}) \right] \right|
    \end{split}
    \label{eq:t3} \\
    \begin{split}
        &\quad + \sup_{\hat \pi, \pi, \pi^0} \sup_{f \in \mathcal F} \left| \E \left[ \iint f(x^{1,N}_{t+1}, u, x^{0,N}_{t+1}, u^0, \hat \mu^N_{t+1}) \pi^0_{t+1}(\mathrm du^0 \mid x^{0,N}_{t+1}, \hat \mu^N_{t+1}) \hat \pi_{t+1}(\mathrm du \mid x^{1,N}_{t+1}, x^{0,N}_{t+1}, \hat \mu^N_{t+1}) \right]
        \right.\\&\hspace{2.5cm}\left.
        - \E \left[ \iint f(x_{t+1}, u, x^0_{t+1}, u^0, \mu_{t+1})\pi^0_{t+1}(\mathrm du^0 \mid x^0_{t+1}, \mu_{t+1}) \hat \pi_{t+1}(\mathrm du \mid x_{t+1}, x^0_{t+1}, \mu_{t+1}) \right] \right|
    \end{split}
    \label{eq:t4}
\end{align}

The first term \eqref{eq:t1} is
\begin{align*}
    &\sup_{\hat \pi, \pi, \pi^0} \sup_{f \in \mathcal F} \left| \E \left[ f(x^{1,N}_{t+1}, u^{1,N}_{t+1}, x^{0,N}_{t+1}, u^{0,N}_{t+1}, \mu^N_{t+1}) - f(x^{1,N}_{t+1}, u^{1,N}_{t+1}, x^{0,N}_{t+1}, u^{0,N}_{t+1}, \hat \mu^N_{t+1}) \right] \right| \\
    &\quad \leq \sup_{\hat \pi, \pi, \pi^0} L_f \E \left[ \left\Vert \mu^N_{t+1} - \hat \mu^N_{t+1} \right\Vert \right] \\
    &\quad = \sup_{\hat \pi, \pi, \pi^0} L_f \E \left[ \sum_{x \in \mathcal X} \left| \mu^N_{t+1}(x) - \hat \mu^N_{t+1}(x) \right| \right] \\
    &\quad = \sup_{\hat \pi, \pi, \pi^0} L_f \sum_{x \in \mathcal X} \E \left[ \left| \frac 1 N \sum_{i=1}^N \mathbf 1_x(x^{i,N}_{t+1}) - \frac 1 N \sum_{i=2}^N \mathbf 1_x(x^{i,N}_{t+1}) \right| \right] \\
    &\qquad + \sup_{\hat \pi, \pi, \pi^0} L_f \sum_{x \in \mathcal X} \E \left[ \E_{\mathcal J_t} \left[ \left| \frac 1 N \sum_{i=2}^N \mathbf 1_x(x^{i,N}_{t+1}) - \E_{\mathcal J_t} \left[ \frac 1 N \sum_{i=2}^N \mathbf 1_x(x^{i,N}_{t+1}) \right] \right| \right] \right] \\
    &\qquad + \sup_{\hat \pi, \pi, \pi^0} L_f \sum_{x \in \mathcal X} \E \left[ \left| \frac 1 N \sum_{u \in \mathcal U} P(x \mid x^{1,N}_{t}, u, x^{0,N}_t, u^{0,N}_t, \mu^N_t) \pi_t(u \mid x^{1,N}_{t}, x^{0,N}_t, \mu^N_t) \right| \right] \\
    &\quad \leq L_f |\mathcal X| \left( \frac{1}{N} + \sqrt{\frac{4}{N}} + \frac{|\mathcal U|}{N} \right) = \mathcal O(1/\sqrt{N})
\end{align*}
and tends to zero at rate $\mathcal O(1/\sqrt{N})$, where the last term is the difference between $\E_{\mathcal J_t} \left[ \frac 1 N \sum_{i=2}^N \mathbf 1_x(x^{i,N}_{t+1}) \right]$ and $\hat \mu^N_{t+1}(x) = \E_{\mathcal J_t} \left[ \frac 1 N \sum_{i=1}^N \mathbf 1_x(x^{i,N}_{t+1}) \right]$, while the middle term is obtained by tower rule and analyzed by a weak LLN argument, i.e.
\begin{align*}
    &\sup_{\hat \pi, \pi, \pi^0} L_f \sum_{x \in \mathcal X} \E \left[ \E_{\mathcal J_t} \left[ \left| \frac 1 N \sum_{i=2}^N \mathbf 1_x(x^{i,N}_{t+1}) - \E_{\mathcal J_t} \left[ \frac 1 N \sum_{i=2}^N \mathbf 1_x(x^{i,N}_{t+1}) \right] \right| \right] \right] \\
    &\quad = \sup_{\hat \pi, \pi, \pi^0} L_f \sum_{x \in \mathcal X} \E \left[ \E_{\mathcal J_t} \left[ \left| \frac 1 N \sum_{i=2}^N \left( \mathbf 1_x(x^{i,N}_{t+1}) - \E_{\mathcal J_t} \left[ \mathbf 1_x(x^{i,N}_{t+1}) \right] \right) \right| \right] \right] \\
    &\quad \leq \sup_{\hat \pi, \pi, \pi^0} L_f \sum_{x \in \mathcal X} \E \left[ \E_{\mathcal J_t} \left[ \left( \frac 1 N \sum_{i=2}^N \left( \mathbf 1_x(x^{i,N}_{t+1}) - \E_{\mathcal J_t} \left[ \mathbf 1_x(x^{i,N}_{t+1}) \right] \right) \right)^2 \right] \right]^{1/2} \\
    &\quad = \sup_{\hat \pi, \pi, \pi^0} L_f \sum_{x \in \mathcal X} \E \left[ \E_{\mathcal J_t} \left[ \frac 1 N \sum_{i=2}^N \left( \mathbf 1_x(x^{i,N}_{t+1}) - \E_{\mathcal J_t} \left[ \mathbf 1_x(x^{i,N}_{t+1}) \right] \right)^2 \right] \right]^{1/2} \\
    &\quad \leq L_f |\mathcal X| \sqrt{\frac{N-1}{N^2} \cdot 2^2} \leq L_f |\mathcal X| \sqrt{\frac{4}{N}}
\end{align*}
by conditional independence of $x^{i,N}_{t+1}$ given $\mathcal J_t$.

Similarly, for the second term \eqref{eq:t2} we obtain
\begin{align*}
    \begin{split}
        &\sup_{\hat \pi, \pi, \pi^0} \sup_{f \in \mathcal F} \left| \E \left[ \int f(x^{1,N}_{t+1}, u^{1,N}_{t+1}, x^{0,N}_{t+1}, u^0, \hat \mu^N_{t+1}) \pi^0_{t+1}(\mathrm du^0 \mid x^{0,N}_{t+1}, \mu^N_{t+1}) \right]
        \right.\\&\hspace{2.5cm}\left.
        - \E \left[ \int f(x^{1,N}_{t+1}, u^{1,N}_{t+1}, x^{0,N}_{t+1}, u^0, \hat \mu^N_{t+1}) \pi^0_{t+1}(\mathrm du^0 \mid x^{0,N}_{t+1}, \hat \mu^N_{t+1}) \right] \right|
    \end{split} \\
    &\quad \leq |\mathcal U^0| M_f L_{\Pi_0} \E \left[ \left\Vert \mu^N_{t+1} - \hat \mu^N_{t+1} \right\Vert \right] = \mathcal O(1/\sqrt{N})
\end{align*}
by Assumption~\ref{ass:m3picont}, and also for the third term \eqref{eq:t3},
\begin{align*}
    \begin{split}
        &\sup_{\hat \pi, \pi, \pi^0} \sup_{f \in \mathcal F} \left| \E \left[ \iint f(x^{1,N}_{t+1}, u, x^{0,N}_{t+1}, u^0, \hat \mu^N_{t+1}) \pi^0_{t+1}(\mathrm du^0 \mid x^{0,N}_{t+1}, \hat \mu^N_{t+1}) \hat \pi_{t+1}(\mathrm du \mid x^{1,N}_{t+1}, x^{0,N}_{t+1}, \mu^N_{t+1}) \right]
        \right.\\&\hspace{2.5cm}\left.
        - \E \left[ \iint f(x^{1,N}_{t+1}, u, x^{0,N}_{t+1}, u^0, \hat \mu^N_{t+1}) \pi^0_{t+1}(\mathrm du^0 \mid x^{0,N}_{t+1}, \hat \mu^N_{t+1}) \hat \pi_{t+1}(\mathrm du \mid x^{1,N}_{t+1}, x^{0,N}_{t+1}, \hat \mu^N_{t+1}) \right] \right|
    \end{split} \\
    &\quad \leq |\mathcal U| M_f L_\Pi \E \left[ \left\Vert \mu^N_{t+1} - \hat \mu^N_{t+1} \right\Vert \right] = \mathcal O(1/\sqrt{N}).
\end{align*}

For the last term \eqref{eq:t4}, we have
\begin{align*}
    &\sup_{\hat \pi, \pi, \pi^0} \sup_{f \in \mathcal F} \left| \E \left[ \iint f(x^{1,N}_{t+1}, u, x^{0,N}_{t+1}, u^0, \hat \mu^N_{t+1}) \pi^0_{t+1}(\mathrm du^0 \mid x^{0,N}_{t+1}, \hat \mu^N_{t+1}) \hat \pi_{t+1}(\mathrm du \mid x^{1,N}_{t+1}, x^{0,N}_{t+1}, \hat \mu^N_{t+1}) \right]
    \right.\\&\hspace{2.5cm}\left.
    - \E \left[ \iint f(x_{t+1}, u_{t+1}, x^0_{t+1}, u^0_{t+1}, \mu_{t+1})\pi^0_{t+1}(\mathrm du^0 \mid x^0_{t+1}, \mu_{t+1}) \hat \pi_{t+1}(\mathrm du \mid x_{t+1}, x^0_{t+1}, \mu_{t+1}) \right] \right| \\
    &= \sup_{\hat \pi, \pi, \pi^0} \sup_{f \in \mathcal F} \left| \E \left[ \iint g_2(x, x^0, x^{1,N}_t, u^{1,N}_t, x^{0,N}_t, u^{0,N}_t, \mu^N_t) P^0(\mathrm dx^0 \mid x^{0,N}_t, u^{0,N}_t, \mu^N_t) P(\mathrm dx \mid x^{1,N}_t, u^{1,N}_t, x^{0,N}_t, u^{0,N}_t, \mu^N_t) \right] 
    \right.\\&\hspace{2.5cm}\left.
    - \E \left[ \iiiint g_2(x, x^0, x_t, u_t, x^0_t, u^0_t, \mu_t) P^0(\mathrm dx^0 \mid x^0_t, u^0_t, \mu_t) P(\mathrm dx \mid x_t, u_t, x^0_t, u^0_t, \mu_t) \right] \right|
\end{align*}
where we define 
\begin{align*}
    &g_2(x, x^0, x_t, u_t, x^0_t, u^0_t, \mu_t) \\
    &\quad \coloneqq \iint f(x, u, x^0, u^0, T^\pi_t(x^0_t, u^0_t, \mu_t)) \pi^0_{t+1}(\mathrm du^0 \mid x^0, T^\pi_t(x^0_t, u^0_t, \mu_t)) \hat \pi_{t+1}(\mathrm du \mid x, x^0, T^\pi_t(x^0_t, u^0_t, \mu_t)).
\end{align*} 

We show that the terms inside the expectations are Lipschitz in $(x^{1,N}_t, u^{1,N}_t, x^{0,N}_t, u^{0,N}_t, \mu^N_t)$ and $(x_t, u_t, x^0_t, u^0_t, \mu_t)$, which will imply convergence of the second term at rate $\mathcal O(1/\sqrt{N})$ by the induction assumption, completing the proof of \eqref{eq:m3muconv-min}.

First, note that $T^\pi_t$ is Lipschitz with constant $L_T$ by Lemma~\ref{lem:Tcont}. Therefore, the map $(x, u, x^0, u^0, x^0_t, u^0_t, \mu_t) \mapsto f(x, u, x^0, u^0, T^\pi_t(x^0_t, u^0_t, \mu_t))$ is also Lipschitz with constant $L_f L_T$. We similarly iteratively obtain Lipschitzness of the maps
\begin{align*}
    g_1(x, u, x^0, x_t, u_t, x^0_t, u^0_t, \mu_t) &\coloneqq \int f(x, u, x^0, u^0, T^\pi_t(x^0_t, u^0_t, \mu_t)) \pi^0_{t+1}(\mathrm du^0 \mid x^0, T^\pi_t(x^0_t, u^0_t, \mu_t)) \\
    g_2(x, x^0, x_t, u_t, x^0_t, u^0_t, \mu_t) &\coloneqq \int g_1(x, u, x^0, x_t, u_t, x^0_t, u^0_t, \mu_t) \hat \pi_{t+1}(\mathrm du \mid x, x^0, T^\pi_t(x^0_t, u^0_t, \mu_t)) \\
    g_3(x, x_t, u_t, x^0_t, u^0_t, \mu_t) &\coloneqq \int g_2(x, x^0, x_t, u_t, x^0_t, u^0_t, \mu_t) P^0(\mathrm dx^0 \mid x^0_t, u^0_t, \mu_t) \\
    g_4(x_t, u_t, x^0_t, u^0_t, \mu_t) &\coloneqq \int g_3(x, x_t, u_t, x^0_t, u^0_t, \mu_t) P(\mathrm dx \mid x_t, u_t, x^0_t, u^0_t, \mu_t)
\end{align*}
with Lipschitz constants $L_{g_1} = L_f L_T + |\mathcal U^0| M_f L_{\Pi^0} L_T$, $L_{g_2} = L_{g_1} + |\mathcal U| M_f L_{\Pi} L_T$, $L_{g_3} = L_{g_2} + |\mathcal X^0| M_f L_{P^0}$, $L_{g_4} = L_{g_3} + |\mathcal X| M_f L_{P}$, and finally note that the last term \eqref{eq:t4} is equal to
\begin{align*}
    \sup_{\hat \pi, \pi, \pi^0} \sup_{f \in \mathcal F} \left| \E \left[ g_4(x^{1,N}_t, u^{1,N}_t, x^{0,N}_t, u^{0,N}_t, \mu^N_t) - g_4(x_t, u_t, x^0_t, u^0_t, \mu_t) \right] \right| = \mathcal O(1/\sqrt{N}),
\end{align*}
tending to zero by applying the induction assumption to families of $L_{g_4}$-Lipschitz functions.
\end{proof}

\begin{proof}[Proof of Corollary~\ref{coro:m3muconv}]
The result follows immediately from Theorem~\ref{thm:m3muconv} by noting that $\mathcal F^0 \subseteq \mathcal F$ when considering functions in $\mathcal F^0$ as constant functions over the deviating minor player's variables in $\mathcal F$. 
\end{proof}

\section{Approximate Nash Property} \label{app:varepsNash}
\begin{proof}[Proof of Corollary~\ref{coro:varepsNash}]
Under $(\pi, \pi^0)$, we have for any $\varepsilon > 0$ that there exists $N' \in \mathbb N$ such that for all $N > N'$ we have
\begin{align*}
    &\sup_{\hat \pi \in \Pi} \left| J_N^1((\hat \pi, \pi, \ldots, \pi), \pi^0) - J(\hat \pi, \pi, \pi^0) \right| \\
    &\quad \leq \sup_{\hat \pi \in \Pi} \left| \E \left[ \sum_{t \in \mathcal T} r(x^{1,N}_t, u^{1,N}_t, x^{0,N}_t, u^{0,N}_t, \mu^N_t) \right] - \E \left[ \sum_{t \in \mathcal T} r(x_t, u_t, x^0_t, u^0_t, \mu_t) \right] \right| \\
    &\quad = \sup_{\hat \pi \in \Pi} \left| \sum_{t \in \mathcal T} \E \left[ r(x^{1,N}_t, u^{1,N}_t, x^{0,N}_t, u^{0,N}_t, \mu^N_t) - r(x_t, u_t, x^0_t, u^0_t, \mu_t) \right] \right| \to 0
\end{align*}
by Theorem~\ref{thm:m3muconv} and Assumption~\ref{ass:m3rcont}.

Therefore, using the previous paragraph, for any $\varepsilon > 0$ we also have
\begin{align*}
    &\sup_{\hat \pi \in \Pi} \left( J_1^N((\hat \pi, \ldots, \pi), \pi^0) - J_1^N((\pi, \ldots, \pi), \pi^0) \right) \\
    &\quad \leq \sup_{\hat \pi \in \Pi} \left( J_1^N((\hat \pi, \pi, \ldots, \pi), \pi^0) - J(\hat \pi, \pi, \pi^0) \right) \\
    &\qquad + \sup_{\hat \pi \in \Pi} \left( J(\hat \pi, \pi, \pi^0) - J(\pi, \pi, \pi^0) \right) \\
    &\qquad + \left( J(\pi, \pi, \pi^0) - J_1^N((\pi, \ldots, \pi), \pi^0) \right) \\
    &\quad < \frac{\varepsilon}{2} + 0 + \frac{\varepsilon}{2} = \varepsilon
\end{align*}
for $N$ large enough by definition of M3FNE, which is the desired statement for $i=1$. By symmetry, this applies to all $i \geq 1$. 

In the alternate infinite-horizon case with discount $\gamma \in (0, 1)$ and $\mathcal T \coloneqq \mathbb N$, we first have for any $\varepsilon > 0$ that there exists $N' \in \mathbb N$ such that for all $N > N'$ we have
\begin{align*}
    &\sup_{\hat \pi \in \Pi} \left| J_N^1((\hat \pi, \pi, \ldots, \pi), \pi^0) - J(\hat \pi, \pi, \pi^0) \right| \\
    &\quad \leq \sup_{\hat \pi \in \Pi} \left| \E \left[ \sum_{t=0}^{\infty} \gamma^t r(x^{1,N}_t, u^{1,N}_t, x^{0,N}_t, u^{0,N}_t, \mu^N_t) \right] - \E \left[ \sum_{t=0}^{\infty} \gamma^t r(x_t, u_t, x^0_t, u^0_t, \mu_t) \right] \right| \\
    &\quad = \sup_{\hat \pi \in \Pi} \left| \sum_{t=0}^{T} \gamma^t \E \left[ r(x^{1,N}_t, u^{1,N}_t, x^{0,N}_t, u^{0,N}_t, \mu^N_t) - r(x_t, u_t, x^0_t, u^0_t, \mu_t) \right] \right| + \frac{\varepsilon}{2} \to 0
\end{align*}
by choosing $T$ large enough, and then applying Theorem~\ref{thm:m3muconv}.

An analogous argument for the major player completes the proof, as
\begin{align*}
    &\sup_{\hat \pi \in \Pi} \left| J_N^0((\hat \pi, \pi, \ldots, \pi), \pi^0) - J^0(\hat \pi, \pi, \pi^0) \right| \\
    &\quad \leq \sup_{\hat \pi \in \Pi} \left| \E \left[ \sum_{t \in \mathcal T} r^0(x^{0,N}_t, u^{0,N}_t, \mu^N_t) \right] - \E \left[ \sum_{t \in \mathcal T} r(x_t, u_t, x^0_t, u^0_t, \mu_t) \right] \right| \\
    &\quad = \sup_{\hat \pi \in \Pi} \left| \sum_{t \in \mathcal T} \E \left[ r^0(x^{0,N}_t, u^{0,N}_t, \mu^N_t) - r(x_t, u_t, x^0_t, u^0_t, \mu_t) \right] \right| \to 0
\end{align*}
by Corollary~\ref{coro:m3muconv} and Assumption~\ref{ass:m3rcont}.
\end{proof}

\section{Convergence of Approximate Exploitability} \label{app:disc-opt}
\begin{proof}[Proof of Theorem~\ref{thm:disc-opt}]
Observe that the convergence of approximate minor objectives to the true objectives as $\delta \to 0$,
\begin{align} \label{eq:Vconv}
    \sup_{x, x^0, \mu, \pi, \pi^0}\left| \hat V^{\pi}_{\pi, \pi^0}(t, x, x^0, \mu) - V^{\pi}_{\pi, \pi^0}(t, x, x^0, \mu) \right| = \mathcal O(\delta) \to 0
\end{align}
at all times $t$, follows by the same arguments as in Lemma~\ref{lem:Qconv}. The only difference is that we estimate one more term from policy evaluation instead of the $\max$ operation, using continuity for $\pi$ by Assumption~\ref{ass:m3picont}.

Therefore, the approximate minor objective converges as desired, 
\begin{align*}
    \hat J(\pi, \pi^0) \coloneqq \sum_{x, x^0} \mu_0(x) \mu^0_0(x^0) \hat V^{\pi}_{\pi, \pi^0}(0, x, x^0, \mu_0) \to \sum_{x, x^0} \mu_0(x) \mu^0_0(x^0) V^{\pi}_{\pi, \pi^0}(0, x, x^0, \mu_0) = J(\pi, \pi^0).
\end{align*} 

On the other hand, for the approximate exploitability of the minor player
\begin{align*}
    \hat{\mathcal E}(\pi, \pi^0) = \sum_{x, x^0} \mu_0(x) \mu^0_0(x^0) \left( \max_{\hat \pi' \in \hat \Pi} \hat V^{\hat \pi'}_{\pi, \pi^0}(0, x, x^0, \mu_0)) - \hat V^{\pi}_{\pi, \pi^0}(0, x, x^0, \mu_0) \right)
\end{align*}
and its true exploitability 
\begin{align*}
    \mathcal E(\pi, \pi^0) = \sum_{x, x^0} \mu_0(x) \mu^0_0(x^0) \left( \max_{\pi' \in \Pi} V^{\pi'}_{\pi, \pi^0}(0, x, x^0, \mu_0)) - V^{\pi}_{\pi, \pi^0}(0, x, x^0, \mu_0) \right)
\end{align*}
we first note that $\max_{\hat \pi' \in \hat \Pi} \hat V^{\hat \pi'}_{\pi, \pi^0} = \hat V_{\pi, \pi^0}$ and $\max_{\pi' \in \Pi} V^{\pi'}_{\pi, \pi^0} = V_{\pi, \pi^0}$ \citep[Theorem~3.2.1 and Condition~3.3.4]{hernandez2012discrete}, where the approximate and true optimal value functions are defined by the maximum of the action-value functions over actions, 
\begin{align*}
    \hat V_{\pi, \pi^0}(t, x, x^0, \mu) \coloneqq \max_u \hat Q_{\pi, \pi^0}(t, x, u, x^0, \mu), \quad V_{\pi, \pi^0}(t, x, x^0, \mu) \coloneqq \max_u Q_{\pi, \pi^0}(t, x, u, x^0, \mu).
\end{align*}

Therefore, we obtain
\begin{align*}
    &\hat{\mathcal E}(\pi, \pi^0) - \mathcal E(\pi, \pi^0) \\
    &\quad = \sum_{x, x^0} \mu_0(x) \mu^0_0(x^0) \left( \hat V_{\pi, \pi^0}(0, x, x^0, \mu_0) - V_{\pi, \pi^0}(0, x, x^0, \mu_0) \right) \\
    &\qquad + \sum_{x, x^0} \mu_0(x) \mu^0_0(x^0) \left( V^{\pi}_{\pi, \pi^0}(0, x, x^0, \mu_0) - \hat V^{\pi}_{\pi, \pi^0}(0, x, x^0, \mu_0) \right) = \mathcal O(\delta) \to 0
\end{align*}
where the first term is estimated by Lemma~\ref{lem:Qconv}, and similarly the second by \eqref{eq:Vconv}.

Analogous arguments for the major player complete the proof.
\end{proof}

\section{Additional Experimental Details} \label{app:exp}
In the following, we give a detailed description of the problems considered in evaluation. For initialization of policies, unless mentioned, we use the policy that always picks the first action, in order of definition in the main text. We run our experiments each on a single core of an Intel Xeon Platinum 9242 CPU with 4 GB of memory (RedHat 8.8, and without GPUs). We used around $20 \, 000$ CPU core hours. The code is based on Python 3.9 and NumPy 1.23.4 \citep{harris2020array}, and can be found in the supplementary material.

\subsubsection{SIS epidemics model}
Formally, minor players have states $\mathcal X \coloneqq \{ S, I \}$ for susceptible ($S$) and infected ($I$), and can choose between actions $\mathcal U \coloneqq \{ P, \bar P \}$ for prevention ($P$) and no prevention ($\bar P$). The major player has states $\mathcal X^0 \coloneqq \{ H, L \}$ for high ($H$) and low ($L$) transmissibility regimes (e.g. from seasonal changes or virus mutations), and actions $\mathcal U^0 \coloneqq \{ F, \bar F \}$ for forcing ($F$) preventative actions, or not ($\bar F$). The minor dynamics are then given by
\begin{align*}
    P(I \mid S, \bar P, x^0, u^0, \mu_t) &= (0.5 + \mathbf 1_{H}(x^0) + \mathbf 1_{\bar F}(u^0)) \alpha \mu_t(I) \Delta t, \\
    P(I \mid S, P, \ldots) &= 0, \quad 
    P(S \mid I, \ldots) = \beta \Delta t
\end{align*}
for transmissibility $\alpha > 0$, recovery rate $\beta > 0$ and step size $\Delta t > 0$. The major dynamics are exogeneous and given by
\begin{align*}
    P^0(H \mid L, \ldots) = P^0(L \mid H, \ldots) &= \alpha^0 \Delta t
\end{align*}
for rate $\alpha^0 > 0$. Lastly, the reward functions will be set as
\begin{align*}
    r(x, u, x^0, u^0, \mu) &= - c_I \mathbf 1_{I}(x) - c_P \mathbf 1_{P}(u) (\mathbf 1_{F}(u^0) + 0.5), \\
    r^0(x^0, u^0, \mu) &= - c^0_\mu \mu(I) - c^0_F \mathbf 1_{F}(u^0) (0.5 - \mu(I)),
\end{align*}
i.e. the major player wants to keep infections low via preventative actions, while the minor players are interested only in their own infection, trading off between infection and costly prevention. The major government player has an reputation cost associated with forcing preventative actions that decreases with increasing number of infected players.

Concretely, as parameters we use $\alpha = 0.8$, $\beta = 0.2$, $\mu_0(I) = 0.2$, $\mu^0_0(H) = 0.5$, $\alpha^0 = 0.4$, $\Delta t = 0.1$, $c_I = 0.75$, $c_P = 0.5$, $c^0_\mu = 2$, $c^0_F = 1$ and a horizon of $T = 300$ for each episode.

\subsubsection{Buffet problem}
Formally, minor players have states $\mathcal X = [L]$ for $L$ buffet locations, and can choose to move to any location $\mathcal U = [L]$ with geometric arrival rate, resulting in minor dynamics
\begin{align*}
    P(n \mid [L] \setminus {n}, n, \ldots) &= \alpha \Delta t, \\
    P(n \mid [L] \setminus {n}, [L] \setminus {n}, \ldots) &= 1 - \alpha \Delta t, \\
    P(n \mid n, n, \ldots) &= 1.
\end{align*} 

The major player state consists of the food state of the foraging locations, and the major player at any time tries to fill one of the 3 foraging locations such that the locations remain as full as possible, and optionally as equal as possible. Hence, the major player has states $\mathcal X^0 = \{ 0, \ldots, B-1 \}^L$ indicating the buffet filling status at each of $L$ locations, and actions $\mathcal U^0 = [L]$ for filling up the buffet at a specific location. The major dynamics are such that a filling at location $n$ is gained with probability $\alpha^0_+ \Delta t$ and lost with probability $\alpha^0_- \mu(n) \Delta t$ whenever the current mean field is $\mu$.

Lastly, the rewards are defined as
\begin{align*}
    r(x, u, x^0, u^0, \mu) = c_f x^0_x - c_c \mu(x) - c_u (1 - \mathbf 1_{x}(u)), \\
    r^0(x^0, u^0, \mu) = \frac 1 L \sum_{i \in [L]} \left( c^0_f x^0_i - c^0_b \left| x^0_i - \frac 1 L \sum_{i \in [L]} x^0_i \right| \right),
\end{align*}
where we have the reward coefficients $c_f$ and $c^0_f$ for filled buffets, the crowdedness cost $c_c$, the movement cost $c_u$, and the imbalance cost $c^0_b$.

Concretely, as parameters we use $B=5$, $L=2$, $\alpha = 0.7$, $\mu_0(0) = 1$, $\mu^0_0 = \mathrm{Unif}$, $\alpha^0_+ = 0.9$, $\alpha^0_- = 1.0$, $\Delta t = 0.2$, $c_f = 0.75$, $c_c = 0.5$, $c_u = 1.0$, $c^0_f = 2$, $c^0_b = 1$ and a horizon of $T = 100$ for each episode.

\subsubsection{Advertisement duopoly model}
The regulator chooses one of the actions $\mathcal{U}^0 = \{0, 1, 2\}$, where $0$ denotes average price, $1$ denotes low price and $2$ denotes high price for advertisement by the second company. The regulator's state is $\mathcal{X}^0=\{1,2\}$ where $i,\ i=\{1,2\}$ denotes the case where Company $i$ is more aggressive in their advertisement. According to the state and the action of the regulator, the company $i$ chooses their advertisement level $a_i(u^0, x^0)$ where $a_i(\cdot, \cdot)$ is a deterministic function. Similar to the SIS model, major dynamics are not influenced and given by $P^0(1|2, \dots)=P^0(2|1, \dots)=c\Delta t$ where $c>0$ is an exogenous constant.

Minor players' state space is $\mathcal{X}=\{1,2\}$ where $i,\ i=\{1,2\}$ denotes that they buy product $i$ and they choose one of the actions $\mathcal{U}=\{O, C\}$ where $O$ denotes that they are open to changes and $C$ denotes that they are close to changes.
\begin{equation*}
    P(i|i^{-1}, x^0, u^0, u) = [a_i(x^0, u^0)- a_{i^{-1}}(x^0, u^0)]\lambda^u \Delta t
\end{equation*}
where $\lambda^u$ is a coefficient that depends on the control of minor player with $\lambda^O>\lambda^C$.

The reward functions are given as
\begin{equation*}
    \begin{aligned}
        r(x, u, x^0, u^0, \mu) &= \sum_{x\in\mathcal{X}} \mathbf{1}_{\{x=i\}}\big[ c_\mu (\mu(i)-\mu(i^{-1}))+ c_a a_i(x^0, u^0)\big] - \sum_{u'\in \mathcal{U}}\mathbf{1}_{\{u=u'\}} c_{u'}, \\
        r^0(x^0, u^0, \mu) &= -c^0_m |\mu(1)-\mu(2)| + c^0_a \mathbf{1}_{\{u^0 \geq 1\}}.
    \end{aligned}
\end{equation*}

Concretely, as parameters we use $\Delta t = 0.3$, $c = 0.05$, $\mu_0 = \mathrm{Unif}$, $\mu^0_0(1) = 1$, $c_C = 0.75$, $c_O = 1.0$, $c_a = 1.0$, $c_\mu = 1.0$, $c^0_a = 0.1$, $c^0_m = 1$, $\lambda^U = 0.2$, $\lambda^O = 1.2$, we let $a_i(u^0, x^0) = k_0 + k^0_x \mathbf 1_{i}(x^0) + k^0_u \mathbf 1_{i}(u^0)$ for $k_0 = 0.2$, $k^0_x = 0.5$, $k^0_u = 0.7$, and consider a horizon of $T = 100$ for each episode.

\subsection{More finite horizon results}
Extending the qualitative results in the main text, in Figures~\ref{fig:policies-buf} and \ref{fig:policies-ad} we see plausible qualitative equilibrium behavior in the finite horizon case for the Buffet and Advertisement problem: In Buffet, players begin to move to the other location as the difference in fillings becomes sufficiently large, until the other location reaches a sufficiently high number of other players. This can be seen both in the visualization of policies, and in the example trajectory plot. Such behavior is plausible, as the instantaneous cost of moving from one location to another must be higher than the perceived future gain from being at the desired location, leading to the observed hysteresis effect. We can also observe the effect of a finite time horizon as $t \to T$, as a potential change in location before the buffet ends is not useful in terms of improving rewards. As expected, the learned policies are symmetric in the locations. Meanwhile, in Advertisement, players quickly run into an equilibrium that primarily depends on the exogeneous major player state. 

Further, as shown in Figure~\ref{fig:init}, the behavior of the FP algorithm is consistent regardless of the choice of initialization. This implies robustness against the initialization of the algorithm. And as shown in the main text for the learned policy, we also have for the maximum entropy uniform policy a convergence of objectives over both discretization and number of agents, see Figures~\ref{fig:J_discretization_maxent} and \ref{fig:J_num_agents_maxent} respectively. In particular, this maximum entropy policy trivially fulfills Lipschitz conditions as in Assumption~\ref{ass:m3picont}, and again verifies Theorems~\ref{thm:m3muconv} and \ref{thm:disc-opt}.

\begin{figure}
    \centering
    \includegraphics[width=0.99\linewidth]{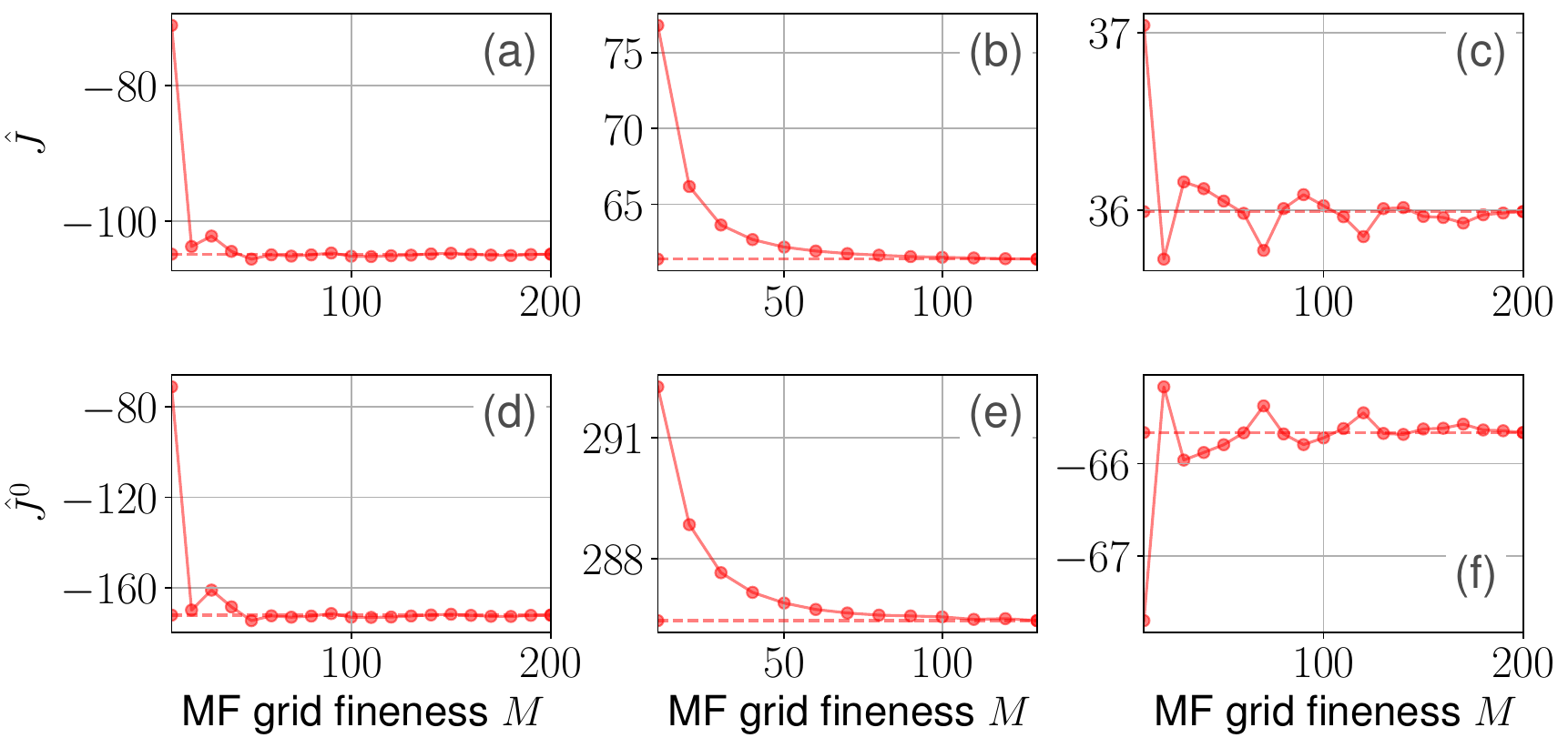}
    \caption{The approximate objectives of the uniform policy (dashed: right-most entry) quickly converge with finer discretization. (a-c): Minor exploitability, (d-f): major exploitability, (a, d): SIS, (b, e): Buffet, (c, f): Advertisement.}
    \label{fig:J_discretization_maxent}
\end{figure}

\begin{figure}
    \centering
    \includegraphics[width=0.99\linewidth]{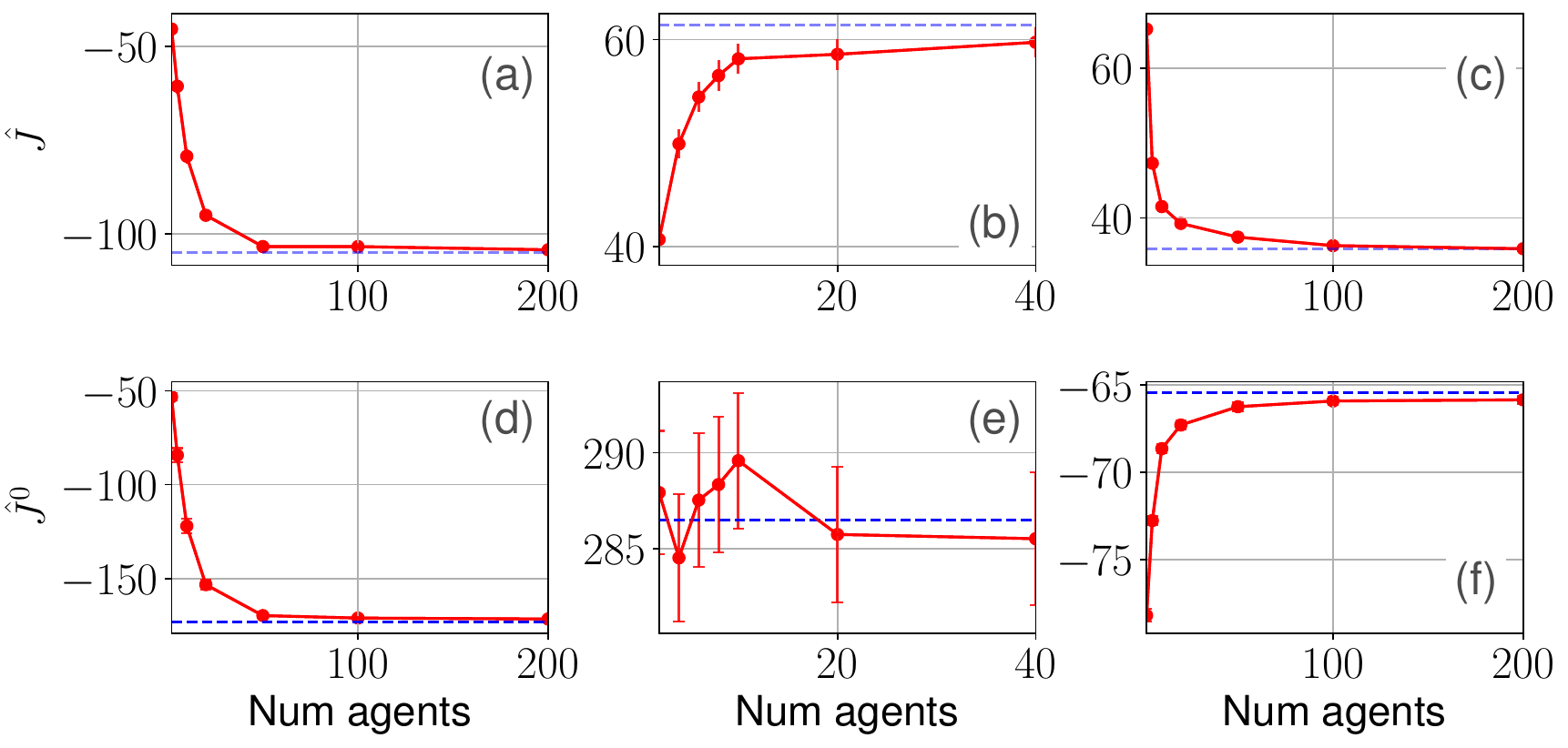}
    \caption{The mean $N$-player objective (red) over $1000$ (or $5000$ in Buffet) episodes, with $95\%$ confidence interval, against MF predictions $\hat J$, $\hat J^0$ of maximum entropy policy (blue, dashed). (a, d): SIS, (b, e): Buffet, (c, f): Advertisement.}
    \label{fig:J_num_agents_maxent}
\end{figure}

\begin{figure}
    \centering
    \includegraphics[width=0.7\linewidth]{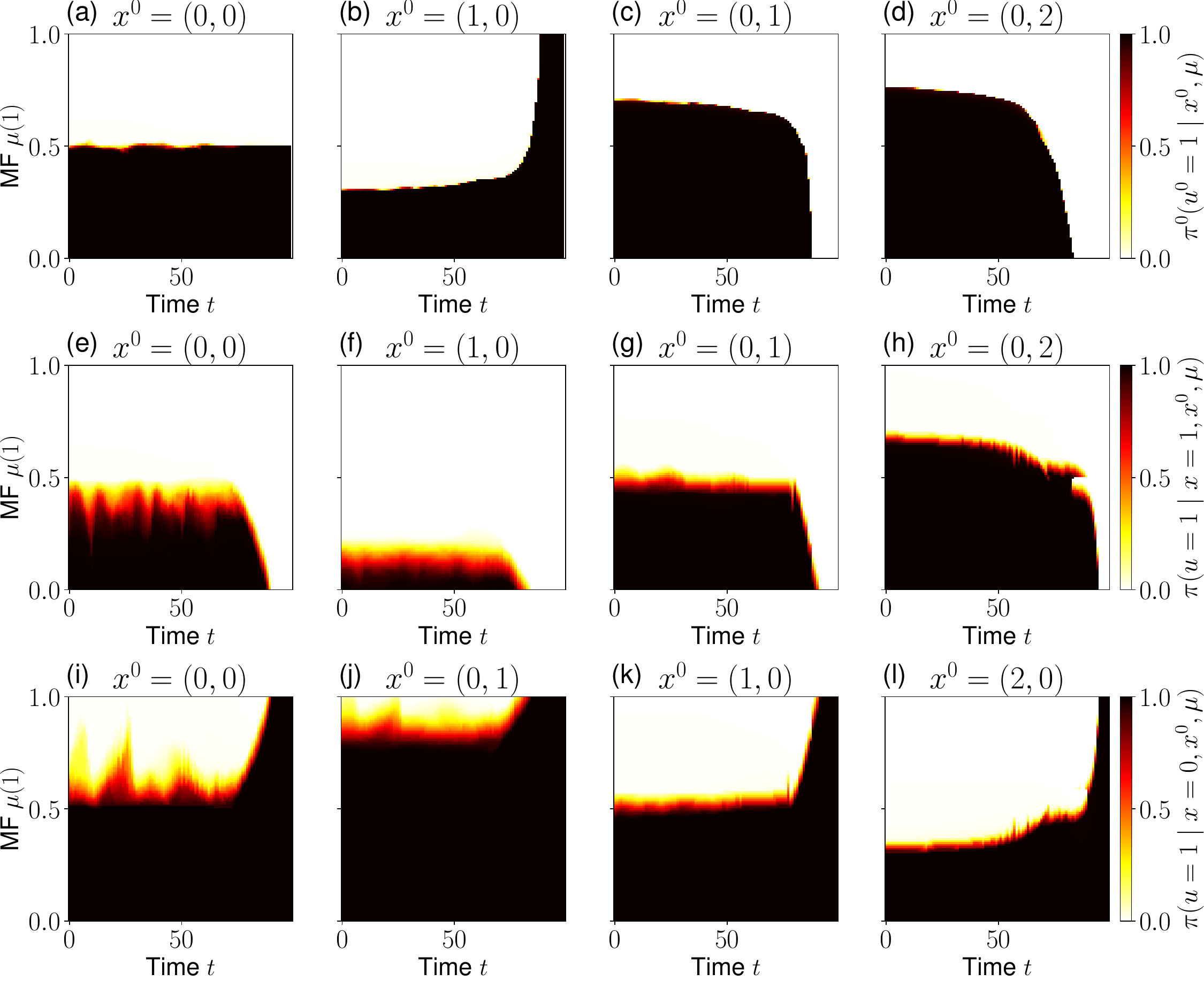}
    \includegraphics[width=0.7\linewidth]{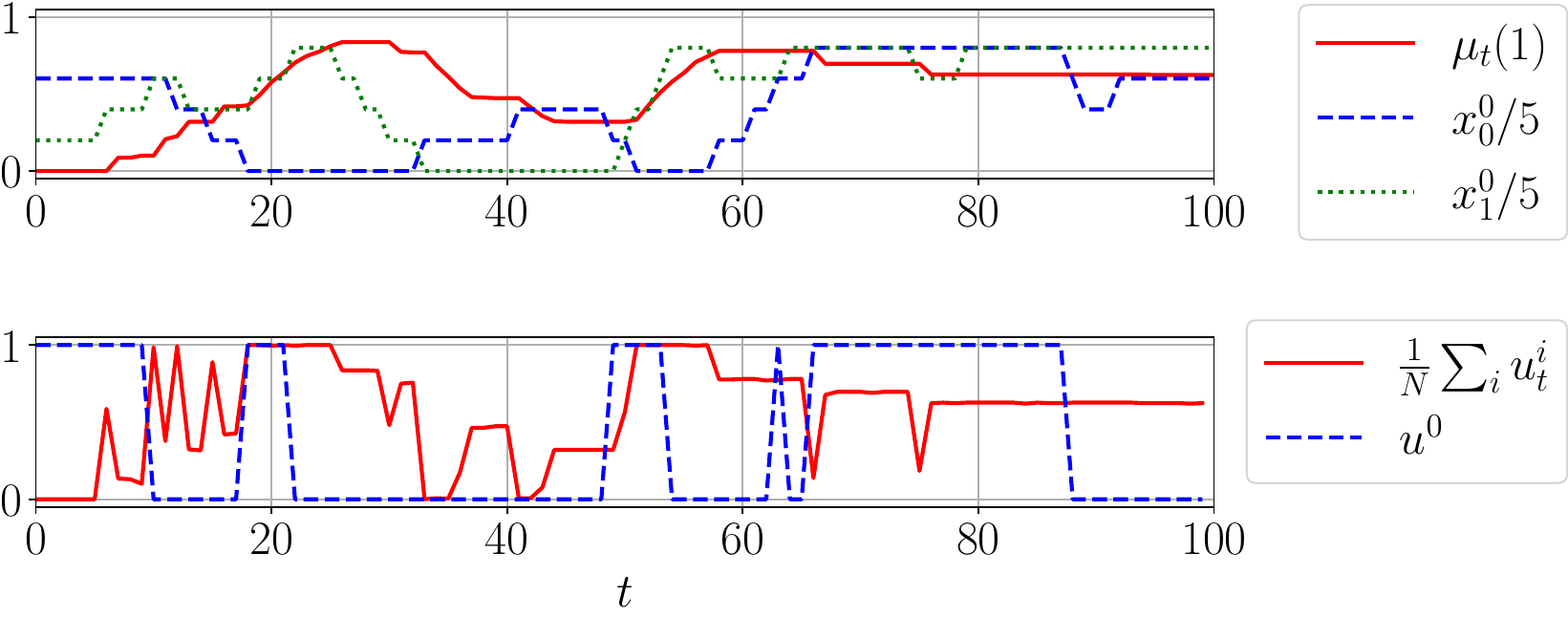}
    \caption{Qualitative behavior in the finite horizon case for the Buffet problem.}
    \label{fig:policies-buf}
\end{figure}

\begin{figure}
    \centering
    \includegraphics[width=0.7\linewidth]{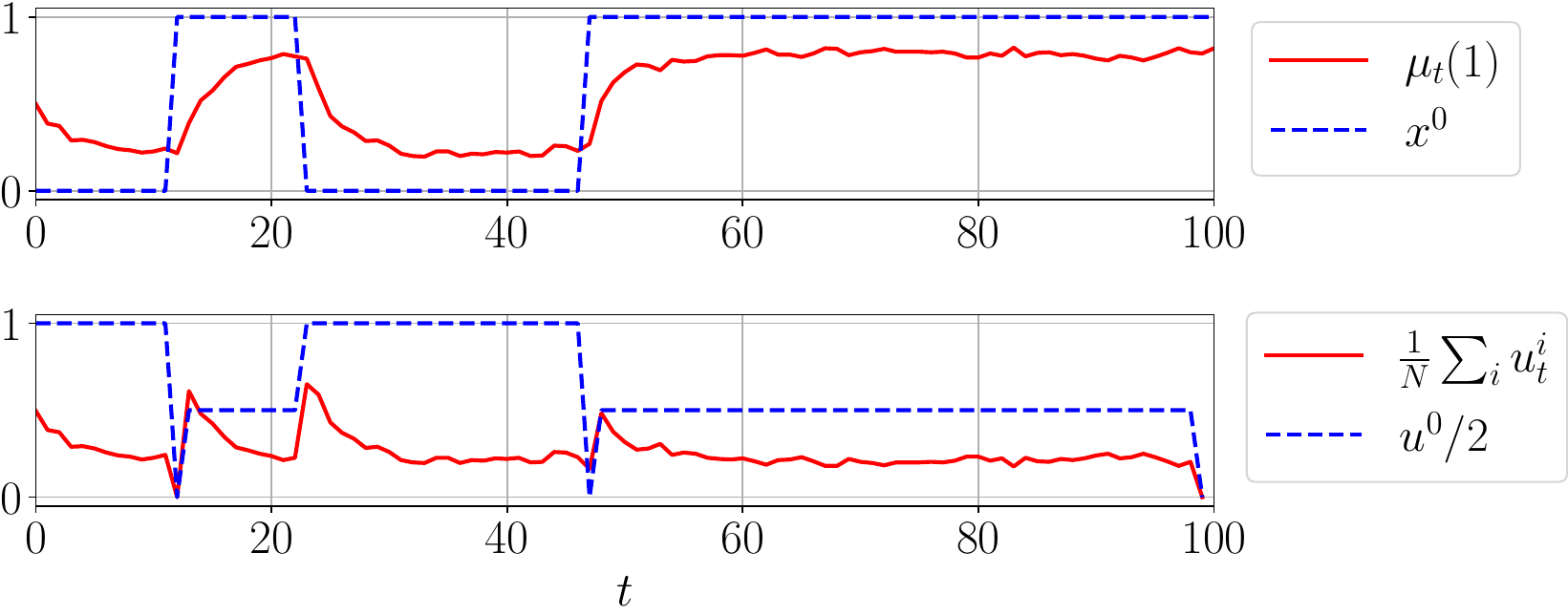}
    \caption{Qualitative behavior in the finite horizon case for the Advertisement problem.}
    \label{fig:policies-ad}
\end{figure}

\begin{figure}
    \centering
    \includegraphics[width=0.7\linewidth]{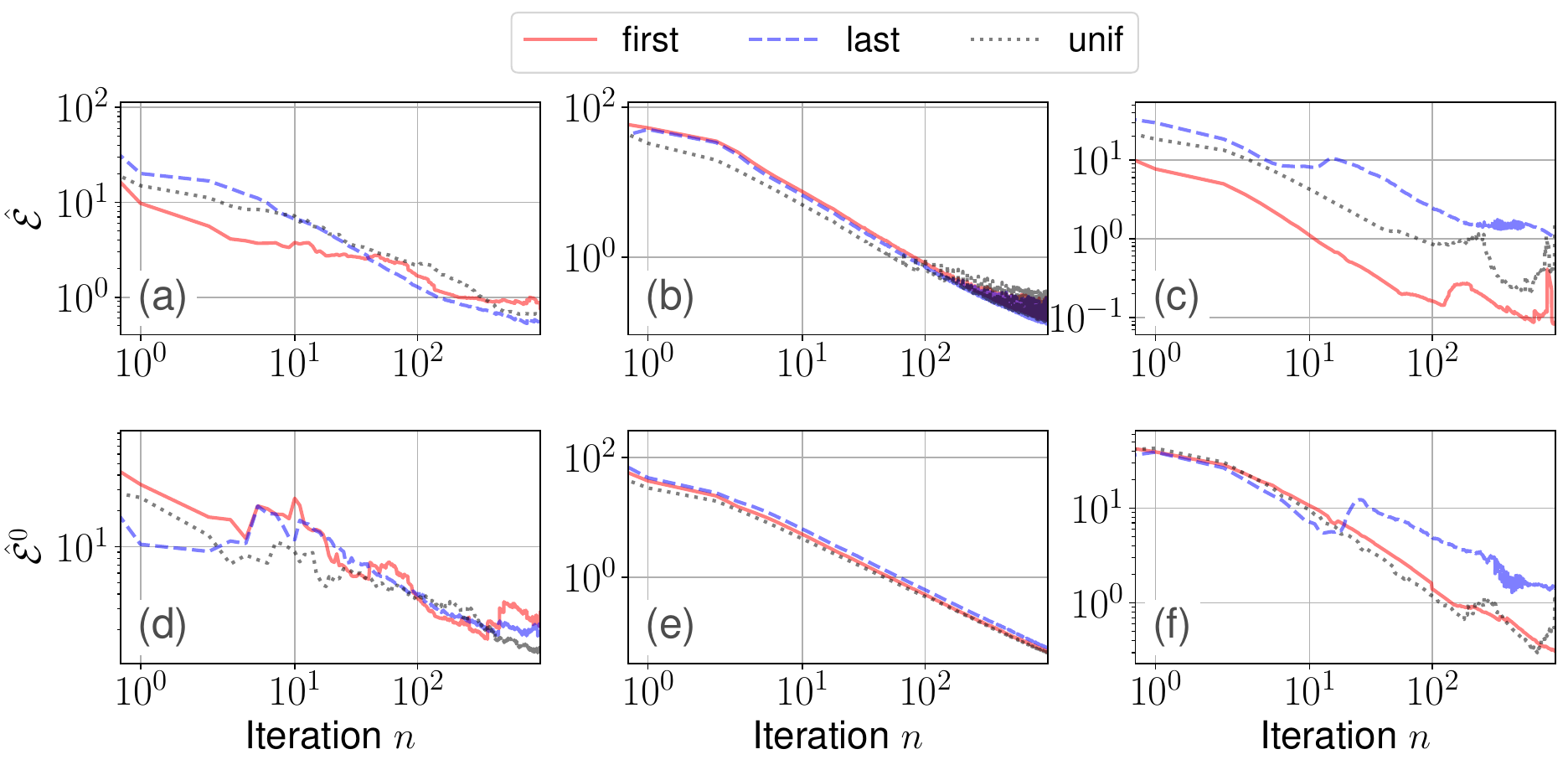}
    \caption{The training curve of FP for various initializations. \textit{first}: initial policy assigns all mass on the first action; \textit{last}: all mass on the last action; \textit{unif}: the uniform maximum entropy policy. Here, actions are ordered as they appear in the problem description. (a-c): Minor exploitability, (d-f): major exploitability, (a, d): SIS, (b, e): Buffet, (c, f): Advertisement.}
    \label{fig:init}
\end{figure}

\subsection{Infinite-horizon discounted results}
As discussed in the main text, we can extend our algorithm to the infinite-horizon discounted objective case. We observe similar results and behavior as for the finite-horizon. For the infinite-horizon case, we apply value iteration to compute best responses, stopping value iteration when the maximum TD error over all states is less than $10^{-5}$.

In particular, in Figure~\ref{fig:inf_exploitability-fpi}, we observe the usual non-convergence of FPI, whereas the FP algorithm together with value iteration converges in terms of exploitability. Only sometimes does the fixed-point iteration converge (here in SIS for $M=80$), which again motivates the formulation of a FP algorithm. In the second part of the figure, we verify our empirical contribution, i.e. the FP algorithm, which generalizes also to the infinite-horizon discounted objectives.

In Figure~\ref{fig:inf_J_discretization_maxent},we can see the convergence of objectives over discretization, both for the maximum entropy policy and the FP-learned policy, as well as the stability of the FP algorithm over discretization. The qualitative behavior is similar as the one seen in the main text for the finite horizon case. Further, in Figure~\ref{fig:inf_J_num_agents_maxent}, the convergence of objectives and therefore the propagation of chaos over an increasing number of players is again supported, both for the maximum entropy policy and the FP-learned policy.

Lastly, in Figure~\ref{fig:inf_qual}, the qualitative behavior of SIS is shown and is comparable to the behavior in Figure~\ref{fig:qual}, except for the absence of a transient finite-horizon effect near the end of the problem, due to the stationarity of the optimal policy under the discounted infinite-horizon objective. 

\begin{figure}
    \centering
    \hfill{}
    \includegraphics[width=0.47\linewidth]{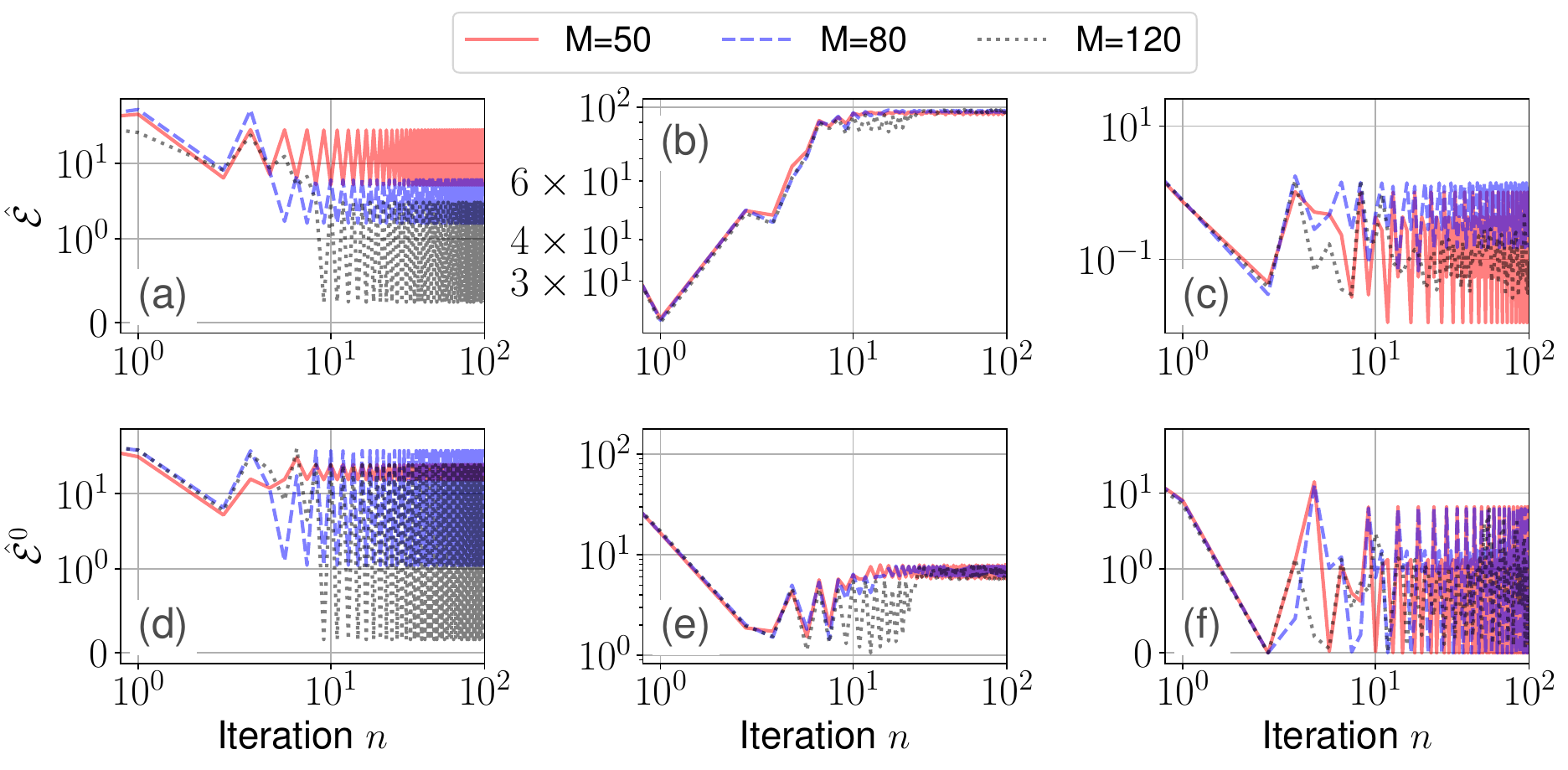}
    \hfill{}
    \includegraphics[width=0.47\linewidth]{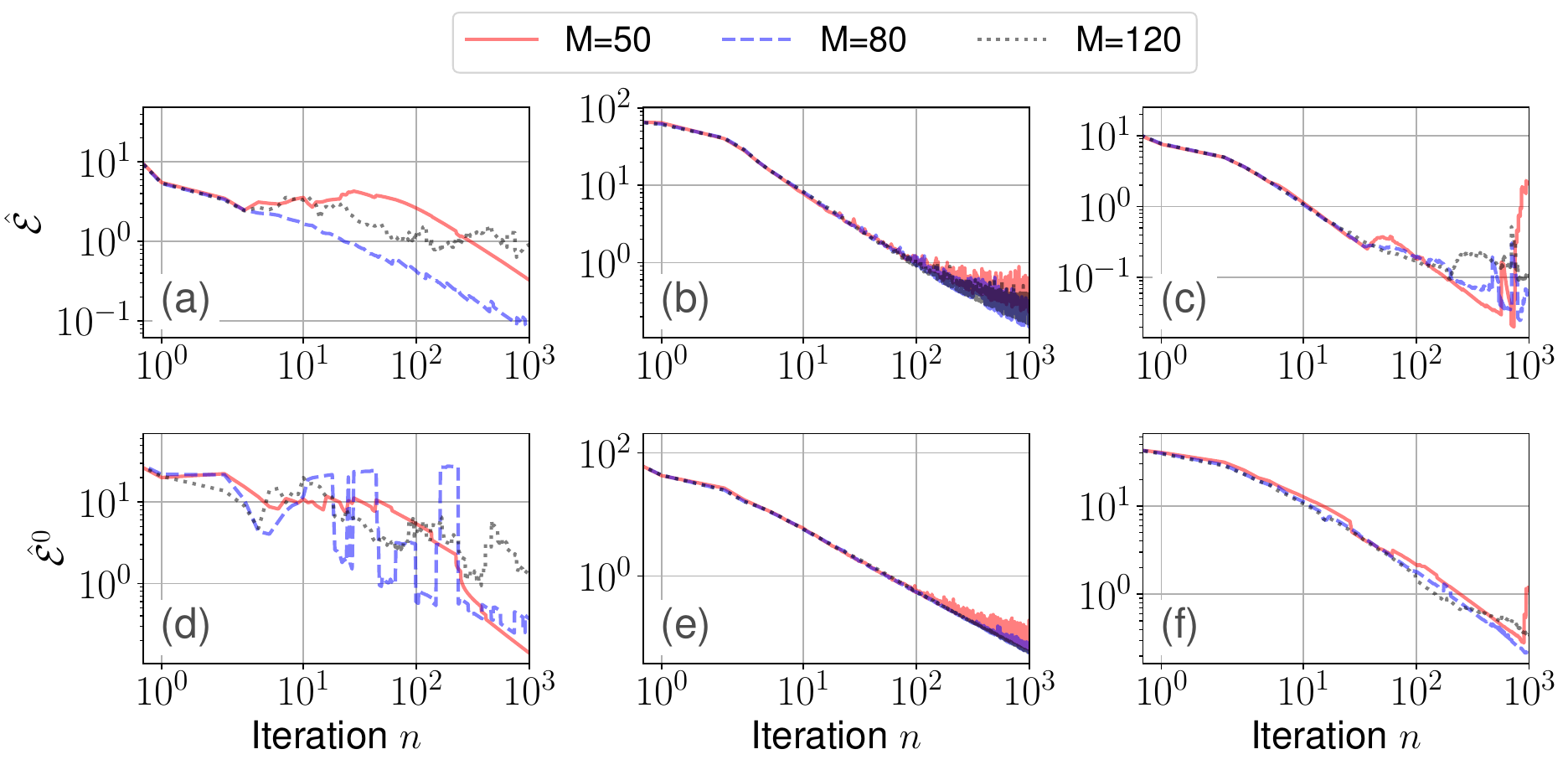}
    \hfill{}
    \caption{Exploitability over iterations of infinite-horizon FPI (left) and FP (right) can run into a limit cycle. (c, f: Advertisement), (a, d: SIS), (b, e: Buffet). (a-c): Minor exploitability, (d-f): major exploitability.}
    \label{fig:inf_exploitability-fpi}
\end{figure}

\begin{figure}
    \centering
    \hfill{}
    \includegraphics[width=0.47\linewidth]{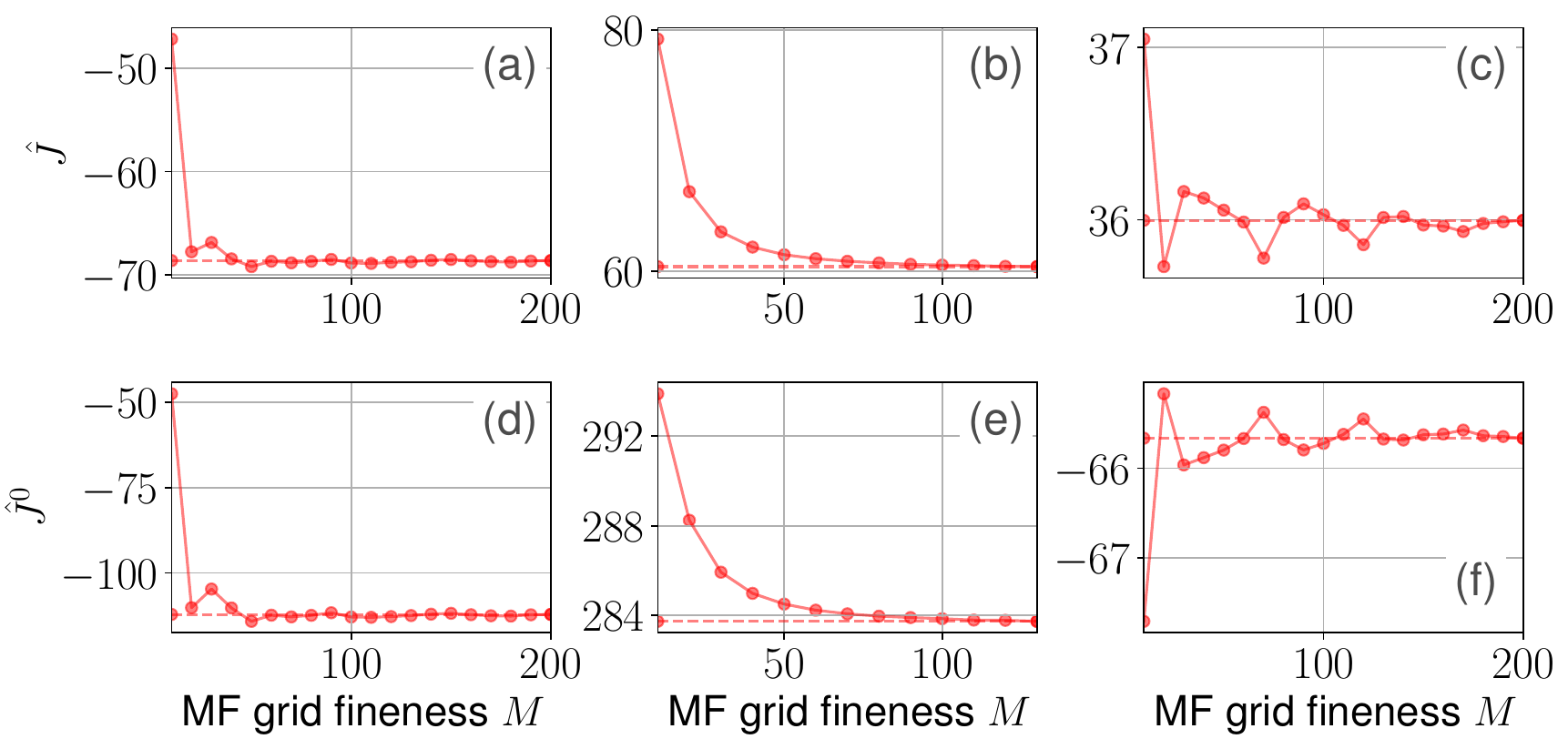}
    \hfill{}
    \includegraphics[width=0.47\linewidth]{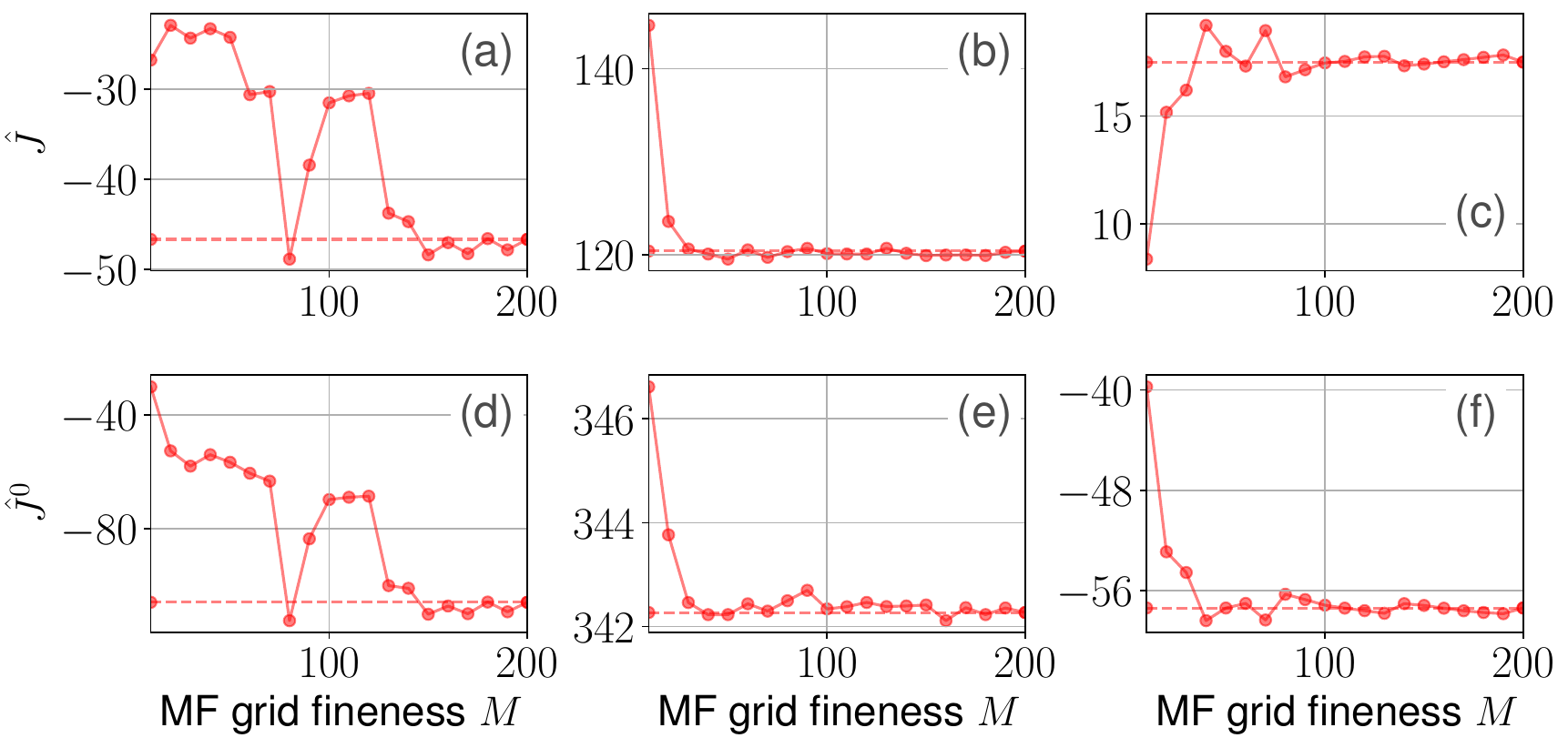}
    \hfill{}
    \caption{The infinite-horizon objective $J$, $J^0$ of the maximum entropy policy (left) and FP-learned policy (right) under discretization (dashed: right-most entry). The objectives quickly converge with increasing discretization fineness. (a-c): Minor exploitability, (d-f): major exploitability, (a, d): SIS, (b, e): Buffet, (c, f): Advertisement.}
    \label{fig:inf_J_discretization_maxent}
\end{figure}

\begin{figure}
    \centering
    \hfill{}
    \includegraphics[width=0.47\linewidth]{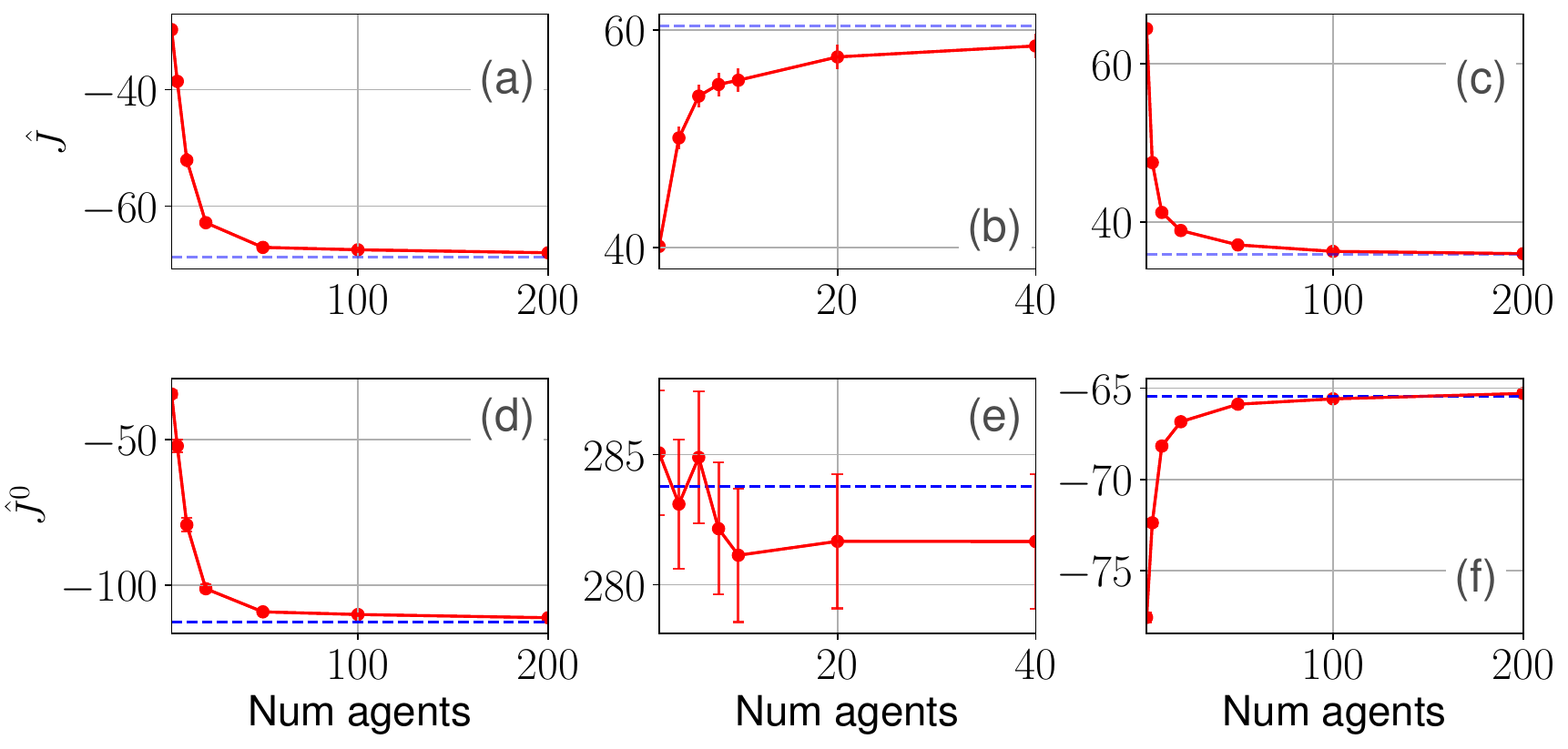}
    \hfill{}
    \includegraphics[width=0.47\linewidth]{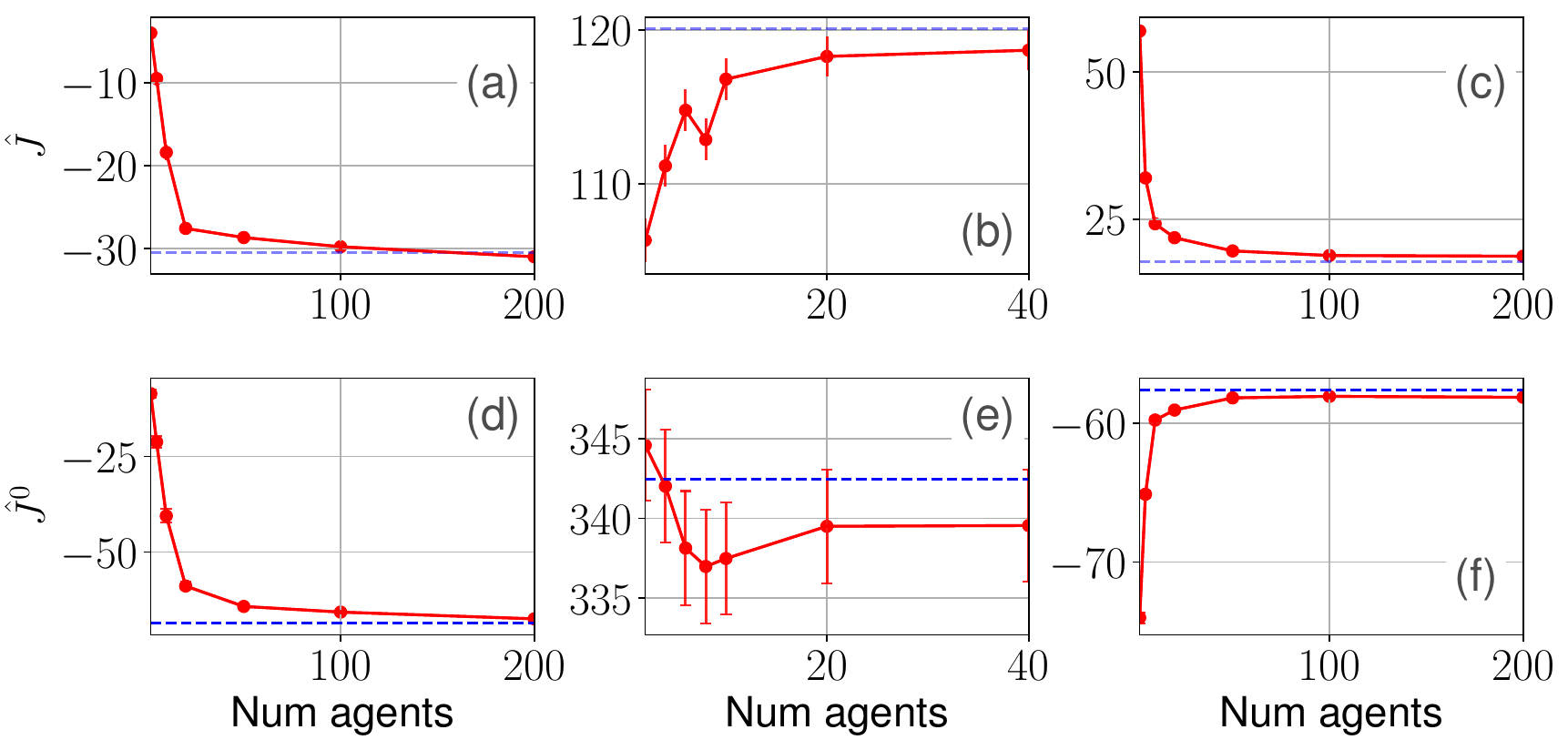}
    \hfill{}
    \caption{The mean $N$-player infinite-horizon objective (red) over $1000$ episodes with $95\%$ confidence interval, compared against approximate objectives $\hat J$, $\hat J^0$ of the maximum entropy policy (left) or FP-learned policy (right) as dashed blue line. (a): SIS, (b): Buffet, (c): Advertisement.}
    \label{fig:inf_J_num_agents_maxent}
\end{figure}

\begin{figure}
    \centering
    \includegraphics[width=0.99\linewidth]{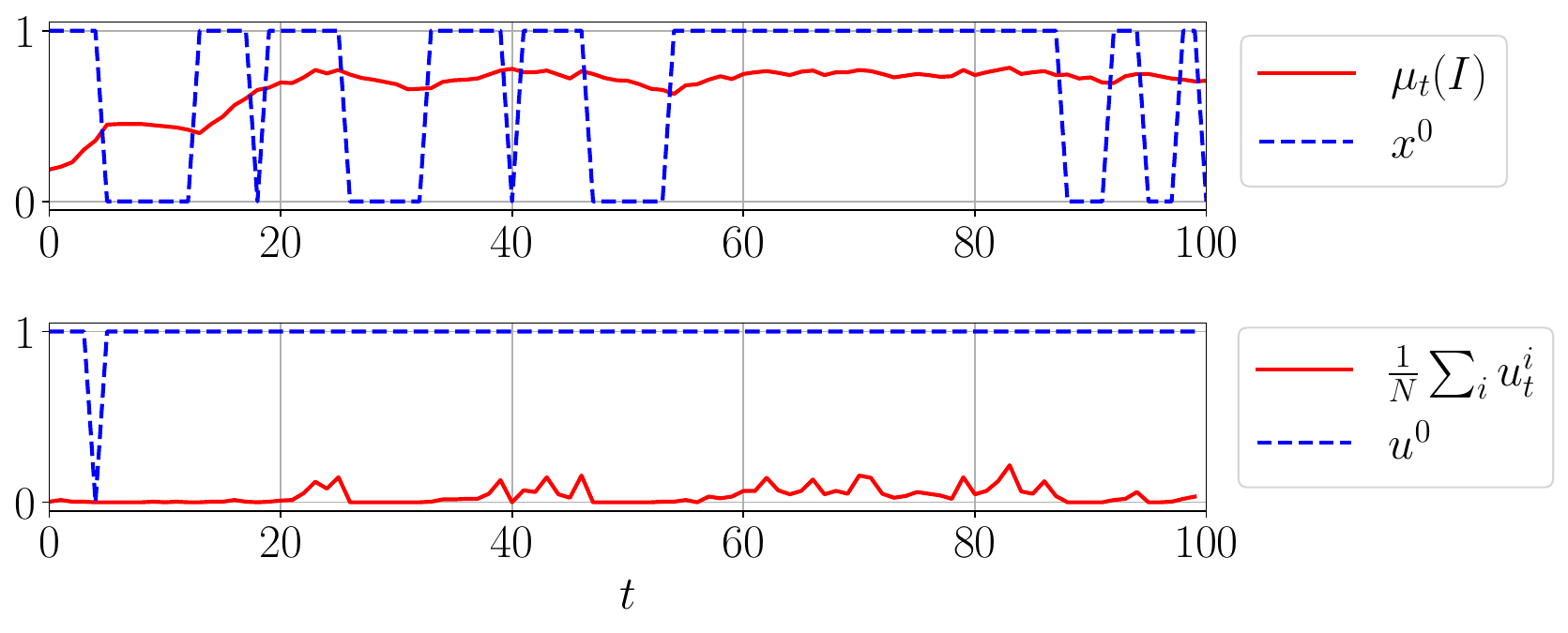}
    \caption{The resulting equilibrium behavior in infinite-horizon SIS is comparable to the finite-horizon case in Figure~\ref{fig:qual}, but without finite horizon effects at the end of the episode.}
    \label{fig:inf_qual}
\end{figure}

\end{document}